\documentclass[12pt]{article}
\usepackage{amssymb}
\usepackage{tikz}

\usetikzlibrary{arrows, calc, decorations.pathmorphing, backgrounds, positioning, fit, petri, automata}
\definecolor{grey1}{rgb}{0.5,0.5,0.5}
\usepackage[top=2.54cm, bottom=2.54cm, left=2.2cm, right=2.2cm]{geometry}
\usepackage{algorithmicx,algcompatible,algorithm,algpseudocode}
\newcommand{\INDSTATE}[1][1]{\STATE\hspace{\algorithmicindent}}

\usepackage{changebar}

\usepackage{hyperref}
\hypersetup{
  colorlinks,
  citecolor=blue,
  linkcolor=red,
  urlcolor=magenta}

\usepackage{mathtools}
\usepackage{graphicx}
\usepackage{subfigure}
\usepackage{amsmath,amsthm,bm}
\usepackage{rotating}
\usepackage{mathrsfs}
\usepackage{chngcntr}
\usepackage{apptools}
\usepackage[titletoc,title]{appendix}
\usepackage{comment}

\usepackage{xcolor}         
\definecolor{grau}{rgb}{0.8,0.8,0.8}

\newcommand{\chen}[1]{\color{orange}}
\usepackage{setspace,lscape,longtable}
\usepackage{amsmath,amsthm,bm}
\usepackage{fullpage}  

\usepackage{titlesec}
\usepackage[english]{babel}
\usepackage{xcolor}         
\definecolor{grau}{rgb}{0.8,0.8,0.8}

\newcommand{\cut}{} 

\theoremstyle{plain}
\newtheorem{proposition}{Proposition}
\newtheorem{theorem}{Theorem}
\newtheorem{lemma}{Lemma}
\newtheorem{corollary}{Corollary}

\theoremstyle{definition}

\newtheorem*{example}{Example}
\newtheorem*{algo}{Algorithm}

\theoremstyle{remark}
\newtheorem{remark}{Remark}

\usepackage[comma,authoryear]{natbib}

\bibpunct{(}{)}{;}{a}{,}{,}
\numberwithin{equation}{section}

\newcommand{\iid}{{\mathrm{i.i.d.}}}
\newcommand{\iidsim}{{\overset{\mathrm{i.i.d.}}{\sim}}}
\newcommand{\transpose}{^{\mathrm{T}}}
\newcommand{\lsigma}{{\underline{\sigma}}}
\newcommand{\usigma}{{\overline{\sigma}}}
\newcommand{\leps}{{\underline{\epsilon}}}
\newcommand{\ueps}{{\overline{\epsilon}}}

\newcommand{\bLambda}{{\bm{\Lambda}}}

\newcommand{\calT}{{\mathcal{T}}}
\newcommand{\calC}{{\mathcal{C}}}
\newcommand{\calS}{{\mathcal{S}}}
\newcommand{\calN}{{\mathcal{N}}}
\newcommand{\bolda}{{\bm{a}}}

\newcommand{\bU}{{\bm{U}}}
\newcommand{\bu}{{\mathbf{u}}}

\newcommand{\bw}{{\bm{w}}}
\newcommand{\by}{{\bm{y}}}
\newcommand{\bz}{{\bm{z}}}

\newcommand{\mb}{{\bm{m}}}
\newcommand{\bV}{{\bm{V}}}

\newcommand{\bgamma}{{\bm{\gamma}}}
\newcommand{\bGamma}{{\bm{\Gamma}}}

\newcommand{\bSigma}{{\bm{\Sigma}}}

\newcommand{\eye}{{\bm{I}}}
\newcommand{\one}{{\mathbf{1}}}

\newcommand{\btheta}{{\bm{\theta}}}
\newcommand{\bmu}{{\bm{\mu}}}

\newcommand{\zero}{{\bm{0}}}

\newcommand{\eps}{\epsilon}
\newcommand{\calM}{{\mathcal{M}}}

\newcommand{\keywords}[1]{\par\addvspace\baselineskip\noindent\enspace\ignorespaces \textbf{Key Words: }#1}

\makeatletter
\g@addto@macro\normalsize{%
	\setlength\abovedisplayskip{1ex}
	\setlength\belowdisplayskip{1ex}
	\setlength\abovedisplayshortskip{1ex}
	\setlength\belowdisplayshortskip{1ex}
}
\makeatother

\usepackage{lipsum}

\usepackage{enumitem}
\setlist[enumerate]{itemsep=0mm}
\setlist[itemize]{itemsep=0mm}

 \author{Fangzheng Xie\thanks{Department of Applied Mathematics and Statistics, Johns Hopkins University, Baltimore, MD, 21218} \and Yanxun Xu\footnotemark[1] \thanks{Correspondence should be addressed to Yanxun Xu (yanxun.xu@jhu.edu)}}
\title{{Bayesian Repulsive Gaussian Mixture Model}}
\date{}
\begin{document}
\maketitle
\doublespacing

\begin{abstract}
	We develop a general class of Bayesian repulsive Gaussian mixture models that encourage well-separated clusters, aiming at reducing potentially redundant components produced by independent priors for locations (such as the Dirichlet process). The asymptotic results for the posterior distribution of the proposed models are derived, including posterior consistency and posterior contraction rate in the context of nonparametric density estimation. More importantly, we show that
	compared to the independent prior on the component centers, the repulsive prior introduces additional shrinkage effect on the tail probability of the posterior number of components, which serves as a measurement of the model complexity.
	In addition, an efficient and easy-to-implement blocked-collapsed Gibbs sampler is developed based on the exchangeable partition distribution and the corresponding urn model. We evaluate the performance and demonstrate the advantages of the proposed model through extensive simulation studies and real data analysis. The R code is available at \url{https://drive.google.com/open?id=0B_zFse0eqxBHZnF5cEhsUFk0cVE}. 
\end{abstract}
\section{Introduction} 
\label{sec:introduction}
In Bayesian analysis of mixture models, independent priors on the component-specific parameters have been widely used because of their flexibility and technical convenience.  
A nonparametric example is the renowned Dirichlet process (DP) 
where the atoms in the stick-breaking representation 
are independent and identically distributed (i.i.d.) from a base distribution. One of the potential but non-negligible issues for such an approach is the presence of redundant components, especially when parsimony on the number of components is preferred. 
For example, when a mixture model is used in biomedical applications, each component of the mixture may be interpreted as clinically or biologically meaningful subpopulations (of patients, disease types, etc.). 
To address this challenge, in this paper we argue for a Bayesian approach for modeling repulsive mixtures as a competitive alternative, establish its posterior consistency and posterior contraction rate, and study the shrinkage effect on the posterior number of components in the presence of such a repulsion. 

Mixture models have been extensively studied from both the frequentist and the Bayesian perspectives. Formally, given the parameter space $\Theta$, a mixture model with a kernel density $\psi:\mathbb{R}^p\times\Theta\to\mathbb{R}_+$ and a mixing distribution $G\in\mathcal{M}(\Theta)$ can be represented as
$\by_i\sim\int_\Theta \psi(\by,\btheta)\mathrm{d}G(\btheta)$, where $\mathcal{M}(\Theta)$ is a class of probability distributions on $\Theta$ (equipped with an implicitly specified suitable $\sigma$-field). 
The most commonly used kernel density $\psi$ is the normal density, which leads to the Gaussian mixture model (GMM). In particular, the GMM with a discrete (potentially infinitely supported) mixing $G=\sum_kw_k\delta_{(\bmu_k,\bSigma_k)}$ has been widely used for clustering, since an equivalent characterization is $
\by_i\mid z_i\sim\mathrm{N}(\bmu_{z_i},\bSigma_{z_i})$, $\mathbb{P}(z_i = k)=w_k$, where  $z_i$ encodes the clustering membership of the corresponding observation $\by_i$. The parameters for each component $(\bmu_k,\bSigma_k)$, $k=1,\cdots,K$, are referred to as the cluster/component-specific parameters. Throughout we use $K$ to denote the (potentially infinite) number of components in a mixture model. 
When $G$ is completely unknown, the GMM is referred to as nonparametric GMM \citep{chen2017consistency}. 
Frequentists' ways of modeling mixture models require a finite and fixed $K$, the estimation of which could  be accomplished using model selection approaches. Nonparametric Bayesian priors allow us to perform inference without {\it a priori} fixed and finite $K$. For example, the DP prior on $G$ yields an exchangeable partition distribution on $(\btheta_{z_1},\cdots,\btheta_{z_n})$, the inference of which indicates a distribution on the number of clusters among $(\btheta_{z_1},\cdots,\btheta_{z_n})$. 
The development of Markov chain Monte Carlo sampling techniques   \citep{ishwaran2001gibbs,ishwaran2002approximate,antoniak1974mixtures,maceachern1998estimating,neal2000markov,walker2007sampling} further popularized the DP mixture model in a wide array of applications, such as biomedicine, machine learning, pattern recognition, etc. 

Meanwhile, 
the asymptotic results of the DP mixture of Gaussians as a method of nonparametric density estimation have been studied. In the univariate case, the posterior consistency of the DP mixture of univariate Gaussians was established by \cite{ghosal1999posterior}, and the posterior convergence rate in the context of density estimation in nonparametric Gaussian mixture model was studied by \cite{ghosal2001entropies}. 
Posterior consistency in the multivariate setting \citep{wu2010l1} is harder due to the exponential growth of the $L_1$-entropy of sieves. \cite{shen2013adaptive,canale2017posterior} derived the posterior contraction rates of general smooth densities for multivariate density estimation using the DP mixture of Gaussians. 

Nevertheless, as shown in \cite{xu2016bayesian}, the DP mixture model typically produces relatively large number of clusters, some of which are typically redundant. Theoretically, \cite{miller2013simple} showed that when the underlying data generating density is a finite mixture of Gaussians, the posterior number of clusters under the DP mixture model is not consistent. In other words, the posterior distribution of the number of clusters does not converge to the point mass at the underlying true $K$. Alternatively, finite mixture models with a prior on $K$, referred to as the mixture of finite mixtures (MFM) \citep{nobile1994bayesianmixture,miller2016mixture}, was developed. The posterior inference of the MFM can be carried out either by the reversible-jump Markov chain Monte Carlo (RJ-MCMC) \citep{green1995reversible}, or by the collapsed Gibbs sampler derived via the exchangeable partition representation \citep{miller2016mixture}. Meanwhile, the posterior asymptotics for the MFM as a nonparametric density estimator, to the our best knowledge, is restricted to the cases of univariate location-scale mixtures \citep{kruijer2010adaptive} and multivariate location mixtures \citep{shen2013adaptive}, in which the priors on locations are assumed to be conditionally i.i.d. given $K$.

These approaches, however, assume independent prior on the component-specific parameters $(\btheta_1,\cdots,\btheta_K)$. In the context of parametric inference, where the underlying data generating distribution is a finite mixture of Gaussians, repulsive priors \citep{petralia2012repulsive, quinlan2017parsimonious} and non-local priors \citep{fuquene2016choosing} were developed as shrinkage methods to penalize mixture models with redundant components. In particular, theoretical properties regarding only univariate density estimations in parametric GMM (\emph{i.e.}, assuming the ground true density is a finite mixture of Gaussians) were discussed in \cite{petralia2012repulsive} and \cite{quinlan2017parsimonious}. 
In addition, \cite{xu2016bayesian} proposed repulsive mixtures via determinantal point process (DPP)  with a prior on $K$, where the RJ-MCMC sampler for the posterior inference is potentially inefficient in high-dimensional setting. 

In this paper, we propose a Bayesian repulsive Gaussian mixture (RGM) model. The main contributions of this paper are as follows. First, under certain mild regularity conditions, we establish the posterior consistency for density estimation in nonparametric GMM under the RGM prior, and obtain an ``almost'' parametric posterior contraction rate $(\log n)^{t}/\sqrt{n}$ for $t>p+1$.
To the best of our knowledge, earlier work such as \cite{ghosal2001entropies}, \cite{petralia2012repulsive}, and \cite{quinlan2017parsimonious}, have not addressed the asymptotic analysis of repulsive mixture models for density estimation in nonparametric GMM.  \cite{ghosal2001entropies} was the earliest work that discussed the posterior contraction rate for density estimation in nonparametric GMM, where the Dirichlet process (DP) prior is used. 
\cite{petralia2012repulsive} and \cite{quinlan2017parsimonious} discussed the posterior contraction rate using repulsive priors, but under the parametric assumption that the mixing distribution is finitely discrete. 
Second, the relationship between the posterior of $K$ (\emph{i.e.}, the number of components), which serves as a measurement of the model complexity, and the repulsive prior is studied as well. It turns out that compared to the independent prior on the component centers, the repulsive prior introduces additional shrinkage effect on the tail probability of the posterior of $K$ under the nonparametric GMM assumption. 
Furthermore, instead of fixing $K$ or implementing a RJ-MCMC sampler for the posterior inference of the RGM model, we 
develop a more efficient blocked-collapsed Gibbs sampler that is based on the exchangeable partition distributions.

The remainder of the paper is organized as follows. In Section \ref{sec:preliminaries} we formulate the Bayesian repulsive Gaussian mixture model. 
Section \ref{sec:posterior_convergence} elaborates the theoretical properties of the posterior distribution. In particular, we establish the posterior consistency, investigate posterior contraction rate, and study the shrinkage effect on the posterior number of components in the presence of the repulsive prior. In Section \ref{sec:posterior_inference} we develop the generalized urn model for the RGM model by integrating out the mixing weights and $K$, and design an efficient blocked-collapsed Gibbs sampler. Section \ref{sec:numerical_results} demonstrates the advantages of the proposed model as well as the efficiency of the proposed inference algorithm via simulation studies and real data analysis. 
We conclude the paper in Section \ref{sec:conclusion}. 

\section{Bayesian Repulsive Mixture Model} 
\label{sec:preliminaries}
In this section we formulate the RGM model in a Bayesian framework. Suppose $\mathcal{S}\subset\mathbb{R}^{p\times p}$ is a collection of positive definite matrices, 
\cut 
equipped with the Borel $\sigma$-field on $\mathcal{S}$. We consider the Gaussian mixture model, a family of densities of the form
\begin{eqnarray}
f_F(\by)=\int_{\mathbb{R}^p\times\calS}\phi(\by\mid\bmu,\bSigma)\mathrm{d}F(\bmu,\bSigma),
\end{eqnarray} 
where 
$
\phi(\by\mid\bmu,\bSigma)=\det(2\pi\bSigma)^{-\frac{1}{2}}\exp\left[-\frac{1}{2}(\by-\bmu)^\top\bSigma^{-1}(\by-\bmu)\right]$
is the density of the $p$-dimensional Gaussian distribution $\mathrm{N}(\bmu,\bSigma)$ with mean $\bmu$ and covariance matrix $\bSigma$, and $F$ is a distribution on $\mathbb{R}^p\times\calS$. 
We shall also use the shorthand notation $\phi_{\bSigma}(\by-\bmu)=\phi(\by\mid\bmu,\bSigma)$ and $f_F=\phi_{\bSigma}* F$, where $*$ is the conventional notation for convolution of two functions. We assume that the data 
$(\by_n)_{n=1}^{\infty}$ are i.i.d. generated from some unknown density $f_0$, the estimation of which is of interest. 

Denote the space of all probability distributions over $\mathbb{R}^p\times\calS$ by $\mathcal{M}(\mathbb{R}^p\times\calS)$, and that over $\mathbb{R}^p$ by $\mathcal{M}(\mathbb{R}^p)$. We define a prior $\Pi$ on $f$ over the space of all density functions in $\mathbb{R}^p$
by the following hierarchical model:  
\begin{eqnarray}
(f(\by)\mid F)&=&\int_{\mathbb{R}^p\times\calS}\phi(\by\mid\bmu,\bSigma)\mathrm{d}F(\bmu,\bSigma),\nonumber\\
\label{eqns:MFM}
\left(F\mid K, \{w_k, \bmu_k, \bSigma_k\}_{k=1}^K\right)&=&\sum_{k=1}^Kw_k\delta_{(\bmu_k,\bSigma_k)},\\
(\bmu_1,\bSigma_1,\cdots,\bmu_K,\bSigma_K\mid K)&\sim&p(\bmu_1,\bSigma_1,\cdots,\bmu_K,\bSigma_K\mid K),\nonumber\\
(w_1,\cdots,w_K\mid K)\sim\mathcal{D}_K(\beta), && K\sim p_K(K),\quad K\in \mathbb{N}_+. \nonumber
\end{eqnarray}
Here $p(\bmu_1,\bSigma_1,\cdots,\bmu_K,\bSigma_K\mid K)>0$ is some density function with respect to the Lebesgue measure on $(\mathbb{R}^p\times\calS)^K$,  $\mathcal{D}_K(\beta)$ is the symmetric Dirichlet distribution over $\Delta^{K}$ with density function $p(w_1,\cdots,w_K)=\Gamma(K\beta)/\Gamma(\beta)^K\prod_{k=1}^Kw_k^{\beta-1}$, where $\Delta^K=\{(w_1,\cdots,w_K)\transpose:\sum_{k=1}^Kw_k=1,w_k\geq0\}$ is the $\ell_1$-simplex on $\mathbb{R}^K$. 
The prior on $K$ that is supported on all positive integers is essential, as we allow the number of components to grow with the sample size in order to fit the data well. 

Instead of assuming $(\bmu_k,\bSigma_k)_{k=1}^K$ being $\iid$ from a ``base measure'', we introduce repulsion among components $\mathrm{N}(\bmu_k,\bSigma_k)$ through their centers $\bmu_k$, such that they are well separated.  
We assume the density $p(\bmu_1,\bSigma_1,\cdots,\bmu_K,\bSigma_K\mid K)$ is of the following form,
\begin{eqnarray}\label{eqn:repulsive_density}
p(\bmu_1,\bSigma_1,\cdots,\bmu_K,\bSigma_K\mid K)&=&\frac{1}{Z_K}\left[\prod_{k=1}^Kp_{\bmu}(\bmu_k)p_{\bSigma}(\bSigma_k)\right]h_K(\bmu_1,\cdots,\bmu_K),
\end{eqnarray}
where 
$Z_K=\int\cdots\int_{\mathbb{R}^{p\times K}} h_K(\bmu_1,\cdots,\bmu_K)\left[\prod_{k=1}^Kp(\bmu_k)\right]\mathrm{d}\bmu_1\cdots \mathrm{d}\bmu_K$ 
is the normalizing constant, and the function $h_K:(\mathbb{R}^p)^K\to[0,1]$ is invariant under permutation of its arguments: $h_K(\bmu_1,\cdots,\bmu_K)=h_K(\bmu_{\mathfrak{T}(1)},\cdots,\bmu_{\mathfrak{T}(K)})$ for any permutation $\mathfrak{T}:\{1,\cdots,K\}\to \{1,\cdots,K\}$. We require that $h_K$ satisfies the following repulsive condition: $h_K(\bmu_1,\cdots,\bmu_K)=0$ if and only if $\bmu_k=\bmu_{k'}$ for some $k\neq k'$, $k,k'\in\{1,\cdots,K\}$. 
In this paper, we focus on the case where the repulsive property is introduced only through the mean vectors $(\bmu_1,\cdots,\bmu_K)$, \emph{i.e.}, we allow nonvanishing density even when distinct components share an identical covariance matrix. The case where repulsion is introduced through the covariance matrices is of independent interest and may be further explored. 

We consider the following two classes of repulsive functions $h_K(\bmu_1,\cdots,\bmu_K)$:
\begin{eqnarray}
\label{eqn:repulsive_function_1}
h_K(\bmu_1,\cdots,\bmu_K)&=&\min_{1\leq k<k'\leq K}g(\|\bmu_k-\bmu_{k'}\|),\\
\label{eqn:repulsive_function_2}
h_K(\bmu_1,\cdots,\bmu_K)&=&\left[\prod_{1\leq k<k'\leq K}g\left(\|\bmu_k-\bmu_{k'}\|\right)\right]^{\frac{1}{K}},
\end{eqnarray}
for $K\geq2$, and $h_1(\bmu_1)\equiv 1$, 
where $g:\mathbb{R}_+\to[0,1]$ is a strictly monotonically increasing function with $g(0)=0$. Notice that the repulsive functions defined here generalize those in \cite{petralia2012repulsive, quinlan2017parsimonious}, who fix $K$ due to the challenges in estimating $K$ caused by the complicated relation between $Z_K$ and $K$.  
However, for the two repulsive functions \eqref{eqn:repulsive_function_1} and \eqref{eqn:repulsive_function_2}, we are able to find the connection between $Z_K$ and $K$ in \textbf{Theorem \ref{thm:bounding_Z_K}}, the proof of which is deferred to Section \ref{sec:proof_of_theorem_thm:bounding_z_k} of the Supplementary Material. We will discuss the non-asymptotic behavior of the posterior distribution of $K$ in Section \ref{sub:posterior_number_of_components}. 
\begin{theorem}\label{thm:bounding_Z_K}
	Suppose the repulsive function $h_K$ is either of the form \eqref{eqn:repulsive_function_1} or \eqref{eqn:repulsive_function_2}. If 
	$\iint_{\mathbb{R}^p\times\mathbb{R}^p}\left[\log g(\|\bmu_1-\bmu_2\|)\right]^2p(\bmu_1)p(\bmu_2)\mathrm{d}\bmu_1\mathrm{d}\bmu_2<\infty$,
	then $0\leq -\log Z_K\leq c_1K$ for some constant $c_1>0$. 
\end{theorem}
We refer to the prior $\Pi$ on $f\in\calM(\mathbb{R}^p)$ given by \eqref{eqns:MFM}, \eqref{eqn:repulsive_density}, \eqref{eqn:repulsive_function_1} or \eqref{eqn:repulsive_function_2} as the Bayesian \emph{repulsive Gaussian mixture} (RGM) model, denoted by $f\sim\mathrm{RGM}_1(\beta;g, p_{\bmu}, p_{\bSigma}, p_{K})$ if $h_K$ is of the form $\eqref{eqn:repulsive_function_1}$, or $f\sim\mathrm{RGM}_2(\beta;g, p_{\bmu}, p_{\bSigma}, p_{K})$ if $h_K$ is of the form $\eqref{eqn:repulsive_function_2}$.

\section{Theoretical Properties of the Posterior Distribution} 
\label{sec:posterior_convergence}

In this section we discuss the theoretical properties of the posterior of the RGM model defined in Section \ref{sec:preliminaries}. In particular, in the context of density estimation in nonparametric GMM, we establish the posterior consistency, discuss the posterior contraction rate, 
and study  
the shrinkage effect on the tail probability of the posterior number of components 
introduced by the repulsive prior. 
We defer the proofs of all theorems, corollaries, propositions, and lemmas to Sections  \ref{sec:proofs_of_posterior_consistency}, \ref{sec:proofs_for_posterior_contraction_rate}, and \ref{sec:proofs_for_the_model_complexity} of the Supplementary Material. 

\subsection{Preliminaries and Notations} 
\label{sub:preliminaries_and_notations}
We begin with some useful notations. Given a positive definite matrix $\bSigma$, we use $\lambda(\bSigma)$ to denote any eigenvalue  of $\bSigma$, and $\lambda_{\max}(\bSigma)$, $\lambda_{\min}(\bSigma)$ to denote the largest and smallest eigenvalue of $\bSigma$, respectively. Denote $\eye$  the identity matrix, and $\eye_p\in\mathbb{R}^{p\times p}$ the identity matrix of size $p\times p$ if specifying matrix dimension is needed. The Kullback-Leibler (KL) divergence between two densities $f$ and $g$ is denoted by $
\mathrm{D}(f\mid\mid g)=\int f\log(f/g)$. Denote $\|\cdot\|$ the Euclidean norm on $\mathbb{R}^p$. We use $\|\cdot\|_1$ to denote both the $L_1$-norm on $L^1(\mathbb{R}^p)$ and the $\ell_1$-norm on finite dimensional Euclidean space $\mathbb{R}^d$ for any $d\geq1$. $\|\cdot\|_\infty$ is used to denote both the $\ell_\infty$-norm of a vector and supremum norm of a bounded function. 
We use $\lfloor a\rfloor$ to denote the maximum integer that does not exceed $a$.
The notation $a \lesssim b$ is used throughout to represent $a\leq cb$ for some constant $c$ that is universal or unimportant for the analysis. Whenever possible, we use $\Pi$ to represent the prior/posterior probability measure, $\mathbb{P}_0$ and $\mathbb{E}_0$ to denote the probability and expectation with respect to the distribution $f_0$, and $p$ to denote all density functions in the model except $f_0$, $f$, and $\{f_F:F\in\calM(\mathbb{R}^p\times\calS)\}$. 
For random variables, we slightly abuse the notation and do not distinguish between the random variables themselves and their realizations. We shall also use $p(x)$ or $p_x(x)$ to denote the density of the random variable $x$. 

A weak neighborhood of $f_0$ is a set of densities containing a set of the form
$$
V=\left\{
f\in\calM(\mathbb{R}^p):\left|\int \varphi_i f_0-\int \varphi_i f\right|<\epsilon, i=1,\cdots,I
\right\},
$$
where $\varphi_i$'s are bounded continuous functions on $\mathbb{R}^p$ \citep{ghosal1999posterior}.
The posterior distribution is said to be \emph{weakly consistent} at $f_0$, if $\Pi(f\in U\mid \by_1,\cdots,\by_n)\to 1$ a.s. with respect to $\mathbb{P}_0$ for all weak neighborhoods $U$ of $f_0$. Given a prior $\Pi$ on $\calM(\mathbb{R}^p)$, a density function $f_0\in\mathcal{M}(\mathbb{R}^p)$ is said to be \emph{in the KL-support} of $\Pi$, or has the \emph{KL-property} (with respect to $\Pi$), if $\Pi(f\in\calM(\mathbb{R}^p):\mathrm{D}(f_0\mid\mid f)<\epsilon)>0$ for all $\epsilon>0$. The posterior distribution is said to be \emph{$L_1$(strongly) consistent} at $f_0$, if for all $\eps>0$, $\Pi(f\in\calM(\mathbb{R}^p):\|f-f_0\|_1>\eps\mid \by_1,\cdots,\by_n)\to0$ as $n\to\infty$ or in $\mathbb{P}_0$-probability. The \emph{posterior contraction rate} is any sequence $(\eps_n)_{n=1}^{\infty}$ such that $\Pi(f\in\calM(\mathbb{R}^p):\|f-f_0\|_1>M\eps_n\mid \by_1,\cdots,\by_n)\to 0$ as $n\to\infty$ in $\mathbb{P}_0$-probability for some constant $M>0$. 
Given a family of densities $\mathcal{F}$ on $\mathbb{R}^p$ with a metric $d$ on $\mathcal{F}$,
the \emph{$\epsilon$-covering number} of $\mathcal{F}$ with respect to $d$, denoted by $\mathcal{N}(\epsilon, \mathcal{F},d)$, is defined to be the minimum number of $\epsilon$ balls of the form $\{g\in\mathcal{F}:d(f,g)<\epsilon\}$ that are needed to cover $\mathcal{F}$. The $d$-metric entropy is the logarithm of the covering number under the $d$-metric. 

Above all, we assume that $f\sim\mathrm{RGM}_r(\beta;g, p_{\bmu},p_{\bSigma},p_K)$, $r=1$ or $2$. In order to develop the posterior convergence theory, we need some regularity conditions, most of which are typically satisfied in practice. We group these conditions into two categories. The first set of conditions are the requirements for the model.
\begin{itemize}[noitemsep,topsep=0cm]
	\item[\textbf{A0}]The data generating density $f_0$ is of the form $f_0=\phi_{\bSigma}*F_0$ for some $F_0\in\calM(\mathbb{R}^p\times\calS)$ that has a sub-Gaussian tail: $F_0\left(\|\bmu\|\geq t\right)\leq B_1\exp(-b_1t^2)$ for some $B_1,b_1>0$. 
	\item[\textbf{A1}]For some $\delta>0,c_2>0$, we have $g(x)\geq c_2\epsilon$ whenever $x\geq\epsilon$ and $\epsilon\in(0,\delta)$.
	\item[\textbf{A2}]$g$ satisfies $
	\iint_{\mathbb{R}^{p}\times \mathbb{R}^{p}} [\log g(\|\bmu_1-\bmu_2\|)]^2p(\bmu_1)p(\bmu_2)\mathrm{d}\bmu_1\mathrm{d}\bmu_2<\infty.$
	\item[\textbf{A3}]For some $\lsigma^2,\usigma^2\in(0,+\infty)$, we have $\lsigma^2\leq\inf_{\calS}\lambda(\bSigma)\leq\sup_{\calS}\lambda(\bSigma)\leq\usigma^2$.
	\item[\textbf{A4}]For some (non-random) unitary $\bU\in\mathbb{R}^{p\times p}$, $\bU\transpose\bSigma\bU$ is diagonal for all $\bSigma\in\calS$.
\end{itemize}
Condition A2 guarantees that $1/Z_K$ does not grow super-exponentially in $K$ by \textbf{Theorem \ref{thm:bounding_Z_K}}. Conditions A0 and A3 assume that both $f_0$ and $f$ are of the nonparametric GMM form, and hence guarantee that $f_0$ and $f$ are not too ``spiky'' such that a faster rate of convergence is obtainable. Condition A4, the simultaneous diagonalizability of all $\bSigma\in\calS$, appears to be of less importance, but it turns out that a structured space $\calS$ of covariance matrices decreases the $\|\cdot\|_1$-metric entropy of the proposed sieves in Section \ref{subsec:weak_consistency_kullback_leibler_property}, and hence affects the posterior contraction rate. We assume that $\bU\transpose\bSigma\bU=\mathrm{diag}(\lambda_1,\cdots,\lambda_p)$ for all $\bSigma\in\calS$, \emph{i.e.}, the eigenvalues of $\bSigma\in\calS$ are ordered according to the orthonormal eigenvectors in $\bU$.

We also need some requirements for the prior distributions. 
\begin{itemize}[noitemsep,topsep=0cm]
	\item[\textbf{B1}]$(w_1,\cdots,w_K\mid K)\sim\mathcal{D}_K(\beta)$ is weakly informative: $\beta\in(0,1]$.
	\item[\textbf{B2}]$p_{\bmu}$ has a sub-Gaussian tail: $\int_{\{\|\bmu\|\geq t\}}p(\bmu)\mathrm{d}\bmu\leq B_2\exp(-b_2t^2)$ for some $B_2,b_2>0$. 
	\item[\textbf{B3}]For all $\bmu\in\mathbb{R}^p$, $p(\bmu)\geq B_3\exp(-b_3\|\bmu\|^\alpha)$ for some $\alpha\geq2,B_3,b_3>0$.
	\item[\textbf{B4}]$p(\bSigma)$ is induced by $\prod_{j=1}^p p_\lambda(\lambda_j(\bSigma))$ with $\mathrm{supp}(p_\lambda)=[\lsigma^2,\usigma^2]$. 	
	\item[\textbf{B5}]There exists some $B_4,b_4>0$ such that for sufficiently large $K$, we have
	\begin{eqnarray}
	p_K(K)\geq\exp\left(-b_4K\log K\right),\quad\sum_{N=K}^{\infty}p_K(N)\leq\exp\left(-B_4K\log K\right).\nonumber
	\end{eqnarray}
\end{itemize}
Condition B1 assumes a vague prior on $(w_1,\cdots,w_K)$. Conditions B2 and B3 are requirements for the tail behavior of the function $p_{\bmu}$ in the sense that they are neither heavier than Gaussian nor thinner than an exponential power density \citep{scricciolo2011posterior}. Alternatively, one may assume $p(\bmu)\propto\exp(-b_3\|\bmu\|^\alpha)$ for some $b_3>0$, as suggested by \cite{kruijer2010adaptive}. Condition B4 is adopted in \cite{ghosal2001entropies} to obtain an ``almost'' parametric convergence rate. We will also discuss possible extensions to the case where $p_\lambda$ has full support on $(0,+\infty)$ later in this section. Condition B5 is the requirement for the tail behavior of the prior on $K$. Similar assumption on the tail behavior of the prior on $K$ is adopted in \cite{kruijer2010adaptive} and \cite{shen2013adaptive} for finite mixture models. As a useful example, we show that the commonly used zero-truncated Poisson prior on $K$ satisfies condition B5. 
\begin{example}
	The zero-truncated Poisson prior has a density function $p_K(K)=\frac{\lambda^K}{(\mathrm{e}^\lambda-1)K!}\mathbb{I}(K\geq1)$ with respect to the counting measure on $\mathbb{N}_+$ for some intensity parameter $\lambda>0$. Directly compute
	\begin{eqnarray}
	\sum_{N=K+1}^{\infty}p_K(N)=\frac{1}{\mathrm{e}^\lambda-1}\left(\mathrm{e}^\lambda-\sum_{N=0}^{K}\frac{\lambda^N}{N!}\right)=\frac{1}{\mathrm{e}^\lambda-1}\int_0^{\lambda}\frac{(\lambda-t)^{K}\mathrm{e}^t\mathrm{d}t}{K!}\lesssim\frac{\lambda^{K+1}}{(K+1)!}\nonumber,
	\end{eqnarray}
	where the second equality is due to Taylor's expansion. By Stirling's formula, this is further upper bounded by $(\frac{\lambda\mathrm{e}}{K+1})^{K+1}$. Therefore, substituting $K+1$ with $K$, we obtain
	\begin{eqnarray}
	\sum_{N=K}^{\infty}p_K(N)\lesssim\exp\left(K\log (\lambda\mathrm{e})-K\log K\right)\leq\exp\left(-\frac{1}{2}K\log K\right)\nonumber
	\end{eqnarray}
	for sufficiently large $K$. The constant for $\lesssim$ can be absorbed into the exponent, and hence we conclude $\sum_{N=K}^{\infty}p_K(N)\leq\exp(-B_4K\log K)$ for some $B_4>0$.
	
	For the lower bound on $p(K)$, for sufficiently large $K$ we again use Stirling's formula,
	\begin{eqnarray}
	p(K)=\frac{1}{\mathrm{e}^\lambda-1}\frac{\lambda^K}{K!}\geq\exp(K\log(\lambda\mathrm{e})-\log K-K\log K)\geq\exp(-2K\log K).\nonumber
	\end{eqnarray}
	Hence the zero-truncated Poisson prior on $K$ satisfies condition B5. 
\end{example}

\subsection{Posterior Consistency} 
\label{subsec:weak_consistency_kullback_leibler_property}
\textbf{Weak consistency.} Using the result from \cite{schwartz1965bayes}, a sufficient condition for $\Pi$ to be weakly consistent at $f_0$ is that $f_0$ is in the KL-support of $\Pi$. 
The following lemma is useful in that it provides a compactly supported $F_m$ such that $f_{F_m}$ can approximate $f_0$ arbitrarily well in the KL divergence sense.
\begin{lemma}\label{lemma:Approximation}
	Assume conditions A0-A4 and B1-B5 hold. 
	For all $m\in\mathbb{N}_+$, define a sequence of distributions $(F_m)_{m=1}^{\infty}$ by $F_m(A)=c_mF_0(A\cap\calT_m)$ for any measurable $A\subset\mathbb{R}^p\times\calS$, where 
	$$
	\calT_m=\left\{(\bmu:\bSigma)\in\mathbb{R}^p\times\calS:\|\bmu\|\leq m,\lsigma^2+\frac{1}{m}\leq\lambda_{\min}(\bSigma)\leq\lambda_{\max}(\bSigma)\leq\usigma^2-\frac{1}{m}\right\}
	$$
	and $c_m$ is the normalizing constant for $F_m$ with $c_m^{-1}=F_0(\calT_m)$. 
	Then $\int f_0\log\frac{f_0}{f_{F_m}}\to 0$ as $m\to\infty$.
\end{lemma}
We remark that the construction in \cite{wu2008kullback} is not directly applicable. The major reason is that the variance of the convolving $\phi$ is allowed to be arbitrarily close to $0$ there, whereas we impose uniform boundedness on the eigenvalues of the covariance matrices. The sequence of densities constructed in \cite{wu2008kullback} is $(f_m(\by))_{m=1}^\infty=\left(\int_{\mathbb{R}^p}\phi_{\sigma_m^2\eye}(\by-\bmu)f_0(\by)\mathrm{d}\by\right)_{m=1}^\infty$, where $(\sigma_m)_{m=1}^\infty$ is a sequence that converges to $0$ at a certain rate. This construction does not apply when covariance matrices are bounded in spectrum. 
The construction of the sequence of densities $(f_{F_m})_{m=1}^\infty$ in \textbf{Lemma \ref{lemma:Approximation}} also serves as a technical contribution to the Kullback-Leibler property of Bayesian nonparametric GMM.

Based on \textbf{Lemma \ref{lemma:Approximation}}, we are able to establish the weak consistency via the KL-property.
\begin{theorem}\label{thm:weak_consistency}
	Assume conditions A0-A4 and B1-B5 hold. 
	Then $f_0$ is in the KL-support of $\Pi$, and hence $\Pi(\cdot\mid\by_1,\cdots,\by_n)$ is weakly consistent at $f_0$. 
\end{theorem}
\noindent\textbf{Strong consistency.} 
To establish the posterior strong consistency, we utilize Theorem 1 in \cite{canale2017posterior}, which is a standard result for proving consistency for general Bayesian nonparametric density estimation methods (see Section \ref{sec:supporting_results} of the Supplementary Material). Specializing to the RGM model, we need to construct a sequence of submodels and partitions of each of these submodels that satisfy the conditions in Theorem 1 in \cite{canale2017posterior}. We now make these statements precise. 
Consider the following 
submodels of $\calM(\mathbb{R}^p)$:
\begin{eqnarray}
\mathcal{F}_{K_n}=\left\{f_F:F=\sum_{k=1}^Kw_k\delta_{(\bmu_k,\bSigma_k)},K\leq K_n,\bmu_k\in\mathbb{R}^p,\bSigma_k\in\calS\right\}\nonumber
\end{eqnarray}
and the following partition of the submodel $\mathcal{F}_{K_n}$
\begin{eqnarray}
\mathcal{G}_K(\bolda_K)=\mathcal{F}_K\left(\prod_{k=1}^K(a_k,a_k+1]\right),\quad\bolda_K=(a_1,\cdots,a_K)\in\mathbb{N}^K, \quad K=1,\cdots,K_n,\nonumber
\end{eqnarray}
where 
\begin{eqnarray}
\mathcal{F}_{K}\left(\prod_{k=1}^K(a_k,b_k]\right)=\left\{f_F:F=\sum_{k=1}^Kw_k\delta_{(\bmu_k,\bSigma_k)},\|\bmu_k\|_\infty\in(a_k,b_k]\right\}.\nonumber
\end{eqnarray}
According to Theorem 1 in \cite{canale2017posterior}, it suffices to show the following:
$f_0$ is in the KL-support of $\Pi$, and there exists some $b,\tilde{b}>0$, some sequence $(K_n)_{n=1}^\infty$, such that $\Pi(\mathcal{F}_{K_n}^c)\lesssim\mathrm{e}^{-bn}$ for sufficiently large $n$, and for all $\epsilon>0$, 
\begin{eqnarray}\label{eqn:summability_condition}
\lim_{n\to\infty}\mathrm{e}^{-(4-\tilde{b})n\epsilon^2}\sum_{K=1}^{K_n}\sum_{a_1=0}^\infty\cdots\sum_{a_K=0}^\infty \sqrt{\mathcal{N}\left(\epsilon,\mathcal{G}_K(\bolda_K),\|\cdot\|_1\right)}\sqrt{\Pi\left(\mathcal{G}_K(\bolda_K)\right)}=0.
\end{eqnarray}
\begin{lemma}\label{lemma:entropy}
	Let $a_k<b_k$ be non-negative integers, $k=1,\cdots,K$. Then for sufficiently small $\delta>0$,
	there exists constant $c_3>0$ such that 
	\begin{eqnarray}
	\mathcal{N}\left(\delta, \mathcal{F}_K\left(\prod_{k=1}^K(a_k,b_k]\right),\|\cdot\|_1\right)
	\leq\left(\frac{c_3}{\delta^{2p+1}}\right)^K\left(\prod_{k=1}^Kb_k\right)^p.\nonumber
	\end{eqnarray}
\end{lemma}
\begin{lemma}\label{lemma:summability_upper_bound}
	Assume conditions A0-A4 and B1-B5 hold. Then we have
	\begin{eqnarray}
	\sum_{K=1}^{K_n}\sum_{a_1=0}^\infty\cdots\sum_{a_K=0}^\infty \sqrt{\mathcal{N}\left(\delta,\mathcal{G}_K(\bolda_K),\|\cdot\|_1\right)}\sqrt{\Pi\left(\mathcal{G}_K(\bolda_K)\right)}\leq K_n\left(\frac{M}{\delta^{p+\frac{1}{2}}}\right)^{K_n}.\nonumber
	\end{eqnarray}
	for sufficiently small $\delta$ for some constant $M>0$.
\end{lemma}
\noindent Based on \textbf{Lemma \ref{lemma:entropy}} and \textbf{Lemma \ref{lemma:summability_upper_bound}}, we are able to verify \eqref{eqn:summability_condition} and hence establish the strong consistency.
\begin{theorem}\label{thm:strong_consistency}
	Assume conditions A0-A4 and B1-B5 hold. Then $\Pi(\cdot\mid\by_1,\cdots,\by_n)$ is strongly consistent at $f_0$. 
\end{theorem}

\subsection{Posterior Contraction Rate} 
\label{subsec:prior_concentration_rate}
To compute the posterior contraction rate, it is sufficient to find two sequences $(\leps_n)_{n=1}^\infty,(\ueps_n)_{n=1}^\infty$ such that
\begin{align}
\label{eqn:posterior_exponential_decay}
&\Pi\left(\mathcal{F}_n^c\right)\lesssim\exp(-4n\leps_n^2),\\
\label{eqn:summability_rate}
&\exp(-n\ueps_n^2)\sum_{K=1}^{K_n}\sum_{a_1=0}^\infty\cdots\sum_{a_K=0}^\infty \sqrt{\mathcal{N}\left(\ueps_n,\mathcal{G}_K(\bolda_K),\|\cdot\|_1\right)}\sqrt{\Pi\left(\mathcal{G}_K(\bolda_K)\right)}\to 0,\\
\label{ineq:prior_concentration}
&\Pi\left(f:\int f_0\log\frac{f_0}{f}\leq\leps_n^2,\int f_0\left(\log\frac{f_0}{f}\right)^2\leq\leps_n^2\right)\geq\exp(-n\leps_n^2).
\end{align}
(See Theorem 3 in \citealp{kruijer2010adaptive}, which is also provided in Section \ref{sec:supporting_results} in the Supplementary Material). For notation convenience we refer to the set of densities\\ $\left(f:\int f_0\log\frac{f_0}{f}\leq\epsilon_n^2,\int f_0\left(\log\frac{f_0}{f}\right)^2\leq\epsilon^2\right)$ as the KL-type ball, and denote it as $B(f_0,\epsilon)$. Equation \eqref{ineq:prior_concentration} is also known as the prior concentration condition. 

\textbf{Lemma \ref{lemma:summability_upper_bound}} not only plays a fundamental role in establishing the posterior strong consistency, but also provides an upper bound for the sum in terms of $\delta$, which is again used to verify equation \eqref{eqn:summability_rate}. \textbf{Proposition \ref{prop:supersmooth_rate}} finds the rates $(\leps_n)_{n=1}^{\infty},(\ueps_n)_{n=1}^{\infty}$ that satisfy \eqref{eqn:posterior_exponential_decay} and \eqref{eqn:summability_rate}. 
\begin{proposition}\label{prop:supersmooth_rate}
	Assume conditions A0-A4 and B1-B5 hold. 
	Let $\leps_n=(\log n)^{t_0}/\sqrt{n}$, $\ueps_n=(\log n)^{t}/\sqrt{n}$ where $t$ and $t_0$ satisfy $t>t_0+\frac{1}{2}>\frac{1}{2}$, and $K_n=\lfloor(p+1)^{-1}(\log n)^{2t-1}\rfloor$. Then 
	\eqref{eqn:posterior_exponential_decay} and \eqref{eqn:summability_rate} hold. 
\end{proposition}
We are now left with finding the prior concentration rate $(\leps_n)_{n=1}^\infty$ that satisfies \eqref{ineq:prior_concentration}. 
In particular, we need to bound the KL-type balls $B(f_0,\epsilon)$ by the $L_1$ distance. The strategy is to approximate $F_0$ using a finitely discrete distribution with sufficiently small number of support points. \textbf{Lemma \ref{lemma:KL_ball_lower_bound}} allows us to formalize this idea. 
\begin{lemma}\label{lemma:KL_ball_lower_bound}
	Assume conditions A0-A4 and B1-B5 hold. For some constant $\eta>0$ and for all sufficiently small $\epsilon>0$, there exists a discrete distribution $F^\star=\sum_{k=1}^Nw_k^\star\delta_{(\bmu_k^\star,\bSigma_k^\star)}$ supported on a subset of $\{(\bmu,\bSigma)\in\mathbb{R}^p\times\calS:\|\bmu\|_\infty\leq 2a\}$ with $a=b_1^{-\frac{1}{2}}\left(\log\frac{1}{\epsilon}\right)^{\frac{1}{2}}$, $\|\bmu_k^\star-\bmu_{k'}^\star\|_\infty\geq2\epsilon$, $|\lambda_j(\bSigma_k^\star)-\lambda_j(\bSigma_{k'}^\star)|\geq2\epsilon$ whenever $k\neq k'$, $j=1,\cdots,p$, $N\lesssim \left(\log\frac{1}{\epsilon}\right)^{2p}$, such that
	\begin{eqnarray}
	\left\{f_F:F=\sum_{k=1}^Nw_k\delta_{(\bmu_k,\bSigma_k)}:(\bmu_k,\bSigma_k)\in E_k,\sum_{k=1}^N|w_k-w_k^\star|<\epsilon\right\}\subset B\left(f_0,\eta\epsilon^{\frac{1}{2}}\left(\log\frac{1}{\epsilon}\right)^{\frac{p+4}{4}}\right)\nonumber,
	\end{eqnarray}
	where $$	E_k=\left\{(\bmu,\bSigma)\in\mathbb{R}^p\times\calS:\|\bmu-\bmu_k^\star\|_\infty<\frac{\epsilon}{2},|\lambda_j(\bSigma)-\lambda_j(\bSigma_k^\star)|<\frac{\epsilon}{2},j=1,\cdots,p\right\}. $$
\end{lemma}
We are in a position to derive the posterior contraction rates for the RGM model.
\begin{theorem}\label{thm:contraction_rate}
	Assume conditions A0-A4 and B1-B5 hold. Then the posterior distribution $\Pi(\cdot\mid\by_1,\cdots,\by_n)\to 0$ contracts at $f_0$ with rate $\epsilon_n=(\log n)^{t}/\sqrt{n}$, $t>p+\frac{\alpha+2}{4}$. 
\end{theorem}
It is interesting that the RGM model and some other independent-prior models (\emph{e.g.} DP mixtures of Gaussians) yield similar posterior contraction rate. The major complication for the RGM model comes from proving that the KL-type ball is assigned with sufficiently large prior probability, since in the RGM model the repulsive function $h$ can only be lower bounded by $0$, whereas $h$ is always unity in the independent-prior model. 
\begin{remark}
	Notice that the optimal rate $(\log n)^{(p+1)+}/\sqrt{n}$ is achieved when $\alpha=2$, where $(p+1)+$ means that any $t>p+1$ is satisfied. Namely, the posterior contraction rate is optimal when $p_\bmu$ has a Gaussian tail. For comparison, recall that for general location-scale Gaussian mixture problem with bounded variance, {Theorem 6.2} in \cite{ghosal2001entropies} gives a contraction rate of $(\log n)^{3.5}/\sqrt{n}$ in the univariate case ($p=1$) using the DP mixture model, in which the distribution of the location parameters is Gaussian. Analogously, in the RGM model, we may use Gaussian $p_\bmu$ to control the tail rate of the joint distribution of $(\bmu_1,\cdots,\bmu_K)$ as $\|\bmu_k\|$ gets large, since the repulsive function $h_K$ is bounded. \textbf{Theorem \ref{thm:contraction_rate}} improves the  contraction rate to $(\log n)^t/\sqrt{n}$ with $t>2$ compared to that given by \cite{ghosal2001entropies}. 
	However, such an improvement is not due to the repulsive structure of the prior. The underlying reason is that we use Theorem 3 in \cite{kruijer2010adaptive} to derive the posterior contraction rate, whereas \cite{ghosal2001entropies} use Theorem 2.1 in \cite{ghosal2000convergence}, a weaker version of Theorem 3 in \cite{kruijer2010adaptive}, to derive it. 
	In other words, this suggests that it is also possible to obtain an improved posterior contraction rate for some independent-prior GMM for component centers using Theorem 3 in \cite{kruijer2010adaptive}. 
\end{remark}
\begin{remark}
	The boundedness on the eigenvalues of the covariance matrices (condition A3) was originally adopted in \cite{ghosal2001entropies}, which is necessary to obtain an ``almost'' parametric rate $(\log n)^t/\sqrt{n}$ for some $t>0$. \cite{walker2007rates} adopted the same assumption and improved the posterior contraction rate of the location mixture problem. Requiring $p_\lambda$ to have full support on $(0,+\infty)$, however, is necessary in cases where the underlying true density $f_0$ is no longer of the form $f_0=\phi_{\bSigma} * F_0$ for some $F_0\in\calM(\mathbb{R}^p\times\calS)$. For general mixtures of finite location mixture models, the contraction rate is known to be $(\log n)^t n^{-\tilde{\beta}/(2\tilde{\beta}+d)}$ for some $t>0$, where $f_0$ is in a locally $\tilde{\beta}$-H{\" o}lder class \citep{shen2013adaptive}.
	It will be interesting to extend \textbf{Theorem \ref{thm:contraction_rate}} to the case where
	$\mathrm{supp}(p_\lambda)=(0,+\infty)$ and explore the corresponding posterior contraction rate. 
\end{remark}

\subsection{Shrinkage Effect on the Posterior of $K$} 
\label{sub:posterior_number_of_components}

The behavior of the posterior of $K$ is of great interest, since it is a measurement of the complexity of a nonparametric density estimator. If a parametric assumption on $f_0$ is made in the sense that $f_0=\phi_\bSigma * F_0$ for some finitely discrete $F_0\in\calM(\mathbb{R}\times\calS)$, then under mild regularity condition, \cite{nobile1994bayesianmixture} proved that the posterior distribution $p(K\mid\by_1,\cdots,\by_n)$ converges weakly to the point mass at $K_0$ a.s. under the MFM model, where $K_0$ is the number of support points of $F_0$. However, when $F_0$ is no longer assumed to be finitely discrete, and a repulsive prior is introduced among components in MFM, there is little result concerning the mixture complexity in the literature. 
This issue is addressed in \textbf{Theorem \ref{thm:model_complexity}} in terms of the shrinkage effect on the tail probability of the posterior of $K$ in the presence of the repulsive prior. 
For simplicity, we only consider the case where both $f_0$ and the model are of location-mixture form only. 
We also assume that the $g$ function is of the form $g(x)=\frac{x}{g_0+x}$ for some $g_0\in[0,\infty)$, $x>0$. In particular, we allow $g_0=0$ so that the RGM model includes the special case of the independent-prior GMM. 
\begin{theorem}\label{thm:model_complexity}
	Suppose $f_0(\by)=\int_{\mathbb{R}}\phi_{\bSigma_0}(\by-\bmu)F_0(\mathrm{d}\bmu)$ for some fixed $\bSigma{}_0\in\calS$ and conditions A0-A3 and B1-B3 hold with $\beta=1$, $p(\bmu) = \phi(\bmu|\zero,\tau^2\eye)$, and $p_\bSigma = \delta_{\bSigma}$. Without loss of generality assume $\int_{\mathbb{R}^p}\mb F_0(\mathrm{d}\mb) = \zero$. Further assume $p(K)=\Omega Z_K \frac{\lambda^K}{K!}$ where $\Omega=\left[\sum_{K=1}^\infty Z_K\frac{\lambda^K}{K!}\right]^{-1}$ and $g$ is of the form $g(x)=\frac{x}{g_0+x}$ for some $g_0\geq0$, $x>0$. Suppose that $f\sim\mathrm{RGM}_1(1,g,\phi(\bmu|\zero,\tau^2\eye),\delta_{\bSigma_0},p(K))$, where $r=1$ or $2$. Then when $N\geq 3$, we have the following result:
	\begin{align}
	\mathbb{E}_0\left[\Pi(K\geq N\mid\by_1,\cdots,\by_n)\right]\leq C(\lambda)\chi_r(g_0;n,N)\exp\left[\frac{n\tau^2}{2}\mathrm{tr}\left(\bSigma_0^{-1}\right)\right]\sum_{K=N+1}^\infty\frac{\lambda^K}{(\mathrm{e}^\lambda - 1)K!}\nonumber,
	\end{align}
	where
	\[
	\chi_r(g_0;n,N) = 
	\left\{
	\begin{aligned}
	&\frac{\left(1+g_0^{\frac{3}{2}}\delta(\tau)\right)^{\frac{2}{3}}\left[2p\tau^2+\frac{2n}{N}\tau^4\mathbb{E}_0\left(\mb\transpose{}\bSigma_0^{-2}\mb\right)\right]^{\frac{1}{2}}}{g_0+\left[2p\tau^2+\frac{2n}{N}\tau^4\mathbb{E}_0\left(\mb\transpose{}\bSigma_0^{-2}\mb\right)\right]^{\frac{1}{2}}},\quad\text{if }r=1,\\
	&\frac{\left(1+\delta(\tau)\sqrt{g_0}\right)\left[2p\tau^2+\frac{2n}{N}\tau^4\mathbb{E}_0\left(\mb\transpose{}\bSigma_0^{-2}\mb\right)\right]^{\frac{1}{2}}}{g_0+\left[2p\tau^2+\frac{2n}{N}\tau^4\mathbb{E}_0\left(\mb\transpose{}\bSigma_0^{-2}\mb\right)\right]^{\frac{1}{2}}},\quad\text{if }r=2.
	\end{aligned}\right.
	\]
	Here $C(\lambda)$ are some constants depending on $\lambda$ only, $\delta(\tau)$ is a constant depending on $\tau$ only such that $\delta(\tau)<1$ for sufficiently large $\tau$, $\mathbb{E}_0\left(\mb\transpose{}\bSigma{}_0^{-2}\mb\right):=\int_{\mathbb{R}^p}\mb\transpose{}\bSigma_0^{-2}\mb F_0(\mathrm{d}\mb)$, and $\chi_r(g_0;n,N)$ is referred to as the shrinkage constant. 
\end{theorem}
As pointed out in Section \ref{sec:preliminaries}, the normalizing constant $Z_K$ yields complication in the posterior inference of $K$. In \textbf{Theorem \ref{thm:model_complexity}} the prior density $p(K)$ of the number of components is assumed to be proportional to the Poisson density function modulus $Z_K$ to eliminate such effect: $p(K)\propto Z_K\frac{\lambda^K}{K!}$. 
\textbf{Theorem \ref{thm:model_complexity}} unveils the relationship between the tail probability of the marginal posterior of $K$ and the hyperparameter $g_0$ that introduces repulsion: as long as $\tau$ is moderately large so that $\delta(\tau)<1$ (corresponding to the weakly informative prior), the upper bound for $\mathbb{E}_0\left[\Pi\left(K>N|\by_1,\cdots,\by_n\right)\right]$ decreases as $g_0$ increases when $g_0$ is large enough. In particular, the shrinkage constant $\chi_r(g_0;n,N)$ is $1$ when $g_0=0$ (\emph{i.e.}, no repulsion is enforced among component centers), decreases when $g_0$ increases, and is smaller than $1$ for large enough $g_0$. Namely, compared to the independent prior for the component centers $\bmu_k$'s, the repulsive prior introduces additional shrinkage effect on the tail probability of the posterior of $K$. 
In addition, it is worth mentioning that \textbf{Theorem \ref{thm:model_complexity}} is a non-asymptotic result.

\textbf{Theorem \ref{thm:model_complexity}} also serves as a guidance for constructing a sample-size dependent RGM prior that yields a slower rate of growth of $K$ compared to the independent-prior Gaussian mixture model. Specifically, instead of using a hyperparameter $g_0$ that does not change with $n$, it is possible to choose a sample-size dependent hyperparameter $g_0(n)$ that tends to infinity and thus affects the rate of decay of $\mathbb{E}_0\left[\Pi(K\geq K_n|\by_1,\cdots,\by_n)\right]$ for certain sequences of $(K_n)_{n=1}^\infty$. However, the prior concentration condition might no longer hold, potentially resulting a slower posterior contraction rate. It might be interesting to explore the trade-off between the shrinkage effect on $K$ and the posterior contraction rate using sample-size dependent repulsive prior. 
\begin{corollary}\label{corr:model_complexity}
	Assume the conditions in \textbf{Theorem \ref{thm:model_complexity}} hold. If the sequence $(K_n)_{n=1}^{\infty}\subset\mathbb{N}_+$ satisfies $\liminf_{n\to\infty}K_n/n>0$, then the tail probability of the posterior distribution of $K$ satisfies
	$\Pi(K\geq K_n\mid\by_1,\cdots,\by_n)\to 0$
	in $\mathbb{P}_0$-probability as $n\to\infty$. 
\end{corollary}
\begin{remark}
	In terms of $K$, the number of support points in the RGM model, which is a measurement of the model complexity of estimating an unknown density, \textbf{Corollary \ref{corr:model_complexity}} says that the posterior probability that $K$ is at least a non-negligible fraction of $n$ (in the limit) converges to $0$ in $\mathbb{P}_0$-probability as $n\to\infty$. In other words, the posterior number of components grows sub-linearly with respect to the sample size. 
\end{remark}

\section{Posterior Inference} 
\label{sec:posterior_inference}
For the DPP mixture model, \cite{xu2016bayesian} developed a variation of the RJ-MCMC sampler
that can be extended to the RGM model. However, the reversible-jump moves in multi-dimensional problems could be challenging and inefficient. In this section, we design an efficient and easy-to-implement blocked-collapsed Gibbs sampler by representing the RGM model using the random partition distribution.

Let us begin with characterizing the RGM model using the latent cluster configurations. Given a random measure $F=\sum_{k=1}^Kw_k\delta_{(\bmu_k,\bSigma_k)}$ with $(w_1,\cdots,w_K)\sim\mathcal{D}_K(\beta)$, we may represent the finite mixture model as follows by integrating out $(w_1,\cdots,w_K)$:
\begin{eqnarray}
(\by_i\mid z_i,\{\bmu_k,\bSigma_k\}_{k=1}^K,K)&\sim&\mathrm{N}(\bmu_{z_i},\bSigma_{z_i}),\nonumber\\
\label{eqn:latent_class_configuration}
p(z_1,\cdots,z_n\mid K)&=&\frac{\Gamma(K\beta)}{\Gamma(n+K\beta)}\prod_{k=1}^K\frac{\Gamma(\beta+\sum_{i=1}^n\mathbb{I}(z_i=k))}{\Gamma(\beta)}.
\end{eqnarray}
Let $\mathcal{C}_n$ denote the partition of $\{1,\cdots,n\}$ induced by $\bz=(z_1,\cdots,z_n)$ as $\mathcal{C}_n=\{E_k:|E_k|>0\}$ where $E_k=\{i:z_i=k\}$ for $k=1,\cdots,K$, and $|E|$ denotes the cardinality of a finite set $E$. For example, if one has $\bz=(z_1,z_2,z_3,z_4,z_5,z_6)=(1, 3, 4, 4, 3, 1)$ with $n=6$, then the corresponding partition is $\mathcal{C}_6=\{\{1,6\},\{2,5\},\{3,4\}\}$.
Using the exchangeable partition distribution in \cite{miller2016mixture}, we establish the generalized urn-model induced by the RGM model in \textbf{Theorem \ref{thm:urn_model}} after marginalizing out the intractable random distribution $F$. 
The proof is provided in Section \ref{sec:proof_of_the_generalized_urn_model} of the Supplementary Material. 
\begin{theorem}\label{thm:urn_model}
	Suppose the prior $\Pi$ on $\mathcal{M}(\mathbb{R}^p)$ is defined as in Section \ref{sec:preliminaries}, and the latent class configuration variables $\bz = (z_1,\cdots,z_n)$ is defined as in \eqref{eqn:latent_class_configuration}. Let $\bgamma_i=\bmu_{z_i},\bGamma_i=\bSigma_{z_i}$, $\btheta_i=(\bgamma_i,\bGamma_i)$, $i=1,\cdots,n$, $\mathcal{C}_{n-1}$ be the partition on $\{1,\cdots,n-1\}$ induced by $\btheta_1,\cdots,\btheta_{n-1}$, $(\bgamma_c^\star:c\in\mathcal{C}_{n-1})$ be the unique values of $(\bgamma_1,\cdots,\bgamma_{n-1})$, and $(\bSigma_c^\star:c\in\mathcal{C}_{n-1})$ be those of $(\bGamma_1,\cdots,\bGamma_{n-1})$. Let $\ell=|\mathcal{C}_{n-1}|$ be the number of clusters, and $K$  be the number of components in $F$, where $K\geq\ell$. Denote $\mathcal{C}_\varnothing\subset\mathbb{N}_+$ the indexes for the components associated with no observations with $|\mathcal{C}_\varnothing|=K-\ell$, $((\bgamma_c^\star,\bGamma_c^\star)\in\mathbb{R}^p\times\calS:c\in \mathcal{C}_{\varnothing})$ the component-specific parameters of the components that are not associated with any observation, and $\underline{c}=\min(c:c\in C_\varnothing)$ provided that $K\geq \ell+1$. Denote $\Pi(\btheta_n\in\cdot\mid-)$ the full conditional distribution of $\btheta_n$ with $F$ marginalized out. Then
	\begin{align}
	\label{eqn:Gibbs_sampler_urn_model}
	\Pi(\btheta_n\in\cdot|-)
	\propto\left[\frac{V_n(\ell+1)\beta}{V_n(\ell)}\right]\sum_{K=\ell+1}^{\infty}\alpha_K
	G_K(\cdot)+\sum_{c\in\mathcal{C}_{n-1}}(|c|+\beta)\phi(\by_n|\bgamma_c^\star,\bGamma_c^\star)\delta_{(\bgamma_c^\star,\bGamma_c^\star)}(\cdot),
	\end{align}
	where 
	\begin{eqnarray}
	V_n(\ell)&=&\sum_{K=\ell}^\infty\frac{K(K-1)\cdots (K-\ell+1)}{(\beta K)(\beta K + 1)\cdots(\beta K+n-1)}p_K(K),\nonumber\\
	\alpha_K&=&m_Kp(K\mid \mathcal{C}_n=\mathcal{C}_{n-1}\cup\{\{n\}\}),\nonumber\\
	m_K&=&\frac
	{\displaystyle\int\cdots\iint \phi(\by_n\mid\bgamma_{\underline{c}}^\star,\bGamma_{\underline{c}}^\star)
		h_K(\bgamma_{c}^\star:c\in\mathcal{C}_{n-1}\cup\mathcal{C}_\varnothing)p_{\bSigma}(\bGamma_{\underline{c}}^\star)\mathrm{d}\bGamma_{\underline{c}}^\star\prod_{c\in\mathcal{C}_\varnothing}p_{\bmu}(\bgamma_c^\star)\mathrm{d}\bgamma_{c}^\star}
	{\displaystyle\int\cdots\int h_K(\bgamma_{c}^\star:c\in\mathcal{C}_{n-1}\cup\mathcal{C}_\varnothing)\prod_{c\in\mathcal{C}_\varnothing}p_{\bmu}(\bgamma_c^\star)\mathrm{d}\bgamma_c^\star},\nonumber\\
	G_K(A)&\propto&\iint_{A}L_K(\bgamma_{\underline{c}}^\star)\phi(\by_n\mid\bgamma_{\underline{c}}^\star,\bGamma_{\underline{c}}^\star)p_\bmu(\bgamma_{\underline{c}}^\star)p_\bSigma(\bGamma_{\underline{c}}^\star)\mathrm{d}\bgamma_{\underline{c}}^\star
	\mathrm{d}\bGamma_{\underline{c}}^\star,\nonumber\\
	L_K(\bgamma_{\underline{c}}^\star)&=&\int\cdots\int h_K(\bgamma_c^\star:c\in\mathcal{C}_{n-1}\cup\mathcal{C}_\varnothing)\prod_{c\in\mathcal{C}_\varnothing,c\neq\underline{c}}p_\bmu(\bgamma_c^\star)\mathrm{d}\bgamma_{c}^\star,\nonumber
	\end{eqnarray}
	and $h_K(\bgamma_c:c\in\calC_{n-1}\cup\calC_\varnothing)=h_K(\bgamma_{c_1}^\star,\cdots,\bgamma_{c_K}^\star)$ if one labels $\calC_{n-1}\cup\calC_\varnothing$ as $\{c_1,\cdots,c_K\}$. 
\end{theorem}
\textbf{Theorem \ref{thm:urn_model}} is instructive for deriving the blocked-collapsed Gibbs sampler for the posterior inference of the proposed RGM model. 
We follow the notation in \textbf{Theorem \ref{thm:urn_model}}. Let $\mathcal{C}_{-i}$ be the partition induced by $\btheta_{-i}:=(\btheta_1,\cdots,\btheta_n)\backslash\{\btheta_i\}$, and $(\bgamma_c^\star,\bGamma_c^\star:c\in\calC_{-i})$ be the unique values of $\btheta_{-i}$. Notice that by exchangeability
\begin{align}
&\Pi(\mathcal{C}=\mathcal{C}_{-i}\cup\{\{i\}\}\mid\by_i,\btheta_{-i},\calC_{-i})\propto\left[\frac{V_{n}(|\mathcal{C}_{-i}|+1)\beta}{V_{n}(|\mathcal{C}_{-i}|)}\right]\sum_{K=|\mathcal{C}_{-i}|+1}^{\infty}\alpha_K,\nonumber\\
\label{eqn:Gibbs_sampler_marginal_urn_model}
&\Pi(\mathcal{C}=(\mathcal{C}_{-i}\backslash\{c\})\cup\{c\cup\{i\}\}\mid\by_i,\btheta_{-i},\calC_{-i})\propto\phi(\by_i|\bgamma_c^\star,\bGamma_c^\star)\left(|c|+\beta\right),
\end{align}
where $c\in\mathcal{C}_{-i}$. Namely, given a partition $\mathcal{C}_{-i}$ on $\{1,\cdots,n\}\backslash\{i\}$, the left-out index $i$ forms a new singleton cluster with probability proportional to $
\left[\frac{V_{n}(|\mathcal{C}_{-i}|+1)\beta}{V_{n}(|\mathcal{C}_{-i}|)}\right]\sum_{K=|\mathcal{C}_{-i}|+1}^{\infty}\alpha_K,\nonumber$
and is merged into an existing cluster $c\in\mathcal{C}_{-i}$ with probability proportional to $\phi(\by_i|\bgamma_c^\star,\bGamma_c^\star)\left(|c|+\beta\right)$.

Instead of directly sampling from the above categorical distribution, which involves computing the intractable $\alpha_K$'s, we take advantage of the integral structure of $\alpha_K$ and design auxiliary variables following the data augmentation technique in \cite{neal2000markov}. Roughly speaking, when sampling from $p(x,y)$ via MCMC, one introduces an auxiliary variable $z$ and samples   $p(z\mid x,y)$, $p(y\mid x,z)$, and $p(x\mid z)$ alternately (collapsing). 
The auxiliary $z$ is discarded after such an update. \cut
\begin{theorem}\label{thm:Gibbs_sampler_auxiliary_variable}
Using above notations, we denote 
\begin{align}
\widetilde{G}(A\mid\calC_{-i},\btheta_{-i})=\sum_{K=|\calC_{-i}|+1}^\infty p(K\mid\calC=\calC_{-i}\cup\{\{i\}\})
\frac
{\displaystyle \iint_A L_K(\bgamma_{\underline{c}}^\star)p_\bmu(\bgamma_{\underline{c}}^\star)p_\bSigma(\bGamma_{\underline{c}}^\star)\mathrm{d}\bgamma_{\underline{c}}^\star\mathrm{d}\bGamma_{\underline{c}}^\star}
{\displaystyle \int L_K(\bgamma_{\underline{c}}^\star)p_\bmu(\bgamma_{\underline{c}}^\star)\mathrm{d}\bgamma_{\underline{c}}^\star},\nonumber
\end{align}
where $L_K$ is defined in \textbf{Theorem \ref{thm:urn_model}}. Let $\widetilde{g}(\bgamma_{\underline{c}}^\star,\bGamma_{\underline{c}}^\star|\calC_{-i},\btheta_{-i})$ be the density of $\widetilde{G}(\cdot|\calC_{-i},\btheta_{-i})$ and the density of auxiliary variable $(\bgamma_{\underline{c}}^\star,\bGamma_{\underline{c}}^\star)$ be  
\begin{align}
&p(\bgamma_{\underline{c}}^\star,\bGamma_{\underline{c}}^\star\mid\by_i,\calC_{-i},\btheta_{-i})\nonumber\\
&\qquad=\frac{\displaystyle
\left[\frac{V_n(|\calC_{-i}|+1)\beta}{V_n(|\calC_{-i}|)}\right]\phi(\by_i\mid\bgamma_{\underline{c}}^\star,\bGamma_{\underline{c}}^\star)+\sum_{c\in\calC_{-i}}(|c|+\beta)\phi(\by_i\mid\bgamma_c^\star,\bGamma_c^\star)
}
{\displaystyle
\left[\frac{V_n(|\calC_{-i}|+1)\beta}{V_n(|\calC_{-i}|)}\right]\sum_{K=|\calC_{-i}|+1}^\infty\alpha_K+\sum_{c\in\calC_{-i}}(|c|+\beta)\phi(\by_i\mid\bgamma_{c}^\star,\bGamma_c^\star)
}\widetilde{g}(\bgamma_{\underline{c}}^\star,\bGamma_{\underline{c}}^\star\mid\bgamma_c^\star,\bGamma_c^\star,c\in\calC_{-i})
\nonumber.
\end{align}
Given the auxiliary variable $(\bgamma_{\underline{c}}^\star,\bGamma_{\underline{c}}^\star)$, suppose $\calC$ and $\btheta_n$ are sampled as follows:
\begin{align}
\label{eqn:full_conditional_PXDA1}
&\mathbb{P}(\mathcal{C}=\mathcal{C}_{-i}\cup\{\{i\}\}\mid\bgamma_{\underline{c}}^\star,\bGamma_{\underline{c}}^\star,\by_i,\calC_{-i},\btheta_{-i})\propto\left[\frac{V_{n}(|\mathcal{C}_{-i}|+1)\beta}{V_{n}(|\mathcal{C}_{-i}|)}\right]\phi(\by_i\mid\bgamma_{\underline{c}}^\star,\bGamma_{\underline{c}}^\star),\\
\label{eqn:full_conditional_PXDA2}
&\mathbb{P}(\mathcal{C}=(\mathcal{C}_{-i}\backslash\{c\})\cup\{c\cup\{i\}\}\mid\bgamma_{\underline{c}}^\star,\bGamma_{\underline{c}}^\star,\by_i,\calC_{-i},\btheta_{-i})\propto\left(|c|+\beta\right)\phi(\by_i\mid\bgamma_c^\star,\bGamma_c^\star),\\
&\mathbb{P}\left(\btheta_i\in A\mid\calC=\calC_{-i}\cup\{\{i\}\},\bgamma_{\underline{c}}^\star,\bGamma_{\underline{c}}^\star,\by_i,\calC_{-i},\btheta_{-i}\right)=\delta_{(\bgamma_{\underline{c}}^\star,\bGamma_{\underline{c}}^\star)}(A)\nonumber,\\
&\mathbb{P}\left(\btheta_i\in A\mid\calC=(\calC_{-i}\backslash\{c\})\cup(\{c\cup\{i\}\}),\bgamma_{\underline{c}}^\star,\bGamma_{\underline{c}}^\star,\by_i,\calC_{-i},\btheta_{-i}\right)=\delta_{(\bgamma_{c}^\star,\bGamma_{c}^\star)}(A)\nonumber.
\end{align}
Then the marginal posterior $(\btheta_i|\by_i,\calC_{-i},\btheta_{-i})$ with $(\bgamma_{\underline{c}}^\star,\bGamma_{\underline{c}}^\star)$ and $\calC|\calC_{-i}$ integrated out coincides with \eqref{eqn:Gibbs_sampler_urn_model}, and the complete conditional distribution of $(\bgamma_{\underline{c}}^\star,\bGamma_{\underline{c}}^\star)$ is given by
\begin{align}
&\mathbb{P}((\bgamma_{\underline{c}}^\star,\bGamma_{\underline{c}}^\star)\in A\mid\by_i,\calC,\btheta_{-i},\btheta_i)
\nonumber\\
\label{eqn:Gibbs_sampling_auxiliary_variables}
&\qquad
=\mathbb{I}(\calC = \calC_{-i}\cup\{\{i\}\})\delta_{\btheta_i}(A)+\mathbb{I}(\calC\neq \calC_{-i}\cup\{\{i\}\})\widetilde{G}(A\mid\calC_{-i},\btheta_{-i}).
\end{align}
\end{theorem}
The proof of \textbf{Theorem \ref{thm:Gibbs_sampler_auxiliary_variable}} is deferred to Section \ref{sec:proof_of_the_generalized_urn_model} of the Supplementary Material. Now we are in a position to introduce the blocked-collapsed Gibbs sampler for the posterior inference. We remark that this Gibbs sampler can also be regarded as the generalization of the ``Algorithm 8'' in \cite{neal2000markov} to the case where a repulsive prior among component centers is introduced. The basic idea is to draw samples from $\mathbb{P}(\calC_n\mid\bgamma_{\underline{c}}^\star,\bGamma_{\underline{c}}^\star,\by_i,\calC_{-i},\btheta_{-i})$, $\mathbb{P}(\btheta_i\mid\calC_n,\bgamma_{\underline{c}}^\star,\bGamma_{\underline{c}}^\star,\by_i,\calC_{-i},\btheta_{-i})$, and $\mathbb{P}(\bgamma_{\underline{c}}^\star,\bGamma_{\underline{c}}^\star\mid\calC_n,\btheta_i,\by_i,\calC_{-i},\btheta_{-i})$ alternately, where $(\bgamma_{\underline{c}}^\star,\bGamma_{\underline{c}}^\star)$ is the auxiliary variable introduced in \textbf{Theorem \ref{thm:Gibbs_sampler_auxiliary_variable}}. 

\begin{algo}
Suppose the current state of the Markov chain consists of $(\bgamma_c^\star,\bGamma_c^\star:c\in\mathcal{C}_n)$, and a partition $\mathcal{C}_n$ on $\{1,\cdots,n\}$. We instantiate $(\btheta_1,\cdots,\btheta_n)$ using $(\bgamma_c^\star,\bGamma_c^\star:c\in\mathcal{C}_n)$ and $\mathcal{C}_n$ by letting $\btheta_{z_i}=(\bgamma_c^\star,\bGamma_c^\star)$ if $i\in c$. 
A complete iteration of the blocked-collapsed Gibbs sampler is desribed as below.

\noindent
\begin{itemize}[noitemsep, topsep = 0cm]
	\item\textbf{Step 1: For }$i=1,\cdots,n$:
	\begin{enumerate}[noitemsep, topsep = 0cm]
		\item\textbf{Sample auxiliary variable $(\bgamma_{\underline{c}}^\star,\bGamma{}_{\underline{c}}^\star)$ from \eqref{eqn:Gibbs_sampling_auxiliary_variables}}: If $\mathcal{C}_n=\mathcal{C}_{-i}\cup\{\{i\}\}$, then set $(\bgamma_{\underline{c}}^\star,\bGamma_{\underline{c}}^\star) = \btheta_i$; Otherwise sample $(\bgamma_{\underline{c}}^\star,\bGamma_{\underline{c}}^\star)$ from $\widetilde{G}(\cdot\mid\calC_{-i},\btheta_{-i})$ as follows:
		\begin{itemize}[noitemsep, topsep = 0cm]
			\item[i)]Sample $K\sim p(K\mid\mathcal{C}_n=\mathcal{C}_{-i}\cup\{\{i\}\})$, set $\ell=|\mathcal{C}_{-i}|$, compute $\mathcal{C}_\varnothing$ with $|\mathcal{C}_\varnothing|=K-\ell$, and set $\btheta_{-i}=(\btheta_1,\cdots,\btheta_n)\backslash\{\btheta_i\}$.
			\item[ii)]Sample $\bGamma_{\underline{c}}^\star\sim p_\bSigma(\bGamma_{\underline{c}}^\star)$. Sample $(\bgamma_{c}^\star:c\in\mathcal{C}_\varnothing)$ by accept-reject sampling: Sample $(\bgamma_c^\star:c\in\mathcal{C}_\varnothing)$ independently from  $p_\bmu$ and $U\sim\mathrm{Unif}(0,1)$, independent of $(\bgamma_c^\star:c\in\mathcal{C}_\varnothing)$; If $U<h_K(\bgamma_c^\star:c\in\mathcal{C}_{-i}\cup\mathcal{C}_\varnothing)$, then accept the new proposed samples; Otherwise resample $(\bgamma_c^\star:c\in\mathcal{C}_\varnothing)$ from $p_\bmu$ and $U$ until $U<h_K(\bgamma_c^\star:c\in\mathcal{C}_{-i}\cup\mathcal{C}_\varnothing)$. Discard all $(\bgamma_c^\star,\bGamma_c^\star:c\in\calC_\varnothing\backslash\{\underline{c}\})$.
		\end{itemize}
		\item\textbf{Sample $\mathcal{C}_n$ from $p(\mathcal{C}_n|,\bgamma_{\underline{c}}^\star,\bGamma_{\underline{c}}^\star,\by_i,\calC_{-i},\btheta_{-i})$ according to \eqref{eqn:full_conditional_PXDA1} and \eqref{eqn:full_conditional_PXDA2}:}
		\begin{align}
		&\Pi(\mathcal{C}_n=\mathcal{C}_{-i}\cup\{\{i\}\}\mid-)\propto\left[\frac{V_{n}(|\mathcal{C}_{-i}|+1)\beta}{V_{n}(|\mathcal{C}_{-i}|)}\right]\phi(\by_i\mid\bgamma_{\underline{c}}^\star,\bGamma_{\underline{c}}^\star)\nonumber,\\
		&\Pi(\mathcal{C}_n=(\mathcal{C}_{-i}\backslash\{c\})\cup\{c\cup\{i\}\}\mid-)\propto\left(|c|+\beta\right)\phi(\by_i\mid\bgamma_c^\star,\bGamma_c^\star)\nonumber.
		\end{align}
		\item \textbf{Assign $\btheta_i$ value according to $\mathbb{P}(\btheta_i\in\cdot\mid\calC,\bgamma_{\underline{c}}^\star,\bGamma_{\underline{c}}^\star,\by_i,\calC_{-i},\btheta_{-i})$.}
		Set $\btheta_i=(\bgamma_{\underline{c}}^\star,\bGamma_{\underline{c}}^\star)$ if $\calC_n = \calC_{-i}\cup\{\{i\}\}$, and set $\btheta_i=(\bgamma_c^\star,\bGamma_c^\star)$ if $\calC_n=(\calC_{-i}\backslash\{c\})\cup(\{c\cup\{i\}\})$ for some $c\in\calC_{-i}$. 
	\end{enumerate}
	\item\textbf{Step 2: }\textbf{Sample $K$ from $p(K\mid \mathcal{C}_n,\by_1,\cdots,\by_n,\bGamma_c^\star:c\in\calC_n)$}; Set $\ell=|\mathcal{C}_n|$, and compute $\mathcal{C}_\varnothing$ such that $|\mathcal{C}|=K-\ell$.
	\item\textbf{Step 3: }\textbf{Sample $(\bGamma_c^\star:c\in\mathcal{C}_{n})$ from $p(\bGamma_c^\star\mid\by_i:i\in c,\bgamma_c^\star,\calC_n)$}: For all $c\in\mathcal{C}_n$, sample $\bGamma_c^\star$ from
	\begin{eqnarray}
	p(\bGamma_{c}^\star\mid-)\propto p_{\bSigma}(\bSigma_{c}^\star)\prod_{i\in c}\phi(\by_i\mid\bgamma_c^\star,\bSigma_{c}^\star)\nonumber.
	\end{eqnarray}
	\item\textbf{Step 4 (Blocking): }\textbf{Sample $(\bgamma_c^\star:c\in\mathcal{C}_n)$ from $p(\bgamma_c^\star:c\in\calC_n\mid K,\bGamma_c^\star,\by_1,\cdots,\by_n,\calC_n)$}. This can be done by accept-reject sampling: For each $c\in\mathcal{C}_n$, sample 
	$$
	p(\bgamma_c^\star\mid-)\propto p_\bmu(\bgamma_c^\star)\prod_{i\in c}\phi(\by_i\mid\bgamma_c^\star,\bGamma_c^\star),
	$$
	and for each $c\in\mathcal{C}_\varnothing$, sample $\bgamma_c^\star\sim p_\bmu(\bgamma_c^\star)$. 
	Next independently sample $U\sim\mathrm{Unif}(0,1)$; If $U<h_K(\bgamma_c^\star:c\in\mathcal{C}_{n}\cup\mathcal{C}_\varnothing)$, then accept the new proposed samples; Otherwise resample $(\bgamma_{c}^\star:c\in\calC_n\cup\calC_\varnothing)$ and $U$ until $U<h_K(\bgamma_c^\star:c\in\calC_n\cup\calC_\varnothing)$. 
	\item \textbf{Step 5: }Change the current state to $(\btheta_c^\star,c\in\mathcal{C}_n)$ and $\mathcal{C}_n$.
\end{itemize}
\end{algo}
The detailed implementation of the blocked-collapsed Gibbs sampler, including the discussion of sampling from $p(K\mid\calC_n)$ and $p(K\mid\calC_n,\by_1,\cdots,\by_n,\bGamma_c^\star:c\in\calC_n)$, is provided in Section \ref{sec:details_of_posterior_inference} of the Supplementary Material. 
\begin{remark}
It is worth noticing that in theory, only \textbf{Step 1} in the above Gibbs sampler is necessary to create a Markov chain with the stationary distribution being the full posterior distribution. Nevertheless, such an urn-model-based sampler could potentially yield a Markov chain converging rather slowly, as has been pointed out in \cite{neal2000markov}. The resampling steps (\textbf{Step 2} through \textbf{Step 5}) are hence introduced to improve the mixing of the chain.
\end{remark}
\begin{remark}
	The proposed sampler can be easiliy extended to the case where a non-Gaussian mixture model is used, provided that we use priors $p_\bmu,p_\bSigma$ in \eqref{eqn:repulsive_density} that are conjugate to the non-Gaussian kernel density. In cases where non-conjugate priors $p_\bmu,p_\bSigma$ are used, it is also possible to extend the blocked-collapsed Gibbs sampler either by a method of ``no-gaps'' proposed by \cite{maceachern1998estimating} or a Metropolis-within-Gibbs sampler \citep{neal2000markov}. 
\end{remark}

\section{Numerical Examples}
\label{sec:numerical_results}
We evaluate the performance of the RGM model and the blocked-collapsed Gibbs sampler proposed in Section \ref{sec:posterior_inference} through extensive simulation studies and real data analysis. 
Subsections \ref{sub:multi_modal_estimation_finite_gaussian_mixtures} and \ref{sub:uni_modal_density_continuous_gaussian_mixtures} aim to illustrate the advantages of the RGM concerning accurate density estimation, identification of correct number of components, 
and shrinkage effect on the model complexity. Subsection \ref{sub:multivariate_model_based_clustering} demonstrates the efficiency of the proposed blocked-collapsed Gibbs sampler compared to the DP mixture model and the DPP mixture model \citep{xu2016bayesian}. In Subsection \ref{sub:real_data_analysis} we apply the RGM model to analyze the Old Faithful geyser eruption data \citep{silverman1986density}. 
We assume $\beta=1$, indicating a uniform prior on $(w_1,\cdots,w_K\mid K)$. We assign a zero-truncated Poisson prior on $K$ with intensity $\lambda=1$ (\emph{i.e.}, $p(K)=\frac{\mathbb{I}(K\geq 1)}{(\mathrm{e}-1)K!}$) 
for all numerical examples except the location-mixture problem in Section \ref{sub:uni_modal_density_continuous_gaussian_mixtures}. 
The repulsive function is defined as $g(x)=\frac{x}{g_0+x}$ for some $g_0>0$, and without loss of generality, we let $h_K$ to be of the form \eqref{eqn:repulsive_function_1}. 
Lastly, we assume $p(\bmu)=\phi(\bmu\mid\mathrm{0},\tau^2\eye_p)$ and a truncated inverse Gamma prior on $\lambda(\bSigma)$, $p(\lambda)\propto\mathbb{I}(\lsigma^2\leq \lambda\leq \usigma^2)\lambda^{-a_0-1}\exp(-b_0/\lambda)$ for some $a_0,b_0>0$. 

We give the convergence diagnostics via trace plots and autocorrelation plots in Section \ref{sec:convergence_diagnostics} of the Supplementary Material. To compare the performance of the proposed models with the competitors (\emph{e.g.} the DP mixture (DPM) model and the DPP mixture model), we follow the ideas in \cite{pettit1990conditional} and compute the \emph{logarithm of the conditional predictive ordinate} (log-CPO) of different models using the post-burn-in samples as follows: 
\begin{eqnarray}
\text{log-CPO}=-\sum_{i=1}^n\log\left[\frac{1}{n_{\mathrm{mc}}}\sum_{i_{\mathrm{it}}=1}^{n_{\mathrm{mc}}}p(\by_i\mid\boldsymbol{\Theta}_{\mathrm{mc}}^{i_{\mathrm{it}}})\right],\nonumber
\end{eqnarray}
where $n_{\mathrm{mc}}$ is the number of the post-burn-in MCMC samples, $i_{\mathrm{it}}$ indexes the post-burn-in iterations, and $\bm{\Theta}_{\mathrm{mc}}^{i_{\mathrm{it}}}$ represents the post-burn-in samples of all parameters generated by the MCMC at the $i_{\mathrm{it}}$th iteration. 

\subsection{Fitting Multi-modal Density: Finite Gaussian Mixtures} 
\label{sub:multi_modal_estimation_finite_gaussian_mixtures}

In this subsection, to demonstrate multi-modal density fitting, we fit a finite mixture of Gaussians using the RGM model, and evaluate its performance regarding the density estimation and the identification of the number of components. In particular, suppose the simulated data $\by_1,\cdots,\by_{n}$, $n = 1000$, are i.i.d. generated from the bivariate density:
\begin{eqnarray}
f_0(\by)=0.4\phi(\by\mid\mathbf{0},\mathrm{diag}(2,1)) + 0.3\phi(\by\mid(-6,-6)\transpose, 3\eye_2) + 0.3\phi(\by\mid(6,6)\transpose, 2\eye_2)\nonumber.
\end{eqnarray}
We implement the proposed blocked-collapsed Gibbs sampler with $g_0 = 10$, $\tau = 10$, $m = 2,\lsigma=0.1, \usigma = 10$, and a total number of 2000 iterations with the first $1000$ iterations discarded as burn-in.
For comparison, we consider the following DPM model,
\begin{eqnarray}
(\by_i\mid\bmu_{z_i},\bSigma_{z_i})\sim N(\bmu_{z_i},\bSigma_{z_i}),\quad (\bmu_{z_i},\bSigma_{z_i}\mid G)\iidsim G,\quad\text{and }(G\mid\alpha, G_0)\sim\mathrm{DP}(\alpha, G_0), \nonumber
\end{eqnarray}
where $G_0=\mathrm{N}(\bmu, \bSigma)$ with $\bmu \sim \mathrm{N}\left(\mb_1,\bSigma/k_0\right)$ and $\bSigma\sim\text{Inv-Wishart}(4,\boldsymbol{\Psi}_1)$, $\alpha\sim\mathrm{Gamma}(1, 1)$, $\mb_1\sim\mathrm{N}(\mathbf{0},2\eye_2)$, $k_0\sim\mathrm{Gamma}(0.5,0.5)$, and $\boldsymbol{\Psi}_1\sim\text{Inv-Wishart}(4, 0.5\eye_2)$. For the DP mixture model, we use $K$ to represent the number of clusters throughout this section, since the number of components is always infinity. 

Table \ref{table:log_cpo_numerical_results} shows that the log-CPO of the RGM model is higher than that of the DPM model, indicating that RGM is preferred according to the data. 
Figures \ref{multi_modal_figure}a and \ref{multi_modal_figure}c show the posterior density estimation under the RGM model and the DP mixture model, respectively, indicating that both methods perform well in terms of density estimation. 
\begin{table}[htbp]
	\centering
	\caption{Log-Conditional Predictive Ordinate (log-CPO) for Numerical Results}
	\begin{tabular}{c|c|c|c}
		\hline\hline
		Model & Subsection \ref{sub:multi_modal_estimation_finite_gaussian_mixtures} & Subsection \ref{sub:uni_modal_density_continuous_gaussian_mixtures} & 
		Subsection \ref{sub:real_data_analysis}
		 \\
		\hline
		RGM model&-3596.525& -3385.989& -240.2669\\
		DPM model&-4599.204& -3483.667& -315.1032\\
		DPP mixture model&& & -512.6564\\
		\hline\hline
	\end{tabular}%
	\label{table:log_cpo_numerical_results}
\end{table}%

\begin{figure}[htbp!]
	\centerline{\includegraphics[width=1\textwidth]{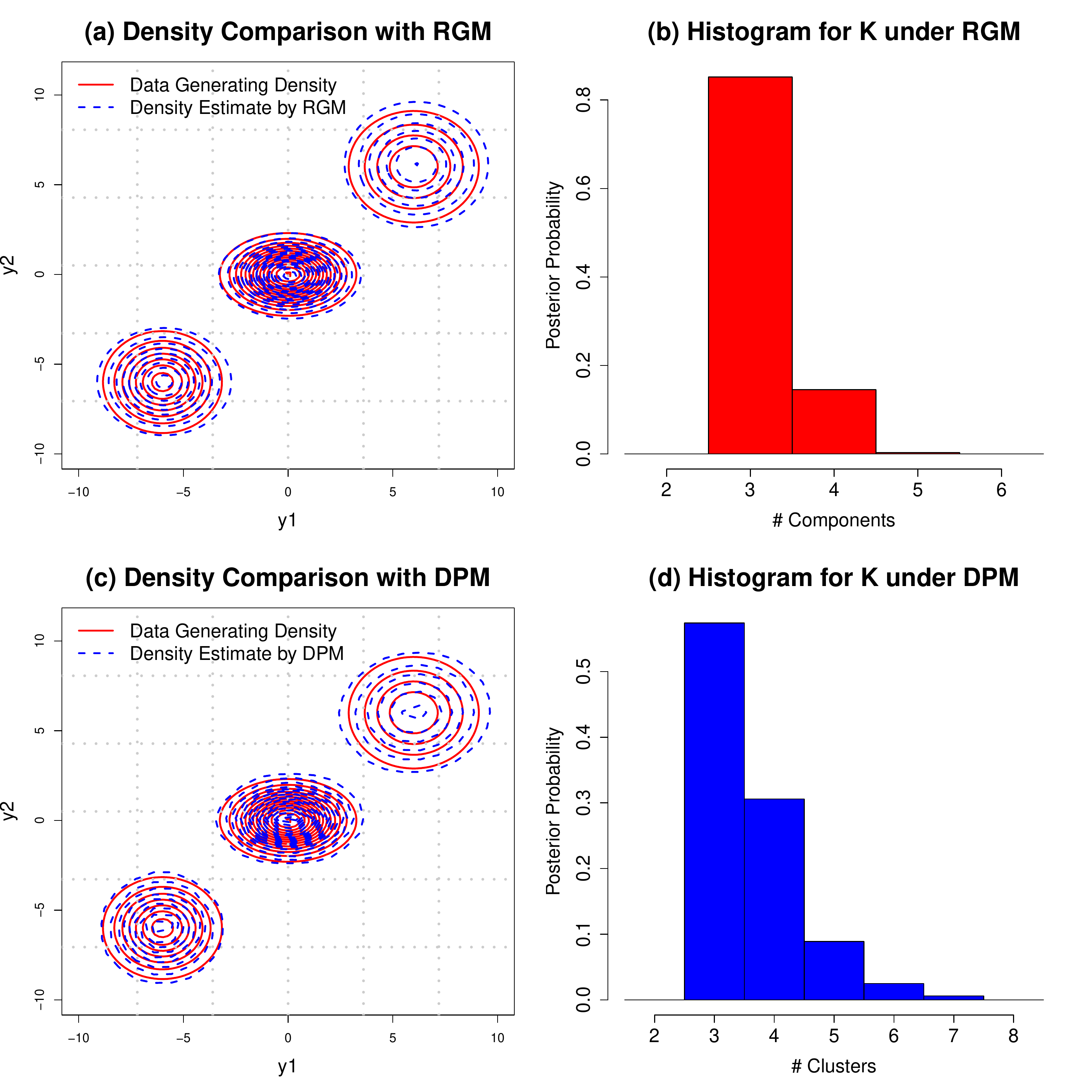}}
	\caption{Fiting Multi-modal Density: Panels (a) and (c) are the contour plots for the posterior density estimation of the RGM model and the DPM model, respectively. Panels (b) and (d) are the histograms of the posterior number of components under the RGM model and the posterior number of clusters under the DPM model, respectively, where the underlying true number of components is $K = 3$. 
	}
	\label{multi_modal_figure}
\end{figure} 
However, as shown in the histograms of the posterior numbers of components/clusters in 
Figures \ref{multi_modal_figure}b and \ref{multi_modal_figure}d,  the posterior distribution of the number of components is highly concentrated around the underlying true $K$ under the RGM model, whereas the DPM model assigns relatively higher posterior probability to redundant clusters. This agrees with the inconsistency phenomenon of the DPM model for the identification of number of components, which is reported in \cite{miller2013simple}. 

\subsection{Fitting Uni-modal Density: Continuous Gaussian Mixtures}  
\label{sub:uni_modal_density_continuous_gaussian_mixtures}
Besides generating the simulated data from a finite discrete Gaussian mixture model, 
in this subsection we consider a continuous mixture of Gaussians, 
\begin{eqnarray}\label{eqn:norm_convolve_exp}
f_0(y_1,y_2)=\prod_{j=1}^2\int_0^{\infty}\phi(y_i-\mu_i-\mu_0\mid 0,1)\exp(-\mu_i)\mathrm{d}\mu_i.
\end{eqnarray}
Notice that $f_0$ is uni-modal. The random variables $y_i$, $i=1,2$ can be i.i.d. generated as the sum of a normal random variable and an exponential random variable with intensity parameter $1$, \emph{i.e.}, $y_i=z_i+\mu_i$ where $z_i\sim\mathrm{N}(\mu_0,1)$ and $\mu_i\sim\mathrm{Exp}(1)$, $i=1,2$. Then $\by=(y_1,y_2)$ is the random vector following the distribution in \eqref{eqn:norm_convolve_exp}. The marginal distribution of $y_i$ is  referred to as the \emph{exponentially modified Gaussian} (EMG) distribution, the density of which can be alternatively represented as
$f(y)=\frac{1}{2}\exp\left(\mu_0-y+\frac{1}{2}\right)\mathrm{erfc}\left(\frac{\mu_0+1-y}{\sqrt{2}}\right)$, 
where $\mathrm{erfc}$ is the well-known complementary error function
$\mathrm{erfc}(x)=\frac{2}{\sqrt{\pi}}\int_x^{\infty}\exp(-t^2)\mathrm{d}t$.
We generate $n=1000$ i.i.d. samples from $f_0$ with $\mu_0=-4$, and implement the proposed blocked-collapsed Gibbs sampler with $g_0 = 7$, $\tau = 10$, $m = 2, \lsigma=0.1, \usigma = 10$, and a total number of 2000 iterations with the first $1000$ iterations discarded as burn-in phase. For comparison, we consider the similar DPM model with the same setting as in Section \ref{sub:multi_modal_estimation_finite_gaussian_mixtures}. 

Figures \ref{uni_modal_figure}a and \ref{uni_modal_figure}c show that the RGM model and the DPM model provide similar accurate density estimation to the underlying true density $f_0$. 
\begin{figure}[htbp!]
	\centerline{\includegraphics[width=1\textwidth]{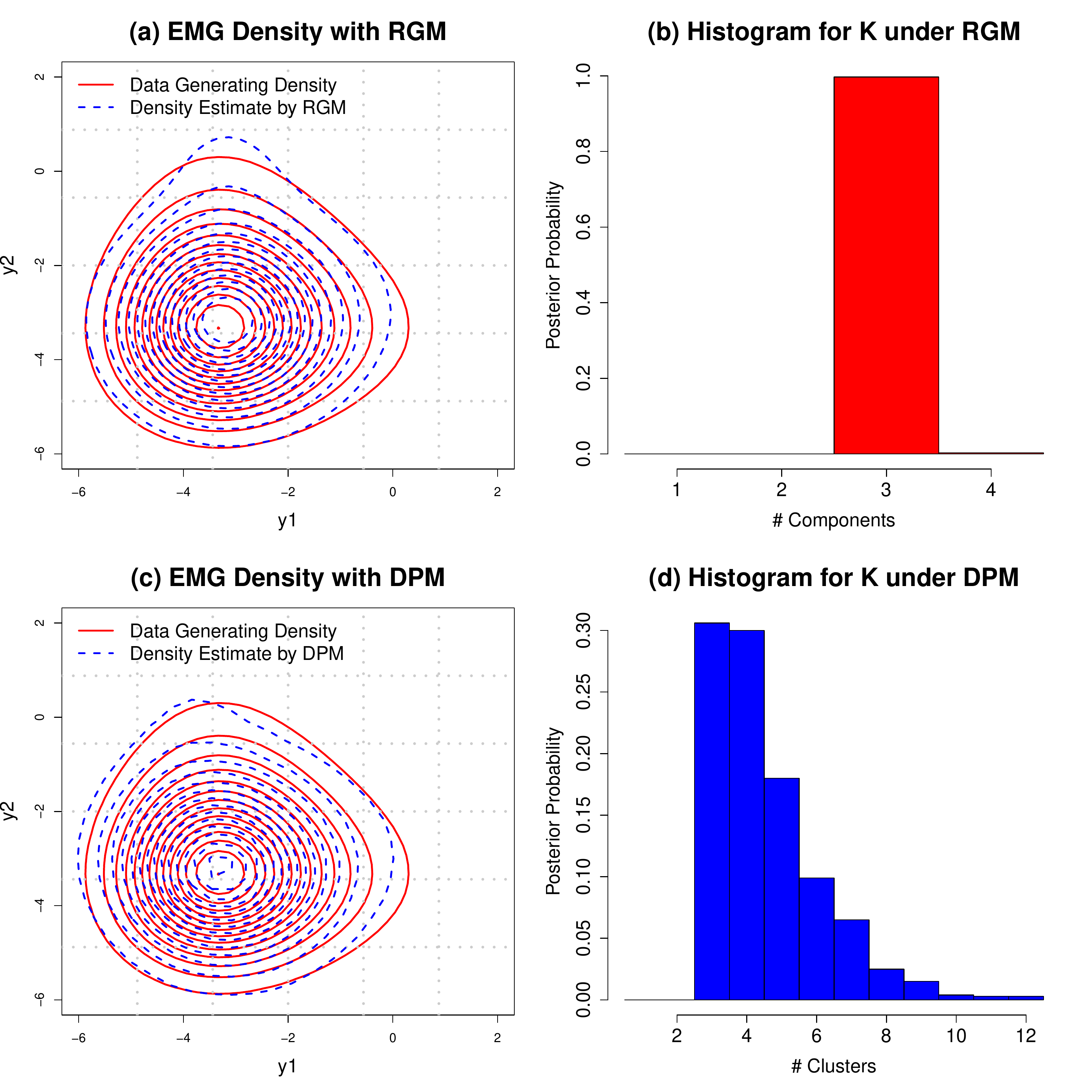}}
	\caption{Fitting Uni-modal Density: Panels (a) and (c) are the contour plots for the posterior density estimation under the RGM model and the DPM model, respectively. Panels (b) and (d) are the histograms of the posterior number of components under the RGM model and the posterior number of clusters under the DPM model, respectively.
	}
	\label{uni_modal_figure}
\end{figure}
However, Figures \ref{uni_modal_figure}b and \ref{uni_modal_figure}d indicate that under the DPM model, the number of active components tends be larger than that under the RGM model in order to fit the data well. In other words, the posterior of the RGM model provides the same level of accuracy in density estimation as the DPM model does, but with less number of components. In this simulation study, with high posterior probability, the RGM model only utilizes $3$ components to fit the density, whereas the DPM model assigns large posterior probability to utilizing $4$ or more components. The log-CPO comparison in Table \ref{table:log_cpo_numerical_results}, clearly show that the RGM model outperforms the DPM model. 

To demonstrate the parsimony effect on the number $K$ of necessary components to fit the density well, we perform comparison between the RGM and the independent-prior MFM. Suggested by \textbf{Theorem 5}, we consider location-mixture problem here only. That is, the covariance matrices for all components under both RGM and MFM are fixed at $\bSigma_k=\eye_2$, $k=1,\cdots,K$. 
We use the prior $p(K)\propto \frac{Z_K\mathbb{I}(K\geq 1)}{K!}$ for the RGM, and $p(K)\propto\frac{\mathbb{I}(K\geq 1)}{K!}$ for the MFM. 
We implement the proposed blocked-collapsed Gibbs sampler with $\tau = 10$, $m = 2$, $g_0=7$ for the location-RGM, $g_0=0$ for the MFM, and a total number of 2000 iterations with the first $1000$ iterations discarded as burn-in phase. 

Since the data generating density is a continuous mixture of Gaussians, there is no ``ground true'' $K$. We evaluate the two methods in terms of the posterior of $K$ and the log-CPO values. Figures \ref{uni_modal_compare_figure}a and \ref{uni_modal_compare_figure}c show that the location-RGM and the MFM provide similar accurate density estimation to the underlying true density $f_0$ and yield similar log-CPO. 
\begin{figure}[htbp!]
	\centerline{\includegraphics[width=1\textwidth]{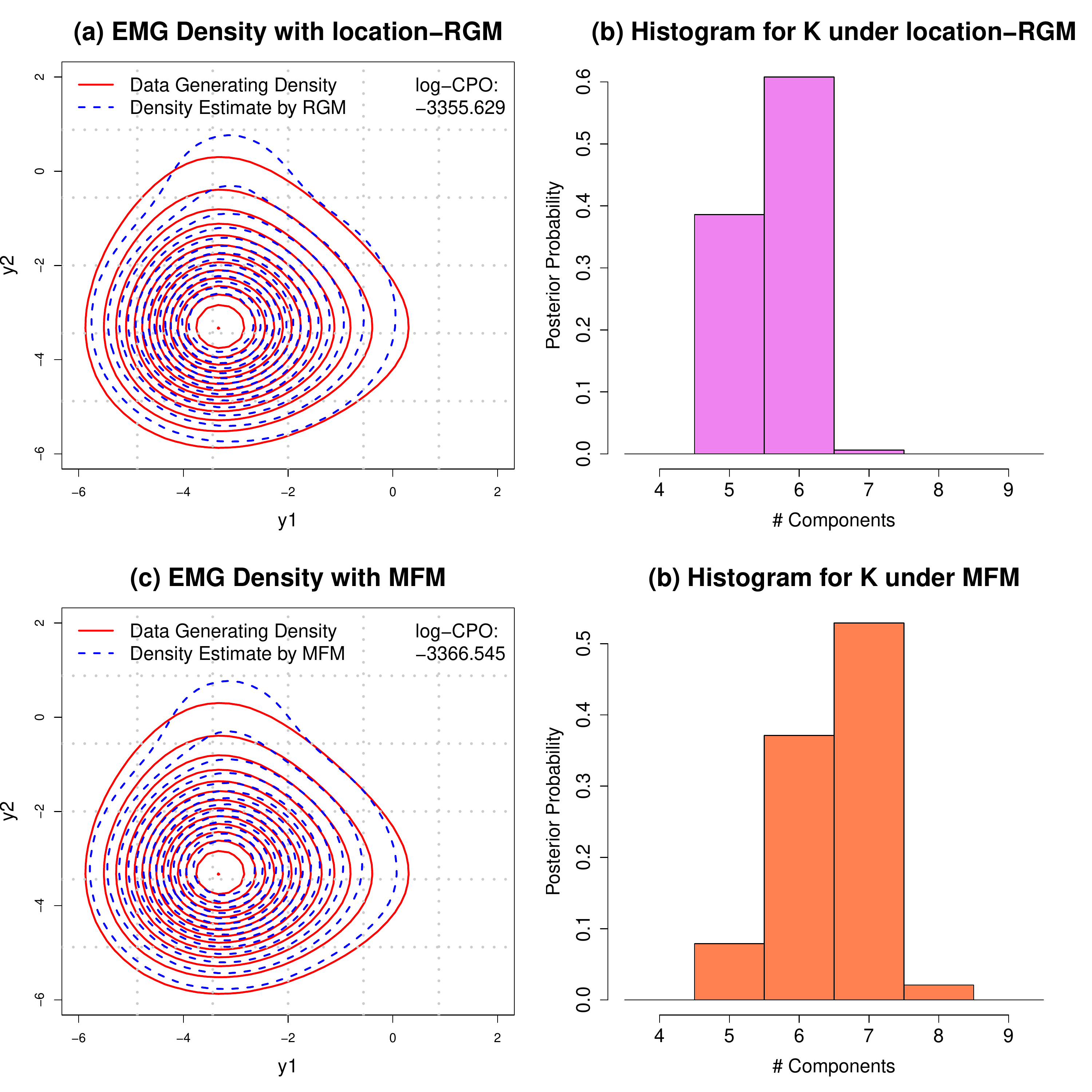}}
	\caption{Fitting Uni-modal Density using Location-Mixtures only: Panels (a) and (c) are the contour plots for the posterior density estimation under the location-RGM and the  MFM, respectively. Panels (b) and (d) are the histograms of the posterior number of components under the locatio-RGM and MFM, respectively.
	}
	\label{uni_modal_compare_figure}
\end{figure}
Nevertheless, it can be seen from Figures \ref{uni_modal_compare_figure}b and \ref{uni_modal_compare_figure}d that the MFM model assigns larger number components than the location-RGM. This phenomenon also numerically verifies \textbf{Theorem \ref{thm:model_complexity}}: compared to the independent prior ($g_0=0$), the posterior number $K$ of components under the repulsive prior ($g_0>0$) tends to be less.
We also observe that both the location-RGM and the MFM provide similar performance in terms of the density estimation, measured by the log-CPO ($-3355.629$ and $-3366.545$ under the the location-RGM and the MFM, respectively). 


\subsection{Multivariate Model-Based Clustering} 
\label{sub:multivariate_model_based_clustering}


Now we focus on a higher dimensional model-based clustering problem. Suppose that we generate $n = 500$ i.i.d. samples from a mixture of 3 10-dimensional Gaussians:
\begin{eqnarray}
f_0(\by)=0.4\phi(\by\mid\bmu_1,\bSigma_1) + 0.3\phi(\by\mid\bmu_2, 3\eye_{10}) + 0.3\phi(\by\mid\bmu_3, 2\eye_{10}),\nonumber
\end{eqnarray}
where the covariance matrix for the first component is a randomly generated diagonal matrix:
$$
\bSigma_1=\mathrm{diag}(5.5729, 5.0110, 3.6832, 8.1931, 5.7717, 3.0267, 3.5011, 7.8291, 4.2233, 4.3885),
$$ 
and $\bmu_1=\mathbf{0}$, $\bmu_2 = (-6,\cdots,-6)\transpose\in\mathbb{R}^{10}$, $\bmu_3 = -\bmu_2$. 
In this simulation study, we focus on the model-based clustering without fixing the number $K$ of components \emph{a priori}. 
Due to the challenge of visualizing high-dimensional clustering, we only show the scatter plot of the 4th versus 8th coordinate of the simulated data in Figure \ref{model_clustering_figure}a. These two dimensions correspond to the first two largest eigenvalues in the covariance matrix. The projection of the data onto this 2-dimensional subspace shows that the three clusters are not well-separated. We implement the proposed blocked-collapsed Gibbs sampler with $g_0 = 70$, $\tau = 10$, $m = 2, \lsigma=0.1, \usigma = 10$. To demonstrate the efficiency of the proposed sampler, we keep all MCMC samples and compare the efficiency of the algorithms in terms of their numbers of burn-in iterations. 
\begin{figure}[h!]
	\centerline{\includegraphics[width=1\textwidth]{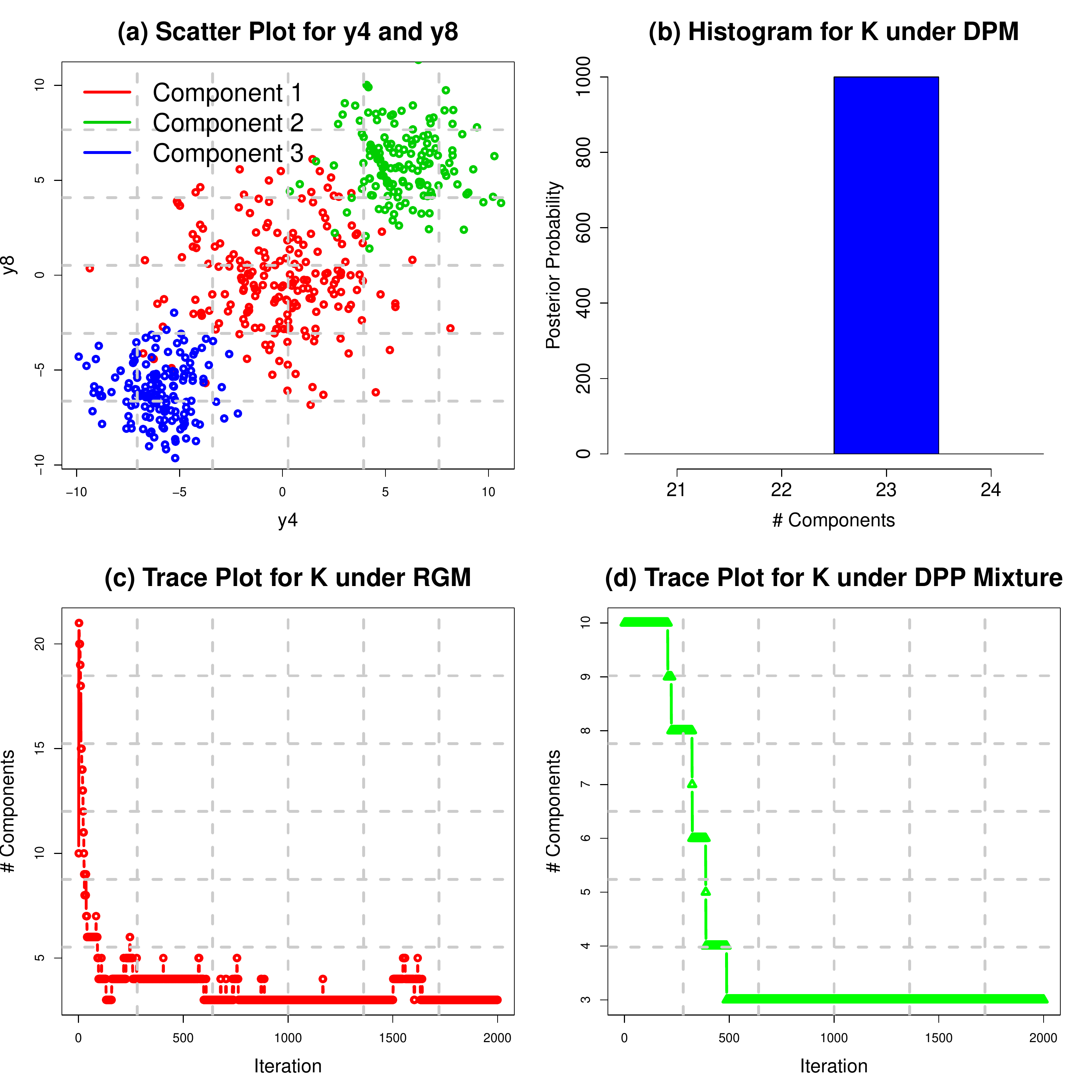}}
	\caption{Multivariate Model-Based Clustering: Panel (a) is the scatter plot of the 4th-versus-8th coordinate of the simulated data; Panel (b) is the histogram of the posterior number of clusters under the DPM model; Panels (c) and (d) are the trace plots for the posterior samples of $K$ under the RGM model, and that of the number of clusters under the DPP mixture model, respectively. }
	\label{model_clustering_figure}
\end{figure}

For comparison, we consider the two alternative clustering models and evaluate their performance in terms of efficiency in estimating posterior number of components. The first one is the {DPM model}: $(\by_i\mid\bmu_{z_i},\bSigma_{z_i})\sim N(\bmu_{z_i},\bSigma_{z_i}),\quad (\bmu_{z_i},\bSigma_{z_i}\mid G)\iidsim G,\quad\text{and }(G\mid\alpha, G_0)\sim\mathrm{DP}(\alpha, G_0)$, where $G_0=\mathrm{N}(\mathbf{\bmu},\bSigma)$ with $\bmu\sim\mathrm{N}(\mathbf{0},\bSigma/k_0)$ and $\bSigma\sim\text{Inv-Wishart}(12,\boldsymbol{\Psi}_1)$, $\alpha=1$, $k_0\sim\mathrm{Gamma}(0.005,0.005)$, and $\boldsymbol{\Psi}_1=0.1\eye_{10}$. The Second alternative model is the \emph{DPP mixture model} proposed in \cite{xu2016bayesian}, who used the determinantal point process as a repulsive function: $h_K(\bmu_1,\cdots,\bmu_K)=\det\left\{\left[\exp\left(-\frac{1}{2\theta^2}\|\bmu_k-\bmu_{k'}\|^2\right)\right]_{K\times K}\right\}$ for $K\geq 2$, $h_K\equiv1$ otherwise.  The posterior inference of the DPP mixture model was performed using a potentially inefficient RJ-MCMC sampler. 
We initialize the Markov chains with $K=10$ for all three models. By comparing the histogram and trace plots of the posterior number of components/clusters in Figures \ref{model_clustering_figure}b, \ref{model_clustering_figure}c, and \ref{model_clustering_figure}d,
we find the DPM model significantly over-estimates the number of components at $23$ in order to fit the 10-dimensional data well; The DPP mixture inferred with RJ-MCMC, though eventually stabilizes at the correct $K = 3$, requires relatively large number of iterations to find the underlying truth (approximately $500$ iterations). In contrast, the posterior number of components under the RGM model highly concentrates around the underlying true $K = 3$, and stabilizes within only $100$ iterations. In terms of efficiency of the Markov chain, the blocked-collapsed Gibbs sampler of the RGM model outperforms the other two alternatives. 

We further report the performance of the model-based clustering procedure under the RGM model. Adopting the ideas in \cite{xu2016bayesian} and \cite{dahl2006model}, we define the association matrix $S\in\{0,1\}^{n\times n}$ with $(i,j)$th entries being $\mathbb{I}(\bgamma_i=\bgamma_j)$, and $H\in\{0,1\}^{n\times n}$ with $(i,j)$th entries being $\mathbb{I}(\bgamma_i=\bgamma_j\mid\by_1,\cdots,\by_n)$. Using the posterior samples, $H$ can be approximated using the posterior mean of $\mathbb{I}(\bgamma_i=\bgamma_j)$ for all $(i,j)$ pairs. 
We compute the mean of the absolute mis-classification matrix $\left(|H_{ij}-S_{ij}|\right)_{n\times n}$. The mis-classification error defined by $\frac{1}{n^2}\|\hat{H}-S\|_{\mathrm{F}}$ is $1.0215\times 10^{-5}$, where $\hat{H}$ is computed using the posterior means. 

\subsection{Old Faithful Geyser Eruption Data} 
\label{sub:real_data_analysis}
In this subsection, we consider the Old Faithful geyser eruption data that record 
the eruption length of the Old faithful geyser in the Yellowstone National Park with the number of observations $n=272$ as a real world example. Following the procedure described in \cite{qin2013maximum, garcia1999robustness}
for each observed eruption duration time, we pair it with the time length of
the next eruption, so that we have a bivariate data of sample size $271$. The points with the ``short followed by short'' eruption property were identified as outliers in \cite{garcia1999robustness}, 
in which
a robust trimmed mean procedure was used to reduce the effects from these outliers. 
Alternatively, we apply the RGM model to analyze the bivariate dataset, and show that the outliers can actually be identified as an extra component. 
We also compare the proposed method with the two alternative models: the {DPM model} and the {DPP mixture model} as described in subsection \ref{sub:multivariate_model_based_clustering}. 

Figure  \ref{old_faithful_figure} shows the predictive densities and the histograms of the number of components/clusters estimated by the three models: the RGM model, the DPM model, and the DPP mixture model.  
\begin{figure}[h!]
	\centerline{\includegraphics[width=.85\textwidth]{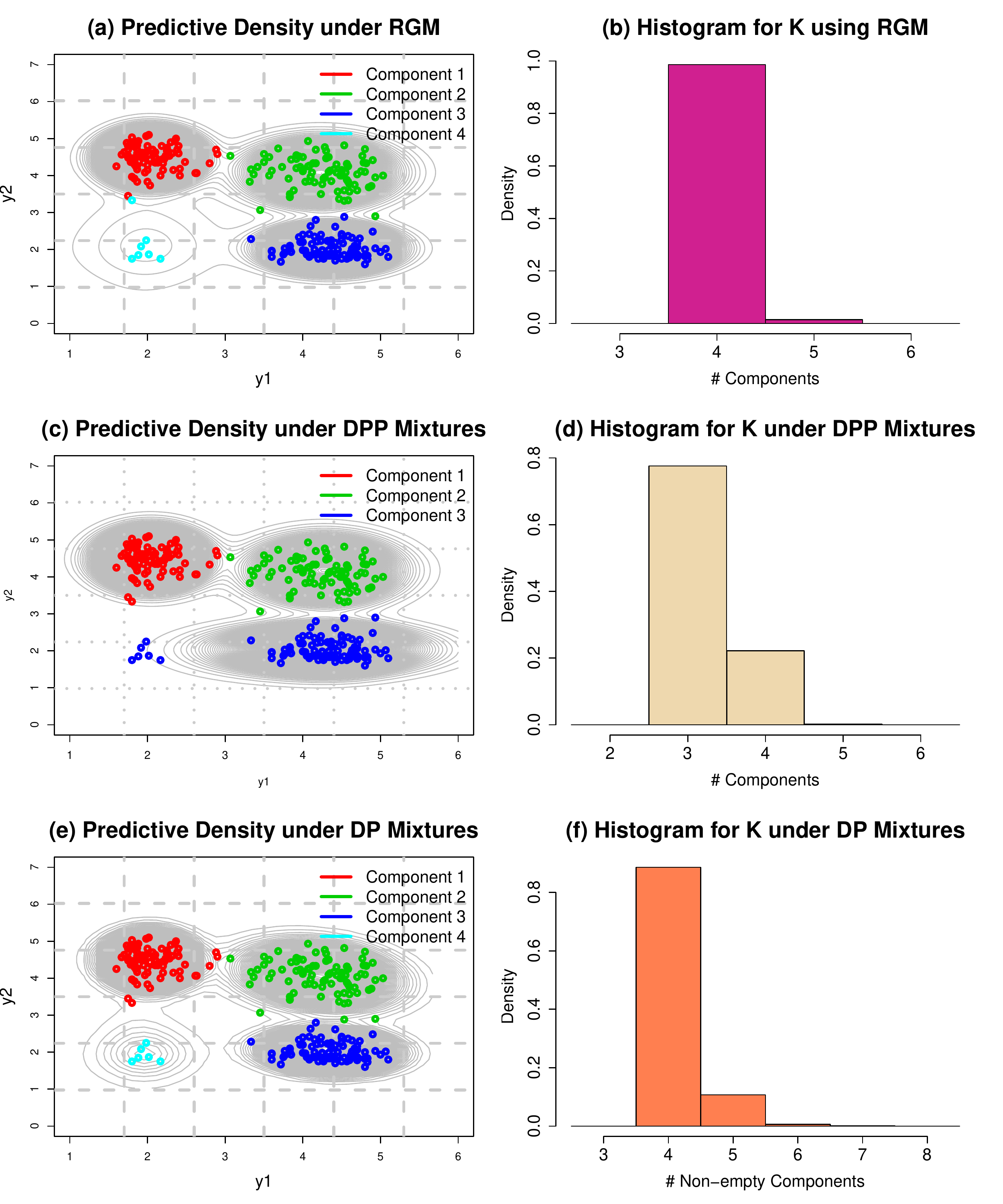}}
	\caption{Old Faithful Geyser Eruption Data: Panels (a), (c), and (e) are the scatter plots of the observations with their corresponding clusters and contour plots of the posterior predictive density estimate(grey level curves) stratified by the RGM model, the DPP mixture model, and the DPM model, respectively. Panels (b), (d), and (f) are the histograms of the posterior distributions of the number of components/clusters under the RGM model, the DPP mixture model, and the DPM model, respectively. }
	\label{old_faithful_figure}
\end{figure}
The proposed RGM, not only identifies the outliers component (Figure \ref{old_faithful_figure}a), but also provides the posterior number of components that is highly concentrated at $K = 4$ (Figure \ref{old_faithful_figure}b). In contrast, Figure \ref{old_faithful_figure}c shows that DPP mixture fails to identify the outliers at the bottom-left corner of the scatter plot -- instead, they are merged into the existing cluster located  at the bottom-right corner. The corresponding posterior number of components $K$, as illustrated in Figure \ref{old_faithful_figure}d, is highly concentrated at $K = 3$, failing to detect the outlier component. In addition, notice that failure in identifying the outliers significantly affects the posterior predictive density estimate, as shown from the comparison of the level curves among Figures \ref{old_faithful_figure}a, \ref{old_faithful_figure}c, and \ref{old_faithful_figure}e. The DPM model in Figure \ref{old_faithful_figure}e, although successfully detects the outliers component, still assigns relatively larger posterior probability to redundant components(Figure \ref{old_faithful_figure}f). Hence the proposed RGM model outperforms the other two alternatives in terms of the robustness or the model complexity measured by the posterior of $K$. This conclusion is also supported by the fact that log-CPO of the RGM model is higher than those of the DPM model and the DPP mixture model (Table \ref{table:log_cpo_numerical_results}). 



\section{Conclusion} 
\label{sec:conclusion}
We propose the RGM model, in which the location parameters for each component are not \emph{a priori} independent, but jointly distributed according to some symmetric repulsive distribution that encourages the separation of the locations for different components. We establish the posterior consistency and obtain an ``almost'' parametric posterior contraction rate($(\log n)^t/\sqrt{n}$ with $t> p+1$), generalizing the repulsive mixture model proposed by \cite{petralia2012repulsive,quinlan2017parsimonious} to the context of density estimation  in nonparametric GMM. Furthermore, we study the shrinkage effect on the model complexity of the proposed RGM model regarding the number of necessary components needed to fit the data well.

Based on the exchangeable partition distribution,
we develop a blocked-collapsed Gibbs sampler for the posterior inference. 
Through extensive simulation studies and real data analysis, we demonstrate that the proposed RGM model is able to detect outliers and simultaneously penalize the number of components to reduce model complexity and accurately estimate the underlying true density.
Moreover, the proposed sampler converges much faster than the RJ-MCMC sampler in \cite{xu2016bayesian} even in slightly higher dimensional clustering problems.

There are several potential further extensions. Beyond mixture models for density estimation, 
it is also interesting to extend the repulsive mixture model to the nested clustering of grouped data, and perform simultaneous clustering of individuals within each group and the group level features when the inference prefers the parsimonious model and the focus is the interpretation of the clusters as meaningful subgroups. Secondly, 
the posterior distribution of the number of components under the RGM model is potentially sensitive to the hyperparameters in the repulsive function $h_K$. Performing sensitivity analysis by imposing suitable priors on the hyperparameters is possible if an efficient updating rule for them can be integrated within the blocked-collapsed Gibbs sampler. 
Lastly, instead of implementing a Gibbs sampler, which is not scalable to large number of observations, one can develop an optimization-based fast inference algorithm, which would greatly improve the computational efficiency and scalability of the posterior inference.

\bibliographystyle{apalike}
\bibliography{reference_appendix_new}

\clearpage
\begin{center}
	\begin{LARGE}
		\textbf{Bayesian Repulsive Gaussian Mixture Model}
	\end{LARGE}
\end{center}
\begin{center}
	\begin{LARGE}
		\textbf{Supplementary Material}
	\end{LARGE}
\end{center}
\appendix
\counterwithin{lemma}{section}
\counterwithin{theorem}{section}
\section{Supporting Results} 
\label{sec:supporting_results}
\subsubsection*{Sufficient Conditions for Posterior Weak Consistency}
We use the results in \cite{wu2008kullback} to establish the weak consistency of $\Pi$. Denote $\Pi^\star$ the prior on $F\in\calM(\mathbb{R}^p\times\calS)$ that induces the prior $\Pi$ on $f$. Notice that the prior $\Pi^*$ on $F$ is supported on the class of all finitely discrete probability distributions on $\mathbb{R}^p\times\calS$, which is dense in $\mathcal{M}(\mathbb{R}^p\times\calS)$ under the weak topology, we conclude that $\Pi^*$ has the weak full support on $\mathcal{M}(\mathbb{R}^p\times\calS)$. As a consequence, we need to verify the conditions A1, A7, A8, and A9 (which we list as C1, C2, C3, and C4) there: For all $\epsilon>0$ in \cite{wu2008kullback} exists some $F_\epsilon\in\mathrm{supp}(\Pi^*)$, a closed set $D\supset\mathrm{supp}(F_\epsilon)$ such that
\begin{enumerate}[noitemsep, topsep = 0cm]
	\item[{C1}]$\int_{\mathbb{R}^p}f_0(\by)\log\frac{f_0(\by)}{f_{F_\epsilon}(\by)}\mathrm{d}\by<\epsilon$;
	\item[{C2}]$\int_{\mathbb{R}^p}f_0(\by)\left|\log\frac{f_{F_\epsilon}(\by)}{\inf_{(\bmu,\bSigma)\in D}\phi(\by\mid\bmu,\bSigma)}\right|\mathrm{d}\by<\infty$;
	\item[{C3}]For any compact $C\subset\mathbb{R}^p$, $c:=\inf_{(\by, \bmu,\bSigma)\in C\times D}\phi(\by\mid\bmu,\bSigma)>0$;
	\item[{C4}]For any compact $C\subset\mathbb{R}^p$, there exists some $E\subset\mathbb{R}^p\times\calS$ such that $D$ is contained in the interior of $E$, the class of functions $\{(\bmu,\bSigma)\mapsto \phi(\by\mid\bmu,\bSigma):\by\in C\}$ is uniformly equicontinuous on $E$, and $\sup\{\phi(\by\mid\bmu,\bSigma):\by\in C, (\bmu,\bSigma)\in E^c\}<c\epsilon/4$. 
\end{enumerate}

\subsubsection*{Sufficient Conditions for Posterior Strong Consistency}
To prove the posterior strong consistency of the RGM model we apply {Theorem 1} in \cite{canale2017posterior}. 
\begin{theorem}
Consider a statistical $\mathcal{F}$ with a prior $\Pi$, let $(\by_i)_{i=1}^n$ be an $\iid$ sequence with density $f_0\in\mathcal{F}$. Assume that there exists a sequence of submodels $(\mathcal{F}_n)_{n=1}^{\infty}$ with partitions $\mathcal{F}_n=\bigcup_{j=1}^{\infty}\mathcal{F}_{nj}$. If $f_0$ is in the KL-support of $\Pi$, and there exists some $a,b>0$ such that $\Pi(\mathcal{F}_n^c)\lesssim\mathrm{e}^{-bn}$, and
\begin{eqnarray}
	\exp\left(-(4-a)n\epsilon^2\right)\sum_{j=1}^{\infty}\sqrt{\mathcal{N}(2\epsilon, \mathcal{F}_{nj},\|\cdot\|_1)}\sqrt{\Pi(\mathcal{F}_{nj})}\to 0,
\end{eqnarray}
then $\Pi(f:\|f-f_0\|>\epsilon\mid \by_1,\cdots,\by_n)\to 0$ in $\mathbb{P}_0$-probability. 
\end{theorem}
\subsubsection*{Theorem 3 in \cite{kruijer2010adaptive}}
To compute the posterior rate of convergence of the RGM model, we rely on the conditions of {Theorem 3} in \cite{kruijer2010adaptive}.
\begin{theorem}
Given a statistical model $\mathcal{F}$ with a prior $\Pi$, let $(\by_i)_{i=1}^n$ be an $\iid$ sequence with density $f_0\in\mathcal{F}$. Assume that there exists a sequence of submodels $(\mathcal{F}_n)_{n=1}^{\infty}$ with partitions $\mathcal{F}_n=\bigcup_{j=1}^{\infty}\mathcal{F}_{nj}$, and two sequences $(\leps_n)_{n=1}^{\infty},(\ueps_n)_{n=1}^{\infty}$ with $\leps_n,\ueps_n\to0$, $n\leps_n^2,n\ueps_n^2\to\infty$, $\ueps_n\geq\leps_n$, such that
\begin{align}
&\Pi\left(\mathcal{F}_n^c\right)\lesssim\exp(-4n\leps_n^2),\\
&\exp\left(-n\ueps_n^2\right)\sum_{j=1}^{\infty}\sqrt{\mathcal{N}(\ueps_n,\mathcal{F}_{nj},\|\cdot\|_1)}\sqrt{\Pi(\mathcal{F}_{nj})}\to0,\\
&\Pi\left(f:\int f_0\log\frac{f_0}{f}\leq\leps_n^2,\int f_0\left(\log\frac{f_0}{f}\right)^2\leq\leps_n^2\right)\geq\exp(-n\leps_n^2).
\end{align}
Then $\Pi(f:\|f-f_0\|>\ueps_n\mid\by_1,\cdots,\by_n)\to 0$ in $\mathbb{P}_0$-probability.
\end{theorem}

\clearpage
\section{Proof of Theorem \ref{thm:bounding_Z_K}} 
\label{sec:proof_of_theorem_thm:bounding_z_k}
\begin{proof}
First of all, notice that $h_K(\bmu_1,\cdots,\bmu_K)\leq 1$, we see immediately that 
\begin{eqnarray}
Z_K\leq\int_{\mathbb{R}^p}\cdots\int_{\mathbb{R}^p}\prod_{k=1}^K p_{\bmu}(\bmu_k)\mathrm{d}\bmu_1\cdots\mathrm{d}\bmu_K=1,\nonumber
\end{eqnarray}
and hence $-\log Z_K\geq 0$. Now we consider the upper bound for $-\log Z_K$. Suppose $h_K$ is of the form \eqref{eqn:repulsive_function_1}. Let $\bmu_1,\cdots,\bmu_K\iidsim p_\bmu$. Then by Jensen's inequality,
\begin{eqnarray}
-\log Z_K=-\log \mathbb{E}\left[\min_{1\leq k<k'\leq K}g(\|\bmu_k-\bmu_{k'}\|)\right]
\leq \mathbb{E}\left[\max_{1\leq k<k'\leq K}-\log g(\|\bmu_k-\bmu_{k'}\|)\right].\nonumber
\end{eqnarray}
Observing that
$$
\left[\max_{1\leq k<k'\leq K}-\log g(\|\bmu_k-\bmu_{k'}\|)\right]^2=\max_{1\leq k<k'\leq K}\left[\log g(\|\bmu_k-\bmu_{k'}\|)\right]^2,
$$
we obtain
\begin{eqnarray}
-\log Z_K
&\leq &\left\{\mathbb{E}\left[\max_{1\leq k<k'\leq K}\left[\log g(\|\bmu_k-\bmu_{k'}\|)\right]^2\right]\right\}^{\frac{1}{2}}
\leq\left\{\sum_{1\leq k<k'\leq K}\mathbb{E}\left[\log g(\|\bmu_k-\bmu_{k'}\|)\right]^2\right\}^{\frac{1}{2}}\nonumber\\
&=&\left\{\frac{1}{2}K(K-1)\mathbb{E}\left[\log g(\|\bmu_1-\bmu_2\|)\right]^2\right\}^{\frac{1}{2}}
\leq c_1K,\nonumber
\end{eqnarray}
where the constant $c_1$ can be taken as
$$
c^2_1=\frac{1}{2}\mathbb{E}\left[\log g(\|\bmu_1-\bmu_2\|)\right]^2=\frac{1}{2}\iint_{\mathbb{R}^p\times\mathbb{R}^p}\left[\log g(\|\bmu_1-\bmu_2\|)\right]^2p(\bmu_1)p(\bmu_2)\mathrm{d}\bmu_1\mathrm{d}\bmu_2<\infty.
$$
Now we consider the case where $h_K$ is of the form \eqref{eqn:repulsive_function_2}. Still let $\bmu_1,\cdots,\bmu_K\iidsim p(\bmu)$. Jensen's inequality yields
\begin{eqnarray}
-\log Z_K&=&-\log\mathbb{E}\left[\prod_{1\leq k<k'\leq K}g(\|\bmu_k-\bmu_{k'}\|)^{\frac{1}{K}}\right]
\leq\sum_{1\leq k<k'\leq K}\frac{1}{K}\mathbb{E}\left[-\log g(\|\bmu_k-\bmu_{k'}\|)\right]\nonumber\\
&\leq&\frac{K-1}{2}\left\{\mathbb{E}\left[\log g(\|\bmu_1-\bmu_2\|)\right]^2\right\}^{\frac{1}{2}}\leq c_2K\nonumber
\end{eqnarray}
for some constant $c_2>0$. 
\end{proof}


\section{Proofs of Posterior Consistency} 
\label{sec:proofs_of_posterior_consistency}

\subsection*{Proof of Lemma \ref{lemma:Approximation}} 
\label{sub:proof_of_lemma_lemma:approximation}
\begin{proof}
	Without loss of generality we assume that $\calT_1$ is non-empty. Clearly, $\calT_m\uparrow\mathbb{R}^p\times\calS$ and $c_m\downarrow 1$ as $m\to\infty$ by the monotone continuity of $F_0$. Furthermore, $\phi(\by\mid\bmu,\bSigma)\leq(2\pi\lsigma^2)^{-\frac{p}{2}}$. Hence, $f_{F_m}=c_m\left[\phi_{\bSigma}(\by-\bmu)\mathbb{I}_{\calT_m}(\bmu)\right]*F_0\to\phi_{\bSigma}*F_0=f_0$ by the bounded convergence theorem, implying that $\log\frac{f_0}{f_{F_m}}\to 0$ as $m\to\infty$. In order to show $\int f_0\log\frac{f_0}{f_{F_m}}\to0$ as $m\to\infty$, it suffices to find a dominating function $g(\by)$ such that $\left|\log\frac{f_0}{f_{F_m}}\right|\leq g$ for all $m\in\mathbb{N}_+$, and the conclusion is guaranteed by the dominating convergence theorem. 
	
	First of all, notice that for all $m\in\mathbb{N}_+$, we have $f_{F_m}\leq c_m\phi_{\bSigma}*F_0\leq c_1(2\pi\lsigma^2)^{-\frac{p}{2}}$, and thus $f_0\leq c_1(2\pi\lsigma^2)^{-\frac{p}{2}}$ by letting $m\to\infty$. It follows that $\log\frac{f_0}{f_{F_m}}\geq\log\frac{f_0}{c_1(2\pi\lsigma^2)^{-\frac{p}{2}}}$. Next, we see that
	\begin{eqnarray}
	f_{F_m}(\by)&=&c_m\int_{\calT_m}\phi(\by\mid\bmu,\bSigma)\mathrm{d}F_0(\bmu,\bSigma)\nonumber\\
	&\geq&\int_{\calT_1}\phi(\by\mid\bmu,\bSigma)\mathrm{d}F_0(\bmu,\bSigma)\nonumber\\
	&\geq&(2\pi\usigma^2)^{-\frac{p}{2}}\int_{\calT_1}\exp\left(-\frac{1}{2\lsigma^2}\|\by-\bmu\|^2\right)\mathrm{d}F_0(\bmu,\bSigma)\nonumber.
	\end{eqnarray}
	If $\|\by\|\leq1$, then $\|\by-\bmu\|\leq 2$ as $\|\bmu\|\leq1$, and hence $\exp\left(-\frac{\|\by-\bmu\|^2}{2\lsigma^2}\right)\geq\exp\left(-\frac{2}{\lsigma^2}\right)$; If $\|\by\|>1$, then $\|\by-\bmu\|\leq2\|\by\|$ as $\|\bmu\|\leq\|\by\|$, and hence $\exp\left(-\frac{\|\by-\bmu\|^2}{2\lsigma^2}\right)\geq\exp\left(-\frac{2\|\by\|^2}{\lsigma^2}\right)$. It follows that
	\begin{eqnarray}\label{ineq:xi_function}
	f_{F_m}(\by)\geq\xi(\by):=(2\pi\usigma^2)^{-\frac{p}{2}}\left\{
	\begin{aligned}
	&\exp\left(-\frac{2}{\lsigma^2}\right)F_0(\{\bmu:\|\bmu\|\leq1\}\times\calT_1),\quad&&\text{if }\|\by\|\leq 1,\\
	&\exp\left(-\frac{2\|\by\|^2}{\lsigma^2}\right)F_0(\{\bmu:\|\bmu\|\leq1\}\times\calT_1),\quad&&\text{if }\|\by\|>1.
	\end{aligned}
	\right.
	\end{eqnarray}
	and thus, $\log\frac{f_0}{f_{F_m}}\leq\log\frac{f_0}{\xi}$. In particular, $f_0\geq\xi$ by letting $m\to\infty$. Together we have
	\begin{align}
	&\log\frac{f_0}{c_1(2\pi\lsigma^2)^{-\frac{p}{2}}}\leq\log\frac{f_0}{f_{F_m}}\leq\log\frac{f_0}{\xi}\nonumber\\
	&\qquad\Longrightarrow
	\left|\log\frac{f_0}{f_{F_m}}\right|\leq g:=\max\left\{\left|\log\frac{f_0}{c_1(2\pi\lsigma^2)^{-\frac{p}{2}}}\right|, \log\frac{f_0}{\xi}\right\}.\nonumber	
	\end{align}
	To show that $g$ is $f_0$-integrable, it suffices to verify the $f_0$-integrability of $\log f_0$ and $\log\xi$. Notice that $\log c_1-(\frac{p}{2})\log(2\pi\usigma^2)\geq \log f_0\geq\log\xi$, implying 
	$$
	|\log f_0|\leq |\log c_1|+\frac{p}{2}|\log(2\pi\usigma^2)|+|\log\xi|,
	$$ it is only left to verify the $f_0$-integrability of $\log\xi$. When $\|\by\|\leq1$, $\log\xi$ is constant, and when $\|\by\|>1$, we have
	\begin{align}
	&\int_{\{\|\by\|\geq1\}}f_0(\by)\left|\log\xi(\by)\right|\mathrm{d}\by\nonumber\\
	&\qquad\leq \frac{p}{2}\left|\log(2\pi\usigma^2)\right|+\left|\log F_0(\{\bmu:\|\bmu\|\leq 1\}\times\calT_1)\right|+\frac{2}{\lsigma^2}\int_{\{\|\by\|\geq1\}}\|\by\|^2f_0(\by)\mathrm{d}\by\nonumber\\
	&\qquad\leq \frac{p}{2}\left|\log(2\pi\usigma^2)\right|+\left|\log F_0(\{\bmu:\|\bmu\|\leq 1\}\times\calT_1)\right|+\frac{2}{\lsigma^2}\mathbb{E}_0\|\by\|^2\nonumber\\
	&\qquad<\infty,\nonumber
	\end{align}
	where the finiteness of $\mathbb{E}_0\|\by\|^2$ is guaranteed by condition A1 and Fubini's theorem. Hence $\log\xi$ is $f_0$-integrable. 
\end{proof}

\subsection*{Proof of Theorem \ref{thm:weak_consistency}} 
\label{sub:proof_of_theorem_thm:weak_consistency}

\begin{proof}
	By {Theorem 1} and {Lemma 3} in \cite{wu2008kullback}, it suffices to verify conditions {C1}, {C2}, {C3}, and {C4}. By \textbf{Lemma \ref{lemma:Approximation}}, for all $\epsilon>0$, there exists an integer $m$ such that $F_\epsilon = F_m$ satisfies {C1}. Noticing that $F_\epsilon\in \mathrm{supp}(\Pi^*)$ automatically holds since $\mathrm{supp}(\Pi^*)=\calM(\mathbb{R}^p\times\calS)$, and that $\calS$ itself is compact, we can take $D=\calT_m$. For any compact $C\subset\mathbb{R}^p$, take large enough $a$ such that $C\subset\{\by:\|\by\|\leq a\}$. In addition, {C3} automatically holds, since $C\times D$ is compact in $\mathbb{R}^p\times \mathbb{R}^p\times\calS$, and $\phi$ is strictly positive. It suffices to verify {C2} and {C4}. 
	
	To verify {C2}, it suffices to show that $\log f_{F_m}$ and $\log \inf_{(\bmu,\bSigma)\in D}\phi(\by\mid\bmu,\bSigma)$ are $f_0$-integrable. Notice that
	$$
	(2\pi\lsigma^2)^{-\frac{p}{2}}\geq\inf_{(\bmu,\bSigma)\in D}\phi(\by\mid\bmu,\bSigma)\geq \zeta_m(\by):=(2\pi\usigma^2)^{-\frac{p}{2}}\left\{
	\begin{aligned}
	&\exp\left(-\frac{2m^2}{\lsigma^2}\right),\quad&&\text{if }\|\by\|\leq m,\\
	&\exp\left(-\frac{2\|\by\|^2}{\lsigma^2}\right),\quad&&\text{if }\|\by\|>m,
	\end{aligned}
	\right.
	$$
	since when $\|\by\|\leq m$ we have $(\by-\bmu)^{\mathrm{T}}\bSigma^{-1}(\by-\bmu)\leq\lsigma^{-2}\|\by-\bmu\|^2\leq 4\lsigma^{-2}m^2$, and when $\|\by\|>m$ we have $(\by-\bmu)^{\mathrm{T}}\bSigma^{-1}(\by-\bmu)\leq\lsigma^{-2}\|\by-\bmu\|^2\leq 4\lsigma^{-2}\|\by\|^2$. It follows that $\log \inf_{(\bmu,\bSigma)\in D}\phi(\by\mid\bmu,\bSigma)$ is $f_0$-integrable if $\log\zeta_m$ is integrable. When $\|\by\|\leq m$, $\zeta_m$ is a constant, and when $\|\by\|>m$, 
	$$
	\int_{\{\|\by\|>m\}}f_0(\by)|\log\zeta_m(\by)|\mathrm{d}y\leq \frac{p}{2}\left|\log(2\pi\usigma^2)\right|+\frac{2}{\lsigma^2}\mathbb{E}_0\|\by\|^2<\infty.
	$$
	Hence $\log\inf_{(\bmu,\bSigma)\in D}\phi(\by\mid\bmu,\bSigma)$ is $f_0$-integrable. Using the $\xi$ function constructed in \eqref{ineq:xi_function} in the proof of \textbf{Lemma \ref{lemma:Approximation}}, we see that $c_1(2\pi\lsigma^2)^{-\frac{p}{2}}\geq f_{F_m}(\by)\geq\xi(\by)$, and it is proved that $\log\xi(\by)$ is $f_0$-integrable. It follows that $\log f_{F_m}$ is $f_0$-integrable. 
	
	To verify {C4}, given compact $C$ with $C\subset\{\by:\|\by\|\leq a\}$ for some large enough $a>0$,  let 
	$$
	E=\left\{\bmu:\|\bmu\|\leq \max(a, m)+\max\left[1,\sqrt{2\usigma^2\log\left(\frac{8}{(2\pi\lsigma^2)^{\frac{p}{2}}c\epsilon}\right)}\right]\right\}\times{\calS}.
	$$
	Then $E$ contains $D$ in its interior, and $E$ is also compact. Therefore the function $(\by,\bmu,\bSigma)\mapsto \phi(\by\mid\bmu,\bSigma)$ on $C\times E$ is uniformly continuous, and hence, as $\by$ varies over $C$, the class of functions $\{(\bmu,\bSigma)\in E\mapsto \phi(\by\mid\bmu,\bSigma):\by\in C\}$ is also uniformly equicontinuous. Now we show that $\sup\{\phi(\by\mid\bmu,\bSigma):\by\in C, (\bmu,\bSigma)\in E^c\}<c\epsilon/4$. Since for any $(\by,\bmu,\bSigma)\in C\times E^c$, we have 
	\begin{align}
	&\|\by\|\leq a,\quad\|\bmu\|>a+\max\left[1,\sqrt{2\usigma^2\log\left(\frac{8}{(2\pi\lsigma^2)^{\frac{p}{2}}c\epsilon}\right)}\right]\nonumber\\
	&\qquad\Longrightarrow\|\by-\bmu\|\geq\|\bmu\|-\|\by\|\geq\max\left[1,\sqrt{2\usigma^2\log\left(\frac{8}{(2\pi\lsigma^2)^{\frac{p}{2}}c\epsilon}\right)}\right],\nonumber
	\end{align}
	then we obtain
	\begin{align}
	\sup_{(\by,\bmu,\bSigma)\in C\times E^c}\phi(\by\mid\bmu,\bSigma)
	\leq\frac{1}{(2\pi\lsigma^2)^{\frac{p}{2}}}\exp\left[-\frac{1}{2\usigma^2}(\|\bmu\|-\|\by\|)^2\right]<\frac{c\epsilon}{4}\nonumber.
	\end{align}
	The proof is thus completed. 
\end{proof}



\subsection*{Proof of Lemma \ref{lemma:entropy}} 
\label{sub:proof_of_lemma_lemma:entropy}

\begin{proof}[Proof of \textbf{Lemma \ref{lemma:entropy}}]
	Suppose $\delta>0$ is given. By {Lemma A.4} in \cite{ghosal2001entropies}, there exists an $\ell_1$ $\delta$-net $\mathcal{I}_0$ of $\Delta^K$, such that the cardinality $|\mathcal{I}_0|$ of $\mathcal{I}_0$ is upper bounded by $(5/\delta)^K$. Now let $\mathcal{R}_k$ be an $\delta$-net of $\{\bmu_k:\|\bmu_k\|_\infty\in(a_k,b_k]\}$ under the $\|\cdot\|_\infty$-metric. Clearly, one can make $|\mathcal{R}_k|\leq \left(b_k/\delta+1\right)^p$. 
	Furthermore let $\mathcal{S}_{jk}$ be an $\delta$-net of $\{\sqrt{\lambda_j(\bSigma_k)}:\lambda_j(\bSigma_k)\in[\lsigma^2,\usigma^2]\}$ with cardinality $|\mathcal{S}_{jk}|\leq(\usigma-\lsigma)/\delta+1$ under the $\|\cdot\|_\infty$-metric. It follows that for all $f_F\in\mathcal{F}_K\left(\prod_{k=1}^K(a_k,b_k]\right)$ with $
	F=\sum_{k=1}^Kw_k\delta_{(\bmu_k,\bSigma_k)},$
	there exists some $\bw^\star=(w_1^\star,\cdots,w_K^\star)\in\mathcal{I}_0$, $\bmu_k^\star\in\mathcal{R}_k$, $\lambda_{jk}^\star\in \mathcal{S}_{jk}$ for $j=1,\cdots,p$ with $\bSigma_k^\star=\bU\mathrm{diag}(\lambda_{1k}^\star,\cdots,\lambda_{pk}^\star)\bU\transpose$ for $k=1,\cdots,K$, such that $\sum_{k=1}^K|w_k-w_k^\star|<\delta$, $\|\bmu_k-\bmu_k^\star\|<\sqrt{p}\|\bmu_k-\bmu_k^\star\|_\infty<\sqrt{p}\delta$, and $|\sqrt{\lambda_j(\bSigma_k)}-\sqrt{\lambda_{jk}^\star}|<\delta$ for $j=1,\cdots,p$. Denote $H(f,g)$ to be the Hellinger distance between densities $f$ and $g$, defined by $H(f,g)=\left(\frac{1}{2}\int(\sqrt{f}-\sqrt{g})^2\right)^{\frac{1}{2}}$. Observe that
	\begin{align}\label{ineq:Hellinger_distance_upper_bound}
	H(\phi_{\bSigma_k}(\by-\bmu_k),\phi_{\bSigma_k^\star}(\by-\bmu_k^\star))^2
	&\leq1-\prod_{j=1}^p\left(1-\frac{(\sqrt{\lambda_j(\bSigma_k)}-\sqrt{\lambda_{jk}^\star})^2}{\lambda_j(\bSigma_k)+\lambda_{jk}^\star}\right)^{\frac{1}{2}}\exp\left(-\frac{\|\bmu_k-\bmu_k^\star\|^2}{8\lsigma^2}\right)\nonumber\\
	&\leq1-\left(1-\frac{\delta^2}{2\lsigma^2}\right)^{\frac{p}{2}}\exp\left(-\frac{p\delta^2}{8\lsigma^2}\right)
	\nonumber\\
	&\leq
	 1-\left(1-\frac{p\delta^2}{2\lsigma^2}\right)^{\frac{p}{2}+1}
	\end{align}
	where we use the fact $\exp(-x)\geq 1-x$ in the last inequality. Denote $
	F^\star=\sum_{k=1}^Kw_k^\star\delta_{(\bmu_k^\star,\bSigma_k^\star)}$. It follows by the triangle inequality that
	\begin{align}
	\left\|
	f_F-f_{F^\star}
	\right\|_1
	&\leq\sum_{k=1}^Kw_{k}\left\|
	\phi_{\bSigma_{k}}(\by-\bmu_{k})-\phi_{\bSigma_{k}^\star}(\by-\bmu_{k}^\star)
	\right\|_1+\sum_{k=1}^K|w_{k}-w_{k^\star}|\nonumber\\
	&\leq\sum_{k=1}^K2\sqrt{2}w_{k}H(\phi_{\bSigma_k}(\by-\bmu_k),\phi_{\bSigma_{k}^\star}(\by-\bmu_k^\star))+\delta\nonumber\\
	&\leq\delta+2\sqrt{2}\left[1-\left(1-\frac{p\delta^2}{2\lsigma^2}\right)^{\frac{p}{2}+1}\right]^{\frac{1}{2}}.\nonumber
	\end{align}
	Observing that $\lim_{t\downarrow 0}\frac{1-(1-t)^{a}}{at}=1$ holds for $a>1$, we see that for sufficiently small $\delta$, $\left\|f_F-f_{F^\star}\right\|_1\leq C_3 \delta$ for some constant $C_3>0$, and therefore
	\begin{align}
	\calN\left(C_3\delta,\mathcal{F}_K\left(\prod_{k=1}^K(a_k,b_k]\right),\|\cdot\|_1\right)
	&\leq\left(\frac{5}{\delta}\right)^{K}\left(\frac{2(\usigma-\lsigma)}{\delta}\right)^{Kp}\prod_{k=1}^K\left(\frac{b_k}{\delta}+1\right)^p.\nonumber\\
	&\leq\frac{\tilde{c}_3^K}{\delta^{Kp+K}}\prod_{k=1}^K\left(\frac{b_k+\delta}{\delta}\right)^p.\nonumber
	\end{align}
	for some constant $\tilde{c}_3>0$. This yields that
	\begin{eqnarray}
	\calN\left(\delta,\mathcal{F}_K\left(\prod_{k=1}^K(a_k,b_k]\right),\|\cdot\|_1\right)
	\leq\left(\frac{c_3}{\delta^{2p+1}}\right)^K\left(\prod_{k=1}^Kb_k\right)^p\nonumber
	\end{eqnarray}
	for some constant $c_3>0$. 
\end{proof}

\subsection*{Proof of Lemma \ref{lemma:summability_upper_bound}} 
\label{sub:proof_of_lemma_lemma:summability_upper_bound}

\begin{proof}
	First we need to bound $\sqrt{\Pi\left(\mathcal{G}_K(\bolda_K)\right)}$. 
	Recall that $\mathrm{e}^{-c_1K}\leq Z_K\leq 1$ for some constant $c_1>0$ by \textbf{Theorem \ref{thm:bounding_Z_K}} and condition A2. We estimate
		\begin{eqnarray}
		\Pi\left(\mathcal{G}_K(\bolda_K)\right)
		&\leq&\Pi\left(\bmu_1,\cdots,\bmu_K:\|\bmu_k\|\geq \sqrt{p}a_k, k=1,\cdots,K\mid K\right)p(K)\nonumber\\
		&\leq&\frac{p(K)}{Z_K}\int\cdots\int \prod_{k=1}^K\mathbb{I}\left(\|\bmu_k\|^2\geq p a_k^2\right)p(\bmu_1)\mathrm{d}\bmu_1\cdots p(\bmu_K)\mathrm{d}\bmu_K\nonumber\\
		&\leq&\mathrm{e}^{c_1K}\prod_{k=1}^K\int_{\{\|\bmu_k\|^2\geq pa_k^2\}}p(\bmu_k)\mathrm{d}\bmu_{k}\nonumber\quad\text{(by \textbf{Theorem \ref{thm:bounding_Z_K}})}\\
		&\leq&\mathrm{e}^{c_1K}B_2^K\prod_{k=1}^K\exp\left(-pb_2a_k^2\right)\nonumber.\quad\text{(by condition B2)}
		\end{eqnarray}
		Now by \textbf{Lemma \ref{lemma:entropy}} for some constant $c_3>0$, we have
		\begin{eqnarray}
			\calN(\delta, \mathcal{G}_K(\bolda_K),\|\cdot\|_1)\leq\left(\frac{c_3}{\delta^{2p+1}}\right)^K\prod_{k=1}^K(a_k+1)^p.\nonumber
		\end{eqnarray}
		Hence, by defining $
		S=\sum_{a_k=0}^{\infty}(a_k+1)^{\frac{p}{2}}\exp\left(-\frac{pb_2a_k^2}{2}\right)<\infty，
		$ we estimate
		\begin{align}
		&\sum_{K=1}^{K_n}\sum_{a_1=0}^\infty\cdots\sum_{a_K=0}^\infty \sqrt{\calN\left(\delta,\mathcal{G}_K(\bolda_K),\|\cdot\|_1\right)}\sqrt{\Pi\left(\mathcal{G}_K(\bolda_K)\right)}\nonumber\\
		&\qquad\leq\sum_{K=1}^{K_n}\left[\frac{\sqrt{B_2c_3\mathrm{e}^{c_1}}}{\delta^{p+\frac{1}{2}}}\right]^{K}\left[\prod_{k=1}^K\sum_{a_k=0}^{\infty}(a_k+1)^{\frac{p}{2}}\exp\left(-\frac{b_2pa_k^2}{2}\right)\right]\nonumber\\
		&\qquad=\sum_{K=1}^{K_n}\left[\frac{S\sqrt{B_2c_3\mathrm{e}^{c_1}}}{\delta^{p+\frac{1}{2}}}\right]^K\nonumber\\
		&\qquad\leq K_n\left(\frac{M}{\delta^{p+\frac{1}{2}}}\right)^{K_n},\nonumber
		\end{align}
		for some constant $M>0$ for sufficiently small $\delta$. 
\end{proof}

\subsection*{Proof of Theorem \ref{thm:strong_consistency}} 
\label{sub:proof_of_theorem_thm:strong_consistency}

\begin{proof}
	It is sufficient to verify \eqref{eqn:summability_condition} and that $\Pi(\mathcal{F}_{K_n}^c)\lesssim \exp(-bn)$ for some $b>0$, since the KL-property is satisfied. Now take $K_n=\lfloor n/\log n\rfloor$. Then $K_n\log K_n\geq n-\log\log n/\log n\geq n/2$ for large $n$, which yields $\Pi(\mathcal{F}_{K_n}^c)\lesssim \exp(-B_4n/2)$ condition B5. Furthermore by \textbf{Lemma \ref{lemma:summability_upper_bound}} we have
\begin{align}
	&\sum_{K=1}^{K_n}\sum_{a_1=0}^{\infty}\cdots\sum_{a_K=0}^{\infty}\sqrt{\calN(\epsilon,\mathcal{G}_K(\bolda_K),\|\cdot\|_1)}\sqrt{\Pi(\mathcal{G}_K(\bolda_K))}\nonumber\\
	&\qquad\leq \exp\left[\log K_n+K_n\log M+\left(\frac{2p+1}{2}\right)K_n\left(\log\frac{1}{\epsilon}\right)\right]
	\nonumber\\
	&\qquad\leq \exp\left[(p+1)K_n\left(\log\frac{1}{\epsilon}\right)\right]\nonumber
\end{align}
for sufficiently small $\epsilon$ and sufficiently large $n$. The proof is completed by observing that $(p+1)K_n\log(1/\epsilon)-(4-\tilde{b})n\epsilon^2\to -\infty$ as $n\to\infty$ for any fixed $\epsilon>0$ and fixed $\tilde{b}\in(0,4)$. 
\end{proof}

\section{Proofs for Posterior Contraction Rate} 
\label{sec:proofs_for_posterior_contraction_rate}

\subsection*{Proof of Proposition \ref{prop:supersmooth_rate}} 
\label{sub:proof_of_proposition_prop:supersmooth_rate}

\begin{proof}
	Denote $C=1/(p+1)$. Then by condition B5 we have
	\begin{eqnarray}
	\Pi(\mathcal{F}_{K_n}^c)=\Pi(K>K_n)\leq\exp\left(-B_4K_n\log K_n\right)\leq \exp\left[-B_4C\log C(\log n)^{2t-1}\right]\leq\exp(-4n\leps_n^2)\nonumber
	\end{eqnarray}
	with $t>t_0+\frac{1}{2}$ for sufficiently large $n$. Next, by \textbf{Lemma \ref{lemma:summability_upper_bound}}
	\begin{align}
		&\exp(-n\ueps_n^2)\sum_{K=1}^{K_n}\sum_{a_1=0}^{\infty}\cdots\sum_{a_K=0}^{\infty}\sqrt{\calN(\ueps_n,\mathcal{G}_K(a_1,\cdots,a_K),\|\cdot\|_1)}\sqrt{\Pi(\mathcal{G}_K(a_1,\cdots,a_K))}\nonumber\\
		&\qquad\leq \exp\left[-(\log n)^{2t}+(p+1)C(\log n)^{2t-1}\left(\frac{1}{2}\log n-t\log\log n\right)\right]\nonumber\\
		&\qquad\leq \exp\left[-\frac{1}{2}(\log n)^{2t}\right]\nonumber.
	\end{align}
	The RHS of the last display converges to $0$ as $n\to\infty$. 
\end{proof}

\subsection*{Proof of Lemma \ref{lemma:KL_ball_lower_bound}} 
\label{sub:proof_of_lemma_4}

The proof of \textbf{Lemma \ref{lemma:KL_ball_lower_bound}} requires the following auxiliary \textbf{Lemmas \ref{lemma:discretization}-\ref{lemma:discrete_approximation}}   that  generalize {Lemma 3.4, Lemma 4.1}, and {Lemma 5.1} in \cite{ghosal2001entropies}. Since the proofs are quite similar to those there, we defer them in Section \ref{sec:proofs_of_auxiliary_results_in_appendix_e}. 
\begin{lemma}\label{lemma:discretization}
	Let ${F}$ be a probability distribution compactly supported on a subset of $\{(\bmu,\bSigma)\in\mathbb{R}^p\times\calS:\|\bmu\|_\infty\leq a\}$ with $a\lesssim\left(\log\frac{1}{\epsilon}\right)^{\frac{1}{2}}$. Then for sufficiently small $\epsilon>0$, there exists a discrete probability distribution $F^\star$ on a subset of $\{(\bmu,\bSigma)\in\mathbb{R}^p\times\calS:\|\bmu\|_\infty\leq a\}$ with at most $N\lesssim \left(\log\frac{1}{\epsilon}\right)^{2p}$ support points, such that $\|f_F-f_{F^\star}\|_\infty\lesssim\epsilon$, and $\|f_F-f_{F^\star}\|_1\lesssim\epsilon\left(\log\frac{1}{\epsilon}\right)^{\frac{p}{2}}$.
\end{lemma}
\begin{lemma}\label{lemma:discretization_grid}
	Let ${F}$ be a probability distribution compactly supported on a subset of $\{(\bmu,\bSigma)\in\mathbb{R}^p\times\calS:\|\bmu\|_\infty\leq a\}$ with $a\lesssim\left(\log\frac{1}{\epsilon}\right)^{\frac{1}{2}}$. Then for sufficiently small $\epsilon>0$, there exists a discrete probability distribution $F^\star$ on $\{(\bmu,\bSigma)\in\mathbb{R}^p\times\calS:\|\bmu\|_\infty\leq 2a\}$ with at most $N\lesssim \left(\log\frac{1}{\epsilon}\right)^{2p}$ support points that are taken from 
	$$
	\left\{(\bmu,\bSigma)\in\mathbb{R}^p\times\calS:\frac{\bmu}{2\epsilon}\in\mathbb{Z}^p,\frac{\lambda_j(\bSigma)}{2\epsilon}\in\mathbb{N}_+,j=1,\cdots,p\right\},$$
	such that $\|f_F-f_{F^\star}\|_1\lesssim\epsilon\left(\log\frac{1}{\epsilon}\right)^{\frac{p}{2}}$. 
\end{lemma}
\begin{lemma}\label{lemma:KL_ball_versus_Hellinger_ball}
	If $F(\|\bmu\|\leq B)>\frac{1}{2}$ for some constant $B$ and $F_0$ is such that for all $t\geq0$, ${F_0}(\|\bmu\|>t)\leq\exp(-b't^2)$ for some $b'>0$, then for $\epsilon=H(f_{F_0},f_F)$ sufficiently small, 
	$$\int f_0\left(\log\frac{f_0}{f_F}\right)^2\lesssim \epsilon^2\left(\log\frac{1}{\epsilon}\right)^2,\quad
	\int f_0\log\frac{f_0}{f_F}\lesssim\epsilon^2\left(\log\frac{1}{\epsilon}\right).$$ 
\end{lemma}
\begin{lemma}\label{lemma:discrete_approximation}
	Let $\epsilon>0$ be sufficiently small, $F^\star=\sum_{k=1}^Nw_k^\star\delta_{(\bmu_k^\star,\bSigma_k^\star)}$ be such that $\|\bmu_k^\star-\bmu_{k'}^\star\|_\infty\geq 2\epsilon$, and $|\lambda_j(\bSigma_k^\star)-\lambda_j(\bSigma_{k'}^\star)|\geq 2\epsilon$ whenever $k\neq k'$, $j=1,\cdots,p$. 
	Define 
	$$
	E_k=\left\{(\bmu,\bSigma)\in\mathbb{R}^p\times\calS:\|\bmu-\bmu_k^\star\|_\infty<\frac{\epsilon}{2},|\lambda_j(\bSigma)-\lambda_j(\bSigma_k^\star)|<\frac{\epsilon}{2},j=1,\cdots,p\right\}.
	$$
	Then for any probability distribution $F$ on $\mathbb{R}^p\times\calS$, 
	\begin{eqnarray}
		\|f_F-f_{F^\star}\|\lesssim\epsilon+\sum_{k=1}^N|{P}_F(E_k)-w_k^\star|.\nonumber
	\end{eqnarray}
\end{lemma}

\begin{proof}[\textbf{Proof of \textbf{Lemma \ref{lemma:KL_ball_lower_bound}}}]
	The proof is similar to those in {Theorem 5.1} and {Theorem 5.2} in \cite{ghosal2001entropies}. First let $F_0'$ be the re-normalized restriction of $F_0$ on $\{(\bmu,\bSigma)\in\mathbb{R}^p\times\calS:\|\bmu\|\leq a\}$. By {Lemma A.3} in \cite{ghosal2001entropies} we obtain $\|f_0-f_{F_0'}\|_1\lesssim \epsilon$. Next find $F^\star=\sum_{k=1}^Nw_k^\star\delta_{(\bmu_k^\star,\bSigma_k^\star)}$ by \textbf{Lemma \ref{lemma:discretization_grid}} such that $N\lesssim \left(\log\frac{1}{\epsilon}\right)^{2p}$, $\|f_{F_0'}-f_{F^\star}\|_1\lesssim\epsilon\left(\log\frac{1}{\epsilon}\right)^{\frac{p}{2}}$, 
	$$
	(\bmu_k^\star,\bSigma_k^\star)\in\left\{(\bmu,\bSigma)\in\mathbb{R}^p\times\calS:\frac{\bmu}{2\epsilon}\in\mathbb{Z}^p,\frac{\lambda_j(\bSigma)}{2\epsilon}\in\mathbb{N}_+,j=1,\cdots,p\right\},\quad k =1,\cdots,N,
	$$
	and $F^\star$ is supported on a subset of $\{(\bmu,\bSigma)\in\mathbb{R}^p\times\calS:\|\bmu\|_\infty\leq 2a\}$. In addition, we can require that $\int \|\bmu\|^2\mathrm{d}F_0'=\int\|\bmu\|^2\mathrm{d}F^\star$ and still $N\lesssim\left(\log\frac{1}{\epsilon}\right)^{2p}$. Now we claim that there exists some constant $\gamma>0$ such that
	\begin{eqnarray}
		\left\{F=\sum_{k=1}^Nw_k\delta_{(\bmu_k,\bSigma_k)}:(\bmu_k,\bSigma_k)\in E_k,\sum_{k=1}^K|w_k-w_k^\star|<\epsilon\right\}\subset\left\{F:\|f_0-f_F\|_1\leq\gamma\epsilon\left(\log\frac{1}{\epsilon}\right)^{\frac{p}{2}}\right\}.\nonumber
	\end{eqnarray}
	Suppose $F$ is in the LHS of the last display. Observing that $F(E_k)=w_k$, by \textbf{Lemma \ref{lemma:discrete_approximation}}, $F$ must satisfy $\|f_F-f_{F^\star}\|_1\lesssim \epsilon$. By the construction of $F^\star$ and $F_0'$, $\|f_{F_0'}-f_{F^\star}\|_1\lesssim\epsilon\left(\log\frac{1}{\epsilon}\right)^{\frac{p}{2}}$, and $\|f_{F_0'}-f_0\|_1\lesssim \epsilon$. The result follows from the triangle inequality. 

	Now still let $F$ be on the LHS of the last display. Observe that $H(f_0,f_F)\lesssim\|f_F-f_0\|_1^{\frac{1}{2}}\lesssim \epsilon^{\frac{1}{2}}\left(\log\frac{1}{\epsilon}\right)^{\frac{p}{4}}$. Let $B=2\left(\int\|\bmu\|^2\mathrm{d}F_0\right)^{\frac{1}{2}}$. It follows that 
	$$
	{F^\star}(\|\bmu\|> B)\leq\frac{1}{B^2}\int\|\bmu\|^2\mathrm{d}F^\star=\frac{1}{B^2}\int\|\bmu\|^2\mathrm{d}F_0'\leq\frac{1}{B^2}\int\|\bmu\|^2\mathrm{d}F_0=\frac{1}{4},
	$$
	where the second equality is due to the requirement $\int\|\bmu\|^2\mathrm{d}F_0'=\int\|\bmu\|^2\mathrm
	{d}F^\star$, and the last inequality is because the second moment of $F_0'$ is no greater than that of $F_0$. Therefore for $\epsilon<\min(B/\sqrt{p},1/4)$, we have $\|\bmu_k-\bmu_k^\star\|\leq\sqrt{p}\|\bmu_k-\bmu_k^\star\|_\infty<B$, and hence
	$$
	\|\bmu_k\|>2B\quad\Longrightarrow\quad
	\|\bmu_k^\star\|\geq\|\bmu_k\|-\|\bmu_k-\bmu_k^\star\|>2B-B=B.
	$$
	Hence
	\begin{align}
	{F}(\|\bmu\|>2B)&=\sum_{k=1}^Nw_k\mathbb{I}(\|\bmu_k\|>2B)
	\leq\sum_{k=1}^N|w_k-w_k^\star|\mathbb{I}(\|\bmu_k\|>2B)+\sum_{k=1}^Nw_k^\star\mathbb{I}(\|\bmu_k\|>2B)\nonumber\\
	&<\epsilon+\sum_{k=1}^Nw_k^\star\mathbb{I}(\|\bmu_k\|>2B)
	\leq \epsilon+\sum_{k=1}^Nw_k^\star\mathbb{I}(\|\bmu_k^\star\|>B)\nonumber\\
	&=\epsilon+F^\star(\|\bmu_k^\star\|>B)\leq\frac{1}{2}.\nonumber
	\end{align}
	Hence by \textbf{Lemma \ref{lemma:KL_ball_versus_Hellinger_ball}}, we have 
	$$
	\int f_0\left(\log\frac{f_0}{f_F}\right)^2\lesssim\epsilon\left(\log\frac{1}{\epsilon}\right)^{\frac{p+4}{2}},\quad
	\int f_0\log\frac{f_0}{f_F}\lesssim\epsilon\left(\log\frac{1}{\epsilon}\right)^{\frac{p+2}{2}}\leq\epsilon\left(\log\frac{1}{\epsilon}\right)^{\frac{p+4}{2}},
	$$
	and, as a consequence, 
	\begin{eqnarray}
	\left\{f_F:F=\sum_{k=1}^Nw_k\delta_{(\bmu_k,\bSigma_k)}:(\bmu_k,\bSigma_k)\in E_k,\sum_{k=1}^N|w_k-w_k^\star|<\epsilon\right\}\subset B\left(f_0,\eta\epsilon^{\frac{1}{2}}\left(\log\frac{1}{\epsilon}\right)^{\frac{p+4}{4}}\right).\nonumber
	\end{eqnarray}
\end{proof}

\subsection*{Proof of Theorem \ref{thm:contraction_rate}} 
\label{sub:proof_of_theorem_thm:contraction_rate}


\begin{proof}
By \textbf{Proposition \ref{prop:supersmooth_rate}} it suffices to find the prior concentration rate. 
Motivated by \textbf{Lemma \ref{lemma:KL_ball_lower_bound}}, we are interested in finding the prior probability of the following event:
\begin{eqnarray}
\tilde{B}(F^\star,\epsilon):=\left\{f_F:F=\sum_{k=1}^Nw_k\delta_{(\bmu_k,\bSigma_k)}:(\bmu_k,\bSigma_k)\in E_k,\sum_{k=1}^N|w_k-w_k^\star|<\epsilon\right\}.\nonumber
\end{eqnarray}
where $F^\star=\sum_{k=1}^Nw_k^\star\delta_{(\bmu_k^\star,\bSigma_k^\star)}$, $\|\bmu_k^\star\|\leq \kappa \left(\log\frac{1}{\epsilon}\right)^{\frac{1}{2}}$ for $k=1,\cdots,K$ for some $\kappa>0$, $\|\bmu_k^\star-\bmu_{k'}^\star\|_\infty\geq2\epsilon$, $|\lambda_j(\bSigma_k^\star)-\lambda_j(\bSigma_{k'}^\star)|\geq2\epsilon$ whenever $k\neq k'$, $j=1,\cdots,p$, $N\lesssim \left(\log\frac{1}{\epsilon}\right)^{2p}$, and
$$
E_k=\left\{(\bmu,\bSigma)\in\mathbb{R}^p\times\calS:\|\bmu-\bmu_k^\star\|_\infty<\frac{\epsilon}{2},|\lambda_j(\bSigma)-\lambda_j(\bSigma_k^\star)|<\frac{\epsilon}{2},j=1,\cdots,p\right\}.
$$
It follows that
\begin{align}
\Pi(\tilde{B}(F^\star,\epsilon))=\Pi(K=N)\Pi\left(\left.\bigcap_{k=1}^N\{(\bmu_k,\bSigma_k)\in E_k\}\right|K=N\right)
\Pi\left(\|\bw-\bw^\star\|_1<\epsilon\mid K=N\right),\nonumber
\end{align}
where $\bw=(w_1,\cdots,w_N),\bw^\star=(w_1^\star,\cdots,w_N^\star)\in\Delta^N$. Since $(\bmu_k,\bSigma_k)\in E_k$ implies $\|\bmu_k-\bmu_{k'}\|>\epsilon$, for sufficiently small $\epsilon$ we see that
\begin{eqnarray}
\bigcap_{k=1}^N\left\{(\bmu_k,\bSigma_k)\in E_k\right\}\subset
\left\{(\bmu_k,\bSigma_k)_{k=1}^N:h_N(\bmu_1,\cdots,\bmu_N)\geq(c_2\epsilon)^N\right\}\nonumber
\end{eqnarray}
by condition A1 for both $r=1$ and $r=2$. Notice that $\|\bmu_k-\bmu_k^\star\|_\infty<\epsilon/2$ for sufficiently small $\epsilon$ implies that
\begin{eqnarray}
\|\bmu_k\|_\infty\leq\|\bmu_k^\star\|_\infty+\frac{\epsilon}{2}\leq 2\kappa\left(\log\frac{1}{\epsilon}\right)^{\frac{1}{2}}\quad\Longrightarrow\quad\|\bmu_k\|\leq 2\kappa\sqrt{p}\left(\log\frac{1}{\epsilon}\right)^{\frac{1}{2}},
\end{eqnarray}
in which case we have
\begin{eqnarray}
\int_{\|\bmu_k-\bmu_k^\star\|_\infty<\epsilon/2}p(\bmu_k)\mathrm{d}\bmu_k&\geq&B_3\epsilon^p\exp\left[-b_3(2\kappa\sqrt{p})^\alpha\left(\log\frac{1}{\epsilon}\right)^{\frac{\alpha}{2}}\right].\nonumber
\end{eqnarray}
Hence we may proceed to compute
\begin{align}
	&\Pi\left(\bigcap_{k=1}^N\{(\bmu_k,\bSigma_k)\in E_k\}\right)\nonumber\\
	&\qquad\geq\frac{1}{Z_K}\prod_{k=1}^N\left[\int_{\|\bmu_k-\bmu_k^\star\|_\infty<\epsilon/2} c_2\epsilon p(\bmu_k)\mathrm{d}\bmu_k\right]\prod_{k=1}^N\prod_{j=1}^p\left[\int_{\lambda_j(\bSigma_k^\star)-\epsilon/2}^{\lambda_j(\bSigma_k^\star)+\epsilon/2}p_\lambda(\lambda_{jk})\mathrm{d}\lambda_{jk}\right]\nonumber\\
	&\qquad\geq\prod_{k=1}^N\left\{c_2B_3\epsilon^{p+1}\exp\left[-b_3(2\kappa\sqrt{p})^\alpha\left(\log\frac{1}{\epsilon}\right)^{\frac{\alpha}{2}}\right]\right\}\left(\epsilon\min_{\lsigma^2\leq\lambda\leq\usigma^2} p_\lambda(\lambda)\right)^{Np}\nonumber\\
	&\qquad\geq\epsilon^{2Np+N}\left[c_2B_3\min_{\lsigma^2\leq\lambda\leq\usigma^2}p_\lambda(\lambda)^p\right]^{N}\exp\left[-b_3(2\kappa\sqrt{p})^\alpha N\left(\log\frac{1}{\epsilon}\right)^{\frac{\alpha}{2}}\right],\nonumber
\end{align}
For sufficiently small $\epsilon>0$, taking logarithm yields
\begin{eqnarray}
-N\left(\log\frac{1}{\epsilon}\right)^{\frac{\alpha}{2}}\lesssim \log \Pi((\bmu_k,\bSigma_k)\in E_k,k=1,\cdots,N).\nonumber
\end{eqnarray}
Using condition B5 and the fact $N\lesssim\left(\log\frac{1}{\epsilon}\right)^{2p}$, we may further obtain
\begin{eqnarray}
-\left(\log\frac{1}{\epsilon}\right)^{2p+\frac{\alpha}{2}}\lesssim \log \Pi(K=N)+\log \Pi((\bmu_k,\bSigma_k)\in E_k,k=1,\cdots,N)\nonumber.
\end{eqnarray}
By {Lemma A.2} in \cite{ghosal2001entropies}, we have
\begin{eqnarray}
-\left(\log\frac{1}{\epsilon}\right)^{2p+1}\lesssim -N\left(\log\frac{1}{\epsilon}\right)\lesssim\log \Pi\left(w_1,\cdots,w_N:\sum_{k=1}^N|w_k-w_k^\star|<\epsilon\right). \nonumber
\end{eqnarray}
Observing that $\alpha\geq2$, we obtain
\begin{eqnarray}
\exp\left[-c_5\left(\log\frac{1}{\epsilon}\right)^{2p+\frac{\alpha}{2}}\right]\lesssim \Pi\left(\tilde{B}(F^\star,\epsilon)\right)\lesssim \Pi\left(B\left(f_0,\eta\epsilon^{\frac{1}{2}}\left(\log\frac{1}{\epsilon}\right)^{\frac{p+4}{4}}\right)\right)\nonumber
\end{eqnarray}
for some constant $c_5>0$. Since $\log\left[\eta\epsilon^{\frac{1}{2}}\left(\log\frac{1}{\epsilon}\right)^{\frac{p+4}{4}}\right]$ and $\log\epsilon$ are of the same order in the sense that their ratio converges to a positive constant as $\epsilon\to0$, we conclude that
\begin{eqnarray}
\exp\left[-c_5\left(\log\frac{1}{\epsilon}\right)^{2p+\frac{\alpha}{2}}\right]\lesssim \Pi\left(B(f_0,\epsilon)\right).\nonumber
\end{eqnarray}
Setting $\leps_n=(\log n)^{t_0}/\sqrt{n}$, $\ueps_n=(\log n)^{t}/\sqrt{n}$ with $t_0>p+\frac{\alpha}{4}$, $t>t_0+\frac{1}{2}>p+\frac{\alpha+2}{4}$, we see that
\begin{eqnarray}
-n\leps_n^2=-\left(\log n\right)^{2t_0}<-\left(\log\frac{1}{\leps_n}\right)^{2p+\frac{\alpha}{2}}\lesssim \log\Pi\left(B(f_0,\leps_n)\right).\nonumber
\end{eqnarray}
Hence \eqref{ineq:prior_concentration} is satisfied with $\leps_n=(\log n)^{t_0}/\sqrt{n}$, $t_0>p+\frac{\alpha}{4}$. The proof is thus completed by applying {Proposition \ref{prop:supersmooth_rate}} and {Theorem 3} in \cite{kruijer2010adaptive}. 
\end{proof}

\section{Proofs for the Model Complexity} 
\label{sec:proofs_for_the_model_complexity}
\subsection*{Preliminary Lemmas for Theorem \ref{thm:model_complexity}}
The proof of Theorem \ref{thm:model_complexity} is seemingly daunting but quite straightforward: By repeatedly using Jensen's inequality, we directly bound the marginal density $p(\by_1,\cdots,\by_n)$ of the data and the joint density $p(\by_1,\cdots,\by_n,K)$ between the data and $K$ under the RGM prior. To keep track of the road map of the proof, we begin with several preliminary lemmas, the proofs of which are deferred to the end of this section. 
To avoid the confusion of using the parameter $\bmu_1,\cdots,\bmu_K$ in the RGM prior and the dummy variable $\bmu$ in the underlying true density $f_0(\by)=\int_{\mathbb{R}^p}\phi_{\bSigma_0}(\by-\bmu)F_0(\mathrm{d}\bmu)$, we shall write $f_0(\by)=\int_{\mathbb{R}^p}\phi_{\bSigma_0}(\by-\mb)F_0(\mathrm{d}\mb)$. 
For convenience we use the following notation (only for the proof of \textbf{Theorem \ref{thm:model_complexity}} in this section): $\by_i=(y_{i1},\cdots,y_{ip})\transpose$, $\by_{1:n}=(\by_1,\cdots,\by_n)$, $\mb_i=(m_{i1},\cdots,m_{ip})\transpose{}$, $\mb_{1:n}=(\mb_1,\cdots,\mb_n)$, $\bmu_k=(\mu_{k1},\cdots,\mu_{kp})\transpose{}$, $\bmu_{1:K}=(\bmu_1,\cdots,\bmu_K)$, $z_{1:n}=(z_1,\cdots,z_n)$, 
$
n_k=\sum_{i=1}^n\mathbb{I}(z_i=k)$, $\by_{(k)j}=(y_{ij}:z_i=k)\transpose\in\mathbb{R}^{n_k}$, $\mb_{(k)j}=\left(m_{ij}:z_i=k\right)\transpose{}$, $\mathbf{1}_{n_k}=(1,\cdots,1)\transpose\in\mathbb{R}^{n_k}$, $\widehat\by_i=\bU\transpose{}\by_i=\left(\widehat y_{i1},\cdots,\widehat y_{ip}\right)\transpose{}$, $\widehat\by_{1:n}=\left(\widehat\by_1,\cdots,\widehat\by_n\right)$, $\widehat\bmu_k = \bU\transpose{}\bmu_k=\left(\widehat\mu_{k1},\cdots,\widehat\mu_{kp}\right)\transpose{}$, $\widehat\bmu_{1:K} = \left(\widehat\bmu_1,\cdots,\widehat\bmu_K\right)$, $\widehat\mb_i=\bU\transpose{}\mb_i=(\widehat m_{i1},\cdots,\widehat m_{ip})\transpose{}$, $\widehat\mb_{1:n}=(\widehat\mb_1,\cdots,\widehat\mb_n)$, $\bLambda_0=\bU\transpose{}\bSigma_0\bU = \mathrm{diag}(\sigma_1^2,\cdots,\sigma_p^2)$, $\widehat{\by}_{(k)j}=\left(\widehat y_{ij}:z_i=k\right)\transpose{}$, $\widehat\mb_{(k)j}=\left(\widehat m_{ij}:z_i=k\right)\transpose{}$, $\kappa_j^2=\sigma_j^2/\tau^2$, and $F_0^{(n)}$ be the $n$-fold product measure of $F_0$  over $\mathbb{R}^{p\times n}$. 
\begin{lemma}\label{lemma:marginal_likelihood_upper_bound}
Assume the conditions in Theorme \ref{thm:model_complexity} hold. Then for $K\geq 3$
\begin{align}
&p(\by_{1:n}|z_{1:n},K)\nonumber\\
&\quad\leq\frac{1}{Z_K}\prod_{j=1}^p\prod_{k=1}^K\phi\left(\widehat\by_{(k)j}\middle|\zero_{n_k},\sigma_j^2\eye_{n_k}+\tau^2\one_{n_k}\one_{n_k}\transpose{}\right)\nonumber\\
&\qquad \times{K\choose 2}^{-1}\sum_{k<k'}g\left(\left[\sum_{j=1}^p\left(\frac{\one_{n_k}\transpose{}\widehat\by_{(k)j}}{n_k+\kappa_j^2}-\frac{\one_{n_{k'}}\transpose{}\widehat\by_{(k')j}}{n_{k'}+\kappa_j^2}\right)^2+\sum_{j=1}^p\sigma_j^2\left(\frac{1}{n_k+\kappa_j^2}+\frac{1}{n_{k'}+\kappa_j^2}\right)\right]^{\frac{1}{2}}\right)\nonumber.
\end{align}
\end{lemma}

\begin{lemma}\label{lemma:marginal_likelihood_lower_bound}
Assume the conditions in Theorem \ref{thm:model_complexity} hold. Then the marginal density of the data $\by_1,\cdots,\by_n$ under the $\mathrm{RGM}$ prior satisfies:
\begin{enumerate}
	\item[(i)] If $f\sim\mathrm{RGM}_1(1,g,\phi(\bmu|\zero,\tau^2\eye),\delta_{\bSigma_0},p(K))$, 
	\emph{i.e.} $h$ is of the form of \eqref{eqn:repulsive_function_1}, then
	\begin{align}
	p(\by_{1:n})\geq C(\lambda)\exp\left[-\frac{n\tau^2}{2}\mathrm{tr}\left(\bSigma_0^{-1}\right)\right]\prod_{i=1}^n\phi_{\bSigma_0}(\by_i)
	{\Omega(\mathrm{e}^\lambda-1)}{\left(1+g_0^{\frac{2}{3}}\delta(\tau)\right)^{-\frac{3}{2}}}\nonumber;
	\end{align}
	\item[(ii)] If $f\sim\mathrm{RGM}_2(1,g,\phi(\bmu|\zero,\tau^2\eye),\delta_{\bSigma_0},p(K))$,
	\emph{i.e.} $h$ is of the form of \eqref{eqn:repulsive_function_2}, then
	\begin{align}
	p(\by_{1:n})\geq C(\lambda)\exp\left[-\frac{n\tau^2}{2}\mathrm{tr}\left(\bSigma_0^{-1}\right)\right]\prod_{i=1}^n\phi_{\bSigma_0}(\by_i){\Omega(\mathrm{e}^\lambda -1)}{(1+\delta(\tau)\sqrt{g_0})^{-1}}\nonumber.
	\end{align}
\end{enumerate}
Here $C(\lambda)$ is a constant only depending on $\lambda$, and $\delta(\tau)<1$ when $\tau$ is sufficiently large. 
\end{lemma}

\begin{lemma}\label{lemma:conditional_density_upper_bound}
Assume the conditions in Theorem \ref{thm:model_complexity} hold. Then
	\begin{align}
	&\int_{\mathbb{R}^p}\cdots\int_{\mathbb{R}^p}\frac{p(\by_{1:n}|z_{1:n},K)}{p(\by_{1:n})}\prod_{i=1}^n\phi(\by_i|\mb_i,\bSigma_0)\mathrm{d}\by_1\cdots\mathrm{d}\by_n\nonumber\\
	&\qquad\leq 
	C(\lambda)
	\exp\left[\frac{n\tau^2}{2}\mathrm{tr}\left(\bSigma_0^{-1}\right)\right]\frac{
	\omega(g_0)
	}{Z_K\Omega(\mathrm{e}^\lambda - 1)}
	\nonumber\\
	&\qquad\quad \times{K\choose 2}^{-1}\sum_{1\leq k<k'\leq K}
	g\left(\left[\sum_{j=1}^p\frac{1}{\kappa_j^4}\left(\one_{n_k}\transpose{}\widehat\mb_{(k)j}-\one_{n_{k'}}\transpose{}\widehat\mb_{(k')j}\right)^2+2p\tau^2\right]^{\frac{1}{2}}\right)
	\nonumber,
	\end{align}
	where $\omega(g_0)=\left(1+\delta(\tau)g_0^{\frac{2}{3}}\right)^{\frac{3}{2}}$ if $r=1$, and $\omega(g_0) = 1+\delta(\tau)\sqrt{g_0}$ if $r=2$. 
\end{lemma}

\begin{lemma}\label{lemma:expected_value_equation}
Assume the conditions in Theorem \ref{thm:model_complexity} hold. Then
\begin{align}
&\sum_{j=1}^p\frac{1}{\kappa_j^4}\mathbb{E}_\bz\left[\int_{\mathbb{R}^{p\times n}}\int_{\mathbb{R}^{p\times n}}\left(\one_{n_k}\transpose{}\widehat\mb_{(k)j}-\one_{n_{k'}}\transpose{}\widehat\mb_{(k')j}\right)^2F_0^{(n_k+n_{k'})}\left(\mathrm{d}\mb_{(k)j}\mathrm{d}\mb_{(k')j}\right)\right]
\nonumber\\&\qquad 
= \tau^4\mathbb{E}_0\left(\mb\transpose{}\bSigma_0^{-2}\mb\right)\frac{2n}{K}\nonumber,
\end{align}
where $\mathbb{E}_\bz$ is the expected value with respect to $p(z_{1:n}|K)$. 
\end{lemma}

\subsection*{Proof of Theorem \ref{thm:model_complexity}} 
\label{sub:proof_of_theorem_thm:model_complexity}
\begin{proof}
By Fubini's theorem we directly write
	\begin{align}
	&\mathbb{E}_0\left[\Pi(K>N|\by_1,\cdots,\by_n)\right]\nonumber\\
	&\qquad=\int_{\mathbb{R}^p}\cdots\int_{\mathbb{R}^p}\sum_{K=N+1}^\infty\frac{\mathbb{E}_{\bz}\left[p(\by_{1:n}|z_{1:n},K)\right]\pi(K)}{p(\by_{1:n})}\prod_{i=1}^n\int_{\mathbb{R}^p}\phi_{\bSigma{}_0}(\by_i-\mb_i)F_0(\mathrm{d}\mb_i)\mathrm{d}\by_1\cdots\mathrm{d}\by_n\nonumber\\
	&\qquad=\sum_{K=N+1}^\infty\pi(K)
	\mathbb{E}_\bz\left\{\int_{\mathbb{R}^{p\times n}}\left[\int_{\mathbb{R}^{p\times n}}\frac{p(\by_{1:n}|z_{1:n},K)}{p(\by_{1:n})}\prod_{i=1}^n\phi(\by_i|\mb_i,\bSigma{}_0)\mathrm{d}\by_{1:n}\right] F_0^n(\mathrm{d}\mb_{1:n})\right\}\nonumber.	
	\end{align}
	We may without loss of generality assume that $h$ is of the form of \eqref{eqn:repulsive_function_1}, since the following proof directly applies to the case where $h$ is of the form \eqref{eqn:repulsive_function_2}. By Lemma \ref{lemma:conditional_density_upper_bound} the quantity in the square bracket is upper bounded by 
	\begin{align}
	&C(\lambda)\exp\left[\frac{n\tau^2}{2}\mathrm{tr}\left(\bSigma_0^{-1}\right)\right]\frac{\left(1+\delta(\tau)g_0^{\frac{2}{3}}\right)^{\frac{3}{2}}}{Z_K\Omega(\mathrm{e}^\lambda - 1)}\nonumber\\
	&\quad\times{K\choose 2}^{-1}\sum_{1\leq k<k'\leq K}
	g\left(\left[\sum_{j=1}^p\frac{1}{\kappa_j^4}\left(\one_{n_k}\transpose{}\widehat\mb_{(k)j}-\one_{n_{k'}}\transpose{}\widehat\mb_{(k')j}\right)^2+2p\tau^2\right]^{\frac{1}{2}}\right)
	\nonumber,
	\end{align}
	where $C(\lambda)$ is a constant only depending on $\lambda$. Observing that $x\mapsto g(\sqrt{x})$ is concave, we directly obtain by Jensen's inequality and Lemma \ref{lemma:expected_value_equation} that 
	\begin{align}
	&\mathbb{E}_0\left[\Pi(K>N|\by_1,\cdots,\by_n)\right]\nonumber\\
	&\quad\leq C_1(\lambda)\exp\left[\frac{n\tau^2}{2}\mathrm{tr}\left(\bSigma{}_0^{-1}\right)\right]\sum_{K=N+1}^\infty\frac{\lambda^K}{(\mathrm{e}^\lambda - 1)K!}\left(1+\delta(\tau)g_0^{\frac{2}{3}}\right)^{\frac{3}{2}}{K\choose 2}^{-1}
	\nonumber\\
	&\qquad\times \sum_{k<k'}g\left(
	\left\{2p\tau^2+\sum_{j=1}^p\frac{1}{\kappa_j^4}\mathbb{E}_\bz\left[\iint\left(\one_{n_k}\transpose{}\widehat\mb_{(k)j}-\one_{n_{k'}}\transpose{}\widehat\mb_{(k')j}\right)^2F_0^{(n_k+n_{k'})}\left(\mathrm{d}\mb_{(k)j}\mathrm{d}\mb_{(k')j}\right)\right]\right\}^{\frac{1}{2}}
	\right)\nonumber\\
	&\quad\leq C_1(\lambda)\exp\left[\frac{n\tau^2}{2}\mathrm{tr}\left(\bSigma{}_0^{-1}\right)\right]\sum_{K=N+1}^\infty\frac{\lambda^K}{(\mathrm{e}^\lambda - 1)K!}\left(1+\delta(\tau)g_0^{\frac{2}{3}}\right)^{\frac{3}{2}}\nonumber\\
	&\qquad\times {K\choose 2}^{-1}\sum_{1\leq k<k'\leq K}g\left(
	\left[2p\tau^2+\frac{2n}{N}\tau^4\mathbb{E}_0\left(\mb\transpose{}\bSigma_0^{-2}\mb\right)\right]^{\frac{1}{2}}
	\right)\nonumber.
	\end{align}
	The proof is completed by observing the fact that
	\begin{align}
	& g\left(
	\left[2p\tau^2+\frac{2n}{N}\tau^4\mathbb{E}_0\left(\mb\transpose{}\bSigma_0^{-2}\mb\right)\right]^{\frac{1}{2}}
	\right)
	= \frac{\left[2p\tau^2+\frac{2n}{N}\tau^4\mathbb{E}_0\left(\mb\transpose{}\bSigma_0^{-2}\mb\right)\right]^{\frac{1}{2}}}{g_0+\left[2p\tau^2+\frac{2n}{N}\tau^4\mathbb{E}_0\left(\mb\transpose{}\bSigma_0^{-2}\mb\right)\right]^{\frac{1}{2}}}\nonumber.
	\end{align}
\end{proof}
\cbend

\subsection*{Proof of Corollary \ref{corr:model_complexity}} 
\label{sub:proof_of_corollary_corr:model_complexity}
\begin{proof}
By \textbf{Theorem 5}, we see that for any sufficiently large $N$, 
\begin{align}
\mathbb{E}_0\left[\Pi\left(K>N|\by_1,\cdots,\by_n\right)\right]&\lesssim \exp\left[\frac{n\tau^2}{2}\mathrm{tr}\left(\bSigma_0^{-1}\right)\right]\sum_{K=N+1}^\infty\frac{\lambda^K}{(\mathrm{e}^\lambda - 1)K!}\nonumber\\
&\leq\exp\left[\frac{n\tau^2}{2}\mathrm{tr}\left(\bSigma_0^{-1}\right)+\frac{1}{2}N\log N\right]\nonumber.
\end{align}
Since $\liminf_{n\to\infty}K_n/n>0$, then there exists some $\delta_0>0$, such that $K_n\geq\delta_0n$ for sufficiently large $n$. Hence for sufficiently large $n$, 
\begin{eqnarray}
\frac{n\tau^2}{2}\mathrm{tr}(\bSigma_0^{-1})-\frac{1}{2}K_n\log K_n\leq\frac{n\tau^2}{2}\mathrm{tr}(\bSigma_0^{-1})-\frac{\delta_0}{2} n(\log \delta_0 n)\to-\infty\nonumber
\end{eqnarray}
as $n\to\infty$. Since
\begin{eqnarray}
\limsup_{n\to\infty}\mathbb{E}\left[\Pi(K\geq K_n\mid\by_1,\cdots,\by_n)\right]\leq\lim_{n\to\infty}\exp\left[\frac{n\tau^2}{2}\mathrm{tr}(\bSigma_0^{-1})-\frac{1}{2}K_n\log K_n\right]=0.\nonumber
\end{eqnarray}
By Markov's inequality, for any $\epsilon>0$, 
\begin{eqnarray}
\mathbb{P}_0\left[\Pi(K\geq K_n\mid\by_1,\cdots,\by_n)>\epsilon\right]\leq\frac{1}{\epsilon}\mathbb{E}\left[\Pi(K\geq K_n\mid\by_1,\cdots,\by_n)\right]\to 0\nonumber
\end{eqnarray}
as $n\to\infty$, and the proof is thus completed.
\end{proof}

\subsection*{Proofs of Preliminary Lemmas}

\begin{proof}[Proof of \textbf{Lemma \ref{lemma:marginal_likelihood_upper_bound}}]
	Directly compute
	\begin{align}
	&p(\by_{1:n}| z_{1:n},K)\nonumber\\
	&\qquad=\frac{1}{Z_K}\int_{\mathbb{R}^p}\cdots\int_{\mathbb{R}^p}h(\bmu_1,\cdots,\bmu_K)\prod_{k=1}^K\left[\prod_{i:z_i=k}\phi_{\bSigma_0}(\by_i-\bmu_k)\right]p(\bmu_k)\mathrm{d}\bmu_1\cdots\mathrm{d}\bmu_K\nonumber\\
	&\qquad=\frac{1}{Z_K}\int_{\mathbb{R}^p}\cdots\int_{\mathbb{R}^p}h(\bmu_1,\cdots,\bmu_K)\prod_{j=1}^p\left\{\prod_{k=1}^K\left[\prod_{i:z_i=k}\phi_{\sigma_j}(\widehat{y}_{ij}-\widehat{\mu}_{kj})\right]p(\widehat{\mu}_{kj})\right\}\mathrm{d}\widehat\bmu_1\cdots\mathrm{d}\widehat\bmu_K\nonumber\\	
	&\qquad=\frac{1}{Z_K}\prod_{j=1}^p\prod_{k=1}^K\phi\left(\widehat\by_{(k)j}\middle|\zero_{n_k},\sigma_j^2\eye_{n_k}+\tau^2\one_{n_k}\one_{n_k}\transpose{}\right)\nonumber\\
	&\qquad\quad\times \int_{\mathbb{R}^p}\cdots\int_{\mathbb{R}^p}h(\bmu_1,\cdots,\bmu_K)\prod_{j=1}^p\prod_{k=1}^K\phi\left(\widehat\mu_{kj}\middle|\frac{\one_{n_k}\transpose{}\widehat\by_{(k)j}}{n_k+\kappa_j^2},\frac{\sigma_j^2}{n_k+\kappa_j^2}\right)\mathrm{d}\bmu_1\cdots\mathrm{d}\bmu_K\nonumber\\
	&\qquad = \frac{1}{Z_K}\prod_{j=1}^p\prod_{k=1}^K\phi\left(\widehat\by_{(k)j}\middle|\zero_{n_k},\sigma_j^2\eye_{n_k}+\tau^2\one_{n_k}\one_{n_k}\transpose{}\right)\nonumber\\
	&\qquad\quad\times \int_{\mathbb{R}^p}\cdots\int_{\mathbb{R}^p}h(\widehat\bmu_1,\cdots,\widehat\bmu_K)\prod_{j=1}^p\prod_{k=1}^K\phi\left(\widehat\mu_{kj}\middle|\frac{\one_{n_k}\transpose{}\widehat\by_{(k)j}}{n_k+\kappa_j^2},\frac{\sigma_j^2}{n_k+\kappa_j^2}\right)\mathrm{d}\widehat\bmu_1\cdots\mathrm{d}\widehat\bmu_K\nonumber.
	\end{align}
	where 
	we have used the fact that $h$ is unitary invariant. 	For any $k\neq k'$, denote
	\begin{align}
	\Delta_{kk'j}=\widehat\mu_{kj}-\widehat\mu_{k'j},\quad x_{kk'j}=\frac{\one_{n_k}\transpose{}\widehat\by_{(k)j}}{n_k+\kappa_j^2}-\frac{\one_{n_{k'}}\transpose{}\widehat\by_{(k')j}}{n_{k'}+\kappa_j^2},\quad \sigma_{kk'j}^2=\sigma_j^2\left(\frac{1}{n_k+\kappa_j^2}+\frac{1}{n_{k'}+\kappa_j^2}\right)\nonumber.
	\end{align}
	\begin{itemize}
		\item Suppose $h$ is of the form \eqref{eqn:repulsive_function_1}. Since $K\geq 3$, then $K(K-1)/2\geq K$ and hence by the geometric-algorithmic mean inequality
	\begin{align}
	h(\widehat\bmu_1,\cdots,\widehat\bmu_K)&
	=\left[\prod_{1\leq k<k'\leq K}\left(\frac{\|\widehat\bmu_k-\widehat\bmu_{k'}\|}{g_0+\|\widehat\bmu_k-\widehat\bmu_{k'}\|}\right)\right]^{\frac{1}{K}}\nonumber\\
	&\leq \left[\prod_{1\leq k<k'\leq K}\left(\frac{\|\widehat\bmu_k-\widehat\bmu_{k'}\|}{g_0+\|\widehat\bmu_k-\widehat\bmu_{k'}\|}\right)\right]^{\frac{2}{K(K-1)}}
	\nonumber\\
	&\leq{K\choose 2}^{-1}\sum_{k<k'}g\left(\|\widehat\bmu_k-\widehat\bmu_{k'}\|\right)\nonumber.
	\end{align}
	Notice that $g$ is concave, it follows by Jensen's inequality that
	\begin{align}
	&\int_{\mathbb{R}^p}\cdots\int_{\mathbb{R}^p}h\left(\widehat\bmu_1,\cdots,\widehat\bmu_K\right)\prod_{j=1}^p\prod_{k=1}^K\phi\left(\widehat\mu_{kj}\middle|\frac{\one_{n_k}\transpose{}\widehat\by_{(k)j}}{n_k+\kappa_j^2},\frac{\sigma_j^2}{n_k+\kappa_j^2}\right)\mathrm{d}\widehat\bmu_1\cdots\mathrm{d}\widehat\bmu_K\nonumber\\
	&\quad\leq{K\choose 2}^{-1}\sum_{k<k'}\int_{\mathbb{R}^p}\cdots \int_{\mathbb{R}^p}g(\|\bmu_k-\bmu_{k'}\|)\prod_{j=1}^p\prod_{k=1}^K\phi\left(\widehat\mu_{kj}\middle|\frac{\one_{n_k}\transpose{}\widehat\by_{(k)j}}{n_k+\kappa^2_j},\frac{\sigma_j^2}{n_k+\kappa_j^2}\right)\mathrm{d}\widehat\bmu_1\cdots\mathrm{d}\widehat\bmu_K\nonumber\\
	&\quad={K\choose 2}^{-1}\sum_{k<k'}\int_{\mathbb{R}^p}\cdots \int_{\mathbb{R}^p}g\left(\left(\sum_{j=1}^p\Delta_{kk'j}^2\right)^{\frac{1}{2}}\right)\prod_{j=1}^p\prod_{k=1}^K\phi\left(\widehat\mu_{kj}\middle|\frac{\one_{n_k}\transpose{}\widehat\by_{(k)j}}{n_k+\kappa^2_j},\frac{\sigma_j^2}{n_k+\kappa_j^2}\right)\mathrm{d}\widehat\bmu_1\cdots\mathrm{d}\widehat\bmu_K\nonumber\\
	&\quad\leq {K\choose 2}^{-1}\sum_{1\leq k<k'\leq K}g\left(\left[\sum_{j=1}^p\int_{\mathbb{R}}\Delta_{kk'j}^2\phi\left(\Delta_{kk'j}\middle|
	x_{kk'j},\sigma_{kk'j}^2
	\right)\mathrm{d}\Delta_{kk'j}\right]^{\frac{1}{2}}
	\right)\nonumber\\
	&\quad= {K\choose 2}^{-1}\sum_{1\leq k<k'\leq K}g\left(\left[\sum_{j=1}^p\left(x_{kk'j}^2+\sigma_{kk'j}^2\right)\right]^{\frac{1}{2}}\right)\nonumber.
	\end{align}
	\item Suppose $h$ is of the form \eqref{eqn:repulsive_function_2}. Since for $K\geq 3$, the following holds:
	\[
	h\left(\widehat\bmu_1,\cdots,\widehat\bmu_K\right)
	=\min_{1\leq k<k'\leq K}g(\|\bmu_k-\bmu_{k'}\|)\leq {K\choose 2}^{-1}\sum_{1\leq k<k'\leq K}g(\|\bmu_k-\bmu_{k'}\|),
	\]
	then the above derivation directly applies. 
	\end{itemize}
	The proof is thus completed. 
\end{proof}

\clearpage
\begin{proof}[Proof of \textbf{Lemma \ref{lemma:marginal_likelihood_lower_bound}}]
First we obtain directly by Jensen's inequality that
\begin{align}
	\log p(\by_{1:n}|z_{1:n},K)
	&\geq -\log Z_K+\int_{\mathbb{R}^p}\cdots\int_{\mathbb{R}^p}\log h(\bmu_1,\cdots,\bmu_K)\prod_{k=1}^K p(\bmu_k)\mathrm{d}\bmu_1\cdots\mathrm{d}\bmu_K\nonumber\\
	&\quad +\sum_{k=1}^K\sum_{i:z_i=k}\int_{\mathbb{R}^p}\log\phi_{\bSigma_0}(\by_i-\bmu_k)p(\bmu_k)\mathrm{d}\bmu_k\nonumber.
	\end{align}
	Now compute
	\begin{align}
	&\sum_{k=1}^K\sum_{i:z_i=k}\int_{\mathbb{R}^p}\log\phi_{\bSigma_0}(\by_i-\bmu_k)p(\bmu_k)\mathrm{d}\bmu_k\nonumber\\
	&\qquad=\sum_{k=1}^K\sum_{i:z_i=k}\left[-\frac{1}{2}\log(\det(2\pi\bSigma_0))-\frac{1}{2}\by_i\transpose{}\bSigma_0^{-1}\by_i-\frac{1}{2}\mathrm{tr}(\bSigma_0^{-1}\tau^2\eye_p)\right]\nonumber\\
	&\qquad=\sum_{i=1}^n\log\phi_{\bSigma_0}(\by_i)-\frac{n\tau^2}{2}\mathrm{tr}\left(\bSigma{}_0^{-1}\right).\nonumber
	\end{align}
	\begin{itemize}
		\item Suppose $h$ is of the form \eqref{eqn:repulsive_function_1}. Then by Jensen's inequality we obtain
		\begin{align}
		&\int_{\mathbb{R}^p}\cdots\int_{\mathbb{R}^p}\log h(\bmu_1,\cdots,\bmu_K)\prod_{k=1}^K p(\bmu_k)\mathrm{d}\bmu_1\cdots\mathrm{d}\bmu_K\nonumber\\
		&\qquad=-\int_{\mathbb{R}^p}\cdots\int_{\mathbb{R}^p}\log\left(1+\max_{k\neq k'}\frac{g_0}{\|\bmu_k-\bmu_{k'}\|}\right) \prod_{k=1}^Kp(\bmu_k)\mathrm{d}\bmu_1\cdots\mathrm{d}\bmu_K\nonumber\\
		&\qquad\geq -\frac{3}{2}\int_{\mathbb{R}^p}\cdots\int_{\mathbb{R}^p}\log\left[\left(1+\sum_{1\leq k<k'\leq K}\frac{g_0}{\|\bmu_k-\bmu_{k'}\|}\right)^{\frac{2}{3}}\right]\prod_{k=1}^K p(\bmu_k)\mathrm{d}\bmu_1\cdots\mathrm{d}\bmu_K\nonumber\\
		&\qquad\geq -\frac{3}{2}\int_{\mathbb{R}^p}\cdots\int_{\mathbb{R}^p}\log\left[1+\sum_{1\leq k<k'\leq K}\left(\frac{g_0}{\|\bmu_k-\bmu_{k'}\|}\right)^{\frac{2}{3}}\right]\prod_{k=1}^K p(\bmu_k)\mathrm{d}\bmu_1\cdots\mathrm{d}\bmu_K\nonumber\\
		&\qquad\geq -\frac{3}{2}\log\left[1+\int_{\mathbb{R}^p}\cdots\int_{\mathbb{R}^p}\sum_{1\leq k<k'\leq K}\left(\frac{g_0}{\|\bmu_k-\bmu_{k'}\|}\right)^{\frac{2}{3}}\prod_{k=1}^K p(\bmu_k)\mathrm{d}\bmu_1\cdots\mathrm{d}\bmu_K\right]\nonumber\\
		&\qquad\geq -\frac{3}{2}\log\left[1+g_0^{\frac{2}{3}}K^2\frac{1}{2}
		\left(\frac{1}{2\tau^2}\right)^{\frac{1}{3}}
		\int_0^\infty\Delta^{-\frac{1}{3}}
		\frac{1}{2^{\frac{p}{2}}\Gamma\left(\frac{p}{2}\right)}\Delta^{\frac{p}{2}-1}\exp\left(-\frac{\Delta}{2}\right)
		\mathrm{d}\Delta
		\right]\nonumber\\
		&\qquad=-\frac{3}{2}\log\left(1+K^2\delta(\tau)g_0^{\frac{2}{3}}\right)\nonumber\\
		&\qquad\geq -\log\left[\left(1+\delta(\tau)g_0^{\frac{2}{3}}\right)^{\frac{3}{2}}K^3\right]\nonumber
		\end{align}
		where $\Delta\overset{\mathcal{L}}{=}\frac{1}{2\tau^2}\|\bmu_k-\bmu_{k'}\|^2\sim\chi^2(p)$, and $\delta(\tau)$ is a constant with $\delta(\tau)<1$ for sufficiently large $\tau$.
		Hence we can integrate $p(\by_{1:n}|z_{1:n},K)$ against $\pi(z_{1:n}|w_{1:K},K)$, $\pi(w_{1:K}|K)$ and obtain
		\begin{align}
		p(\by_{1:n}|K)\geq\frac{1}{Z_K}\exp\left[-\frac{n\tau^2}{2}\mathrm{tr}\left(\bSigma_0^{-1}\right)\right]\prod_{i=1}^n\phi_{\bSigma_0}(\by_i)\frac{1}{\left(1+g_0^{\frac{2}{3}}\delta(\tau)\right)^{\frac{3}{2}}K^3},
		\end{align}
		and hence by the fact that $\mathbb{E}\left(K^{-3}\right)\geq \left[\mathbb{E}(K^3)\right]^{-1}$,
		\begin{align}
		p(\by_{1:n})
		&\geq\exp\left[-\frac{n\tau^2}{2}\mathrm{tr}\left(\bSigma_0^{-1}\right)\right]\prod_{i=1}^n\phi_{\bSigma_0}(\by_i)
		\left[\frac{\Omega(\mathrm{e}^\lambda-1)}{\left(1+\delta(\tau)g_0^{\frac{2}{3}}\right)^{\frac{3}{2}}}\right]\left[\sum_{K=1}^\infty\frac{K^3\lambda^K}{(\mathrm{e}^\lambda-1)K!}\right]^{-1}\nonumber\\
		&=\exp\left[-\frac{n\tau^2}{2}\mathrm{tr}\left(\bSigma_0^{-1}\right)\right]\prod_{i=1}^n\phi_{\bSigma_0}(\by_i)
		\left[\frac{\Omega(\mathrm{e}^\lambda-1)^2}{\mathrm{e}^\lambda\left(1+\delta(\tau)g_0^{\frac{2}{3}}\right)^{\frac{3}{2}}}
		\right]\left[\sum_{K=0}^\infty\frac{K^3\mathrm{e}^{-\lambda}\lambda^K}{K!}\right]^{-1}
		\nonumber\\
		&=C(\lambda)\exp\left[-\frac{n\tau^2}{2}\mathrm{tr}\left(\bSigma_0^{-1}\right)\right]\prod_{i=1}^n\phi_{\bSigma_0}(\by_i)
		{\Omega(\mathrm{e}^\lambda-1)^2}{\left(1+\delta(\tau)g_0^{\frac{2}{3}}\right)^{-\frac{3}{2}}}\nonumber,
		\end{align}
		where $C(\lambda)$ only depends on $\lambda$, and $\delta(\tau)<1$ for sufficiently large $\tau$.
		\item Suppose $h$ is of the form \eqref{eqn:repulsive_function_2}. Then by Jensen's inequality we obtain for $K\geq 2$
		\begin{align}
		&\int_{\mathbb{R}^p}\cdots\int_{\mathbb{R}^p}\log h(\bmu_1,\cdots,\bmu_K)\prod_{k=1}^K p(\bmu_k)\mathrm{d}\bmu_1\cdots\mathrm{d}\bmu_K\nonumber\\
		&\qquad=-\frac{1}{K}\sum_{1\leq k<k'\leq K}\int_{\mathbb{R}^p}\cdots\int_{\mathbb{R}^p}\log\left(1+\frac{g_0}{\|\bmu_k-\bmu_{k'}\|}\right) \prod_{k=1}^Kp(\bmu_k)\mathrm{d}\bmu_1\cdots\mathrm{d}\bmu_K\\
		&\qquad= -\frac{1}{K}\sum_{k<k'}2\int_{\mathbb{R}^p}\cdots\int_{\mathbb{R}^p}\log\left(1+\frac{g_0}{\|\bmu_k-\bmu_{k'}\|}\right)^{\frac{1}{2}}\prod_{k=1}^K p(\bmu_k)\mathrm{d}\bmu_1\cdots\mathrm{d}\bmu_K\nonumber\\
		&\qquad\geq -\frac{1}{K}2\sum_{k<k'}\int_{\mathbb{R}^p}\cdots\int_{\mathbb{R}^p}\log\left[1+\left(\frac{g_0}{\|\bmu_k-\bmu_{k'}\|}\right)^{\frac{1}{2}}\right]\prod_{k=1}^Kp(\bmu_k)\mathrm{d}\bmu_1\cdots\mathrm{d}\bmu_K\nonumber\\
		&\qquad\geq -\frac{1}{K}2\sum_{k<k'}\log\left[1+\left(\frac{g_0^2}{2\tau^2}\right)^{\frac{1}{4}}\int_0^\infty\Delta^{-\frac{1}{4}}\frac{1}{2^{\frac{p}{2}}\Gamma\left(\frac{p}{2}\right)}\Delta^{\frac{p}{2}-1}\exp\left(-\frac{\Delta}{2}\right)
		\mathrm{d}\Delta\right]\nonumber\\
		&\qquad = -(K-1)\log\left(1+\sqrt{g_0}\delta(\tau)\right)\nonumber\\
		&\qquad \geq -K\log\left(1+\sqrt{g_0}\delta(\tau)\right)\nonumber,
		\end{align}
		where $\delta(\tau)<1$ when $\tau$ is sufficiently large.
		When $K=1$ the above inequality still holds. Hence we can integrate $p(\by_{1:n}|z_{1:n},K)$ against $\pi(z_{1:n}|w_{1:K},K)$, $\pi(w_{1:K}|K)$ and obtain
		\begin{align}
		p(\by_{1:n}|K)\geq\frac{1}{Z_K}\exp\left[-\frac{n\tau^2}{2}\mathrm{tr}\left(\bSigma{}_0^{-1}\right)\right]\prod_{i=1}^n\phi_{\bSigma_0}(\by_i)\exp\left[-\log\left(1+g_0^{\frac{1}{2}}\delta(\tau)\right)K\right],\nonumber
		\end{align}
		and hence,
		\begin{align}
		p(\by_{1:n})&\geq\exp\left[-\frac{n\tau^2}{2}\mathrm{tr}\left(\bSigma{}_0^{-1}\right)\right]\prod_{i=1}^n\phi_{\bSigma_0}(\by_i)\sum_{K=1}^\infty\exp\left[-\log\left(1+g_0^{\frac{1}{2}}\delta(\tau)\right)K\right]{\pi(K)}\nonumber\\
		&=\exp\left[-\frac{n\tau^2}{2}\mathrm{tr}\left(\bSigma{}_0^{-1}\right)\right]\prod_{i=1}^n\phi_{\bSigma_0}(\by_i)\Omega\mathrm{e}^\lambda\sum_{K=1}^\infty\exp\left[-\log\left(1+g_0^{\frac{1}{2}}\delta(\tau)\right)K\right]\frac{\mathrm{e}^{-\lambda}\lambda^K}{K!}\nonumber\\
		&=\exp\left[-\frac{n\tau^2}{2}\mathrm{tr}\left(\bSigma{}_0^{-1}\right)\right])\prod_{i=1}^n\phi_{\bSigma_0}(\by_i)\Omega\mathrm{e}^\lambda\left[\exp\left(-\frac{\lambda \sqrt{g_0}\delta(\tau)}{1+\sqrt{g_0}\delta(\tau)}\right)-\mathrm{e}^{-\lambda}\right]
		\nonumber\\
		&\geq \exp\left[-\frac{n\tau^2}{2}\mathrm{tr}\left(\bSigma{}_0^{-1}\right)\right]\prod_{i=1}^n\phi_{\bSigma_0}(\by_i)\left(\frac{\Omega(\mathrm{e}^\lambda - 1)}{C(\lambda)(1+\delta(\tau)\sqrt{g_0})}\right)\nonumber
		\end{align}
		for some constant $C(\lambda)$ that depends on $\lambda$ only, 
		where the last inequality is due to the mean-value theorem. 
	\end{itemize}
	The proof is thus completed. 
\end{proof}

\begin{proof}[Proof of \textbf{Lemma \ref{lemma:conditional_density_upper_bound}}]
	Suppose $h$ is of the form \eqref{eqn:repulsive_function_1}. Then by Lemma \ref{lemma:marginal_likelihood_upper_bound} and Lemma \ref{lemma:marginal_likelihood_lower_bound} we can write 
	\begin{align}
	&\frac{p(\by_{1:n}|z_{1:n},K)}{p(\by_{1:n})}\prod_{i=1}^n\phi\left(\by_i\middle|\mb_i,\bSigma_0\right)\nonumber\\
	&\qquad
	\leq C(\lambda)
	\exp\left[\frac{n\tau^2}{2}\mathrm{tr}\left(\bSigma{}_0^{-1}\right)\right]\frac{\left(1+\delta(\tau)g_0^{\frac{2}{3}}\right)^{\frac{3}{2}}}{\Omega Z_K(\mathrm{e}^\lambda - 1)}\nonumber\\
	&\qquad\quad\times {K\choose 2}^{-1}\sum_{k<k'} g\left(\left[\sum_{j=1}^p
	\left(
	\frac{\one_{n_k}\transpose{}\widehat\by_{(k)j}}{n_k+\kappa_j^2}-\frac{\one_{n_{k'}}\transpose{}\widehat\by_{(k')j}}{n_{k'}+\kappa_j^2}
	\right)^2
	+\sum_{j=1}^p\sigma_j^2\left(\frac{1}{n_k+\kappa_j^2}+\frac{1}{n_{k'}+\kappa_j^2}\right)\right]^{\frac{1}{2}}\right)\nonumber
	\\
	&\qquad\quad\times \prod_{j=1}^p\prod_{k=1}^K\frac{\phi\left(\widehat\by_{(k)j}\middle|\zero_{n_k},\sigma^2_j\eye_{n_k}+\tau^2\one_{n_k}\one_{n_k}\transpose{}\right)\phi(\widehat\by_{(k)j}|\widehat\mb_{(k)j},\sigma^2_j\eye_{n_k})}{\phi(\widehat\by_{(k)j}|\zero_{n_k},\sigma^2_j\eye_{n_k})}\nonumber.
	\end{align}
	Simple algebra shows that
	\begin{align}
	&\prod_{j=1}^p\prod_{k=1}^K\frac{\phi\left(\widehat\by_{(k)j}\middle|\zero_{n_k},\sigma^2_j\eye_{n_k}+\tau^2\one_{n_k}\one_{n_k}\transpose{}\right)\phi(\widehat\by_{(k)j}|\widehat\mb_{(k)j},\sigma^2_j\eye_{n_k})}{\phi(\widehat\by_{(k)j}|\zero_{n_k},\sigma^2_j\eye_{n_k})}\nonumber\\
	&\qquad =\prod_{j=1}^p\prod_{k=1}^K\phi\left(\widehat\by_{(k)j}\middle|\left(\eye_{n_k}+\frac{\tau^2}{\sigma_j^2}\one_{n_k}\one_{n_k}\transpose{}\right)
	\widehat\mb_{(k)j},\sigma_j^2\eye_{n_k}+\tau^2\one_{n_k}\one_{n_k}\transpose{}\right)\exp\left[-\frac{(\one_{n_k}\transpose{}\widehat\mb_{(k)j})^2}{2\sigma^2_j(n_k+\kappa_j^2)}\right]\nonumber\\
	&\qquad\leq \prod_{j=1}^p\prod_{k=1}^K\phi\left(\widehat\by_{(k)j}\middle|\left(\eye_{n_k}+\frac{\tau^2}{\sigma_j^2}\one_{n_k}\one_{n_k}\transpose{}\right)
	\widehat\mb_{(k)j},\sigma_j^2\eye_{n_k}+\tau^2\one_{n_k}\one_{n_k}\transpose{}\right)\nonumber.
	\end{align}
	It follows by Jensen's inequality that
	\begin{align}
	&\int_{\mathbb{R}^{n}}\cdots\int_{\mathbb{R}^n}\frac{p(\by_{1:n}|z_{1:n},K)}{p(\by_{1:n})}\prod_{i=1}^n\phi(\by_i|\mb_i,\bSigma_0)\mathrm{d}\by_1\cdots\mathrm{d}\by_n\nonumber\\
	&\qquad\leq C(\lambda)
	\exp\left[\frac{n\tau^2}{2}\mathrm{tr}\left(\bSigma_0^{-1}\right)\right]\frac{\left(1+\delta(\tau)g_0^{\frac{2}{3}}\right)^{\frac{3}{2}}}{\Omega Z_K(\mathrm{e}^\lambda - 1)}
	\nonumber\\
	&\qquad\quad \times {K\choose 2}^{-1}\sum_{k<k'}\int_{\mathbb{R}}
	g\left(\left[\sum_{j=1}^p
	\left(
	\frac{\one_{n_k}\transpose{}\widehat\by_{(k)j}}{n_k+\kappa_j^2}-\frac{\one_{n_{k'}}\transpose{}\widehat\by_{(k')j}}{n_{k'}+\kappa_j^2}
	\right)^2
	+\sum_{j=1}^p\sigma_j^2\left(\frac{1}{n_k+\kappa_j^2}+\frac{1}{n_{k'}+\kappa_j^2}\right)\right]^{\frac{1}{2}}\right)\nonumber
	\\
	&\qquad\quad\times 
	\prod_{j=1}^p\phi\left(x_{kk'j}\middle|\frac{1}{\kappa_j^2}\left(\one_{n_k}\transpose{}\widehat\mb_{(k)j}-\one_{n_{k'}}\transpose{}\widehat\mb_{(k')}\right),\sigma^2_j\left(\frac{n_k}{\kappa_j^2(n_k+\kappa_j^2)}+\frac{n_{k'}}{\kappa_j^2(n_{k'}+\kappa_j^2)}\right)\right)\mathrm{d}x_{kk'j}\nonumber\\
	&\qquad \leq 
	C(\lambda)
	\exp\left[\frac{n\tau^2}{2}\mathrm{tr}\left(\bSigma{}_0^{-1}\right)\right]\frac{\left(1+\delta(\tau)g_0^{\frac{2}{3}}\right)^{\frac{3}{2}}}{\Omega Z_K(\mathrm{e}^\lambda - 1)}\nonumber\\
	&\qquad\quad\times {K\choose 2}^{-1}\sum_{k<k'}
	g\left(\left[\sum_{j=1}^p\frac{1}{\kappa_j^4}\left(\one_{n_k}\transpose{}\widehat{\mb}_{(k)j}-\one_{n_{k'}}\transpose{}\widehat{\mb}_{(k')j}\right)^2+2p\tau^2\right]^{\frac{1}{2}}\right)
	\end{align}
	where
	\[
	x_{kk'j}=\frac{\one_{n_k}\transpose{}\widehat\by_{(k)j}}{n_k+\kappa_j^2}-\frac{\one_{n_{k'}}\transpose{}\widehat\by_{(k')j}}{n_{k'}+\kappa_j^2},\qquad\text{and}\qquad \sigma_{kk'j}^2=\sigma^2_j\left(\frac{1}{n_k+\kappa_j^2}+\frac{1}{n_{k'}+\kappa_j^2}\right).
	\]
	The case where $h$ is of the form \eqref{eqn:repulsive_function_2} can be proved in the exactly same fashion. The proof is thus completed.
\end{proof}

\begin{proof}[Proof of \textbf{Lemma \ref{lemma:expected_value_equation}}]
	First for each fixed $j$ we write
	\begin{align}
	&\int_{\mathbb{R}^{p\times n}}\int_{\mathbb{R}^{p\times n}}\left(\one_{n_k}\transpose{}\widehat\mb_{(k)j}-\one_{n_{k'}}\transpose{}\widehat\mb_{(k')j}\right)^2F_0^{(n_k+n_{k'})}(\mathrm{d}\mb_{(k)}\mathrm{d}\mb_{(k')})\nonumber\\
	&\qquad = \iint\left[\left(\sum_{i:z_i=k}\widehat m_{ij}\right)^2+\left(\sum_{i:z_i=k'}\widehat m_{ij}\right)^2-2\left(\sum_{i:z_i=k}\widehat m_{ij}\right)\left(\sum_{i:z_i=k'}\widehat m_{ij}\right)\right]F_0^{(n_k+n_{k'})}(\mathrm{d}\mb_{(k)}\mathrm{d}\mb_{(k')})\nonumber\\
	&\qquad = \left(n_k\mathbb{E}_0\widehat m_j\right)^2+n_k\mathrm{Var}_0(\widehat m_j)+\left(n_{k'}\mathbb{E}_0\widehat m_j\right)^2+n_{k'}\mathrm{Var}_0(\widehat m_j)-2n_kn_{k'}\left(\mathbb{E}_0\widehat m_j\right)^2\nonumber\\
	&\qquad = \mathbb{E}_0\left[(\widehat m_j)^2\right](n_k+n_k').\nonumber
	\end{align}
	Writting $\bU=\left(\bu_1,\cdots,\bu_p\right)$ 
	where $\bu_j\in\mathbb{R}^p$, it follows that
	\begin{align}
	&\sum_{j=1}^p\frac{1}{\kappa_j^4}\mathbb{E}_\bz\left[\iint\left(\one_{n_k}\transpose{}\widehat\mb_{(k)j}-\one_{n_{k'}}\transpose{}\widehat\mb_{(k')j}\right)^2F_0^{(n_k+n_{k'})}(\mathrm{d}\mb_{(k)}\mathrm{d}\mb_{(k')})\right]\nonumber\\
	&\qquad =\mathbb{E}_{\bw_{1:K}}\left[\mathbb{E}_{\bz_{1:n}}\left( n_k+n_{k'} \mid \bw_{1:K}\right)\right]
	\sum_{j=1}^p\frac{\tau^4}{\sigma_j^4}\mathbb{E}_0\left(\left[\bu_j\transpose{}\mb\right]^2\right)\nonumber\\
	&\qquad =\frac{2n}{K}\tau^4\mathbb{E}_0\left(\sum_{j=1}^p\frac{1}{\sigma_j^4}\mb\transpose{}\bu_j\bu_j\transpose{}\mb\nonumber\right)\nonumber\\
	&\qquad = \tau^4\frac{2n}{K}\mathbb{E}_0\left[\mb\transpose{}\left(\sum_{j=1}^p\frac{1}{\sigma_j^4}\bu_j\bu_j\transpose{}\right)\mb\right]\nonumber\\
	&\qquad = \tau^4\mathbb{E}_0\left(\mb\transpose{}\bSigma{}_0^{-2}\mb\right)\frac{2n}{K}.\nonumber
	\end{align}
\end{proof}



\section{Proofs of Auxiliary Results for Lemma \ref{lemma:KL_ball_lower_bound}} 
\label{sec:proofs_of_auxiliary_results_in_appendix_e}
\subsection*{Proof of Lemma \ref{lemma:discretization}}
\begin{proof}
	The proofs are similar to those in {Lemma 3.1}, {Lemma 3.2}, and {Lemma 3.4} in \cite{ghosal2001entropies}. Let $M=\max\left\{2a,\sqrt{8}\usigma\left(\log\frac{1}{\epsilon}\right)^{\frac{1}{2}}\right\}$, and let $\epsilon$ be sufficiently small such that $M>2a$. Then
	\begin{eqnarray}
		\sup_{\|\by\|\geq M}|f_F(\by)-f_{F^\star}(\by)|\leq 2\phi_{\bSigma}(M-a)\leq 2\phi_{\bSigma}(M/2)\lesssim \exp(-M^2/(8\usigma^2))= \epsilon,\nonumber
	\end{eqnarray}
	so that it suffices to consider $\|\by\|\leq M$. Denote $Q_{\bSigma}(\by)=\by\transpose\bSigma^{-1}\by$. By Taylor's expansion we have
	\begin{eqnarray}
		\left|\phi_{\bSigma}(\by-\bmu)-\sum_{j=1}^{J-1}\frac{(-1)^j}{2^j(2\pi)^{\frac{p}{2}}}\det({\bSigma})^{-\frac{1}{2}}Q^j_{\bSigma}(\by-\bmu)\right|\lesssim\left(\frac{\mathrm{e}/2Q_{\bSigma}(\by-\bmu)}{J}\right)^J.\nonumber
	\end{eqnarray}
	Hence for any probability distribution $F^\star$ on $\{\bmu:\|\bmu\|\leq a\}\times\calS$, a standard argument of triangle inequality yields
	\begin{eqnarray}\label{ineq:sup_norm_lessthan_M}
		\sup_{\|\by\|\leq M}|f_F(\by)-f_{F^\star}(\by)|
		&\leq&\sup_{\|\by\|\leq M}\left|\sum_{j=1}^{J-1}\frac{(-1)^j}{2^j(2\pi)^{\frac{p}{2}}}\int\det(\bSigma)^{-\frac{1}{2}}Q^j_{\bSigma}(\by-\bmu)\left(\mathrm{d}F-\mathrm{d}F^\star\right)\right|\nonumber\\
		&&+2\sup_{\|\by\|\leq M,\|\bmu\|\leq a}\left|
		\phi_{\bSigma}(\by-\bmu)-\sum_{j=1}^{J-1}\frac{(-1)^j}{2^j(2\pi)^{\frac{p}{2}}}\det(\bSigma)^{-\frac{1}{2}}Q^j_{\bSigma}(\by-\bmu)
		\right|\nonumber\\
		&\leq&
		\sup_{\|\by\|\leq M}\left|\sum_{j=1}^{J-1}\frac{(-1)^j}{2^j(2\pi)^{\frac{p}{2}}}\int \det(\bSigma)^{-\frac{1}{2}}Q^j_{\bSigma}(\by-\bmu)\left(\mathrm{d}F-\mathrm{d}F^\star\right)\right|\nonumber\\
		&&+2c_1\sup_{\|\by\|\leq M,\|\bmu\|\leq a}\left(\frac{\mathrm{e}/2Q_{\bSigma}(\by-\bmu)}{J}\right)^J,
	\end{eqnarray}
	for some constant $c_1>0$. Suppose $\bU=\eye_p$. Expanding $Q^j_{\bSigma}(\by-\bmu)$ by multinomial theorem:
	\begin{eqnarray}
	Q_{\bSigma}^j(\by-\bmu)&=&
	\sum_{\substack{r+s+t=j\\ r_1+\cdots+r_p=r\\ t_1+\cdots+t_p=t\\ s_1+\cdots+s_p=s}}
	\left({j\choose r_1\cdots r_p,s_1\cdots s_p, t_1\cdots t_p}\prod_{i=1}^py_i^{2r_i}\right)\left(\prod_{i=1}^p\frac{\mu_i^{s_i+2t_i}}{\lambda_i^{r_i+s_i+t_i}(\bSigma)}\right).\nonumber
	\end{eqnarray}
	In order that the first term on the RHS of \eqref{ineq:sup_norm_lessthan_M} vanishes, it is sufficient that 
	$$
	\int \det(\bSigma)^{-\frac{1}{2}}Q_{\bSigma}^j(\by-\bmu)(\mathrm{d}F-\mathrm{d}F^\star)=0
	$$
	for all $j=0,1,\cdots,J-1$. By the multinomial expansion, a sufficient condition for the last display is that
	$$
	\int \det(\bSigma)^{-\frac{1}{2}}\prod_{i=1}^p\frac{\mu_i^{s_i+2t_i}}{\lambda_i^{r_i+s_i+t_i}(\bSigma)}\left(\mathrm{d}F'-\mathrm{d}F^\star\right)=0
	$$
	for all possible $r_i,s_i,t_i,i=1,\cdots,p$. 
	According to {Lemma A.1} in \cite{ghosal2001entropies}, ${F}^\star$ can be select to be a discrete distribution with at most $N\lesssim J^p(2J-1)^p+1\lesssim J^{2p}$ support points. For the case $\bU$ is not the identity matrix, the above argument can be applied with $y_i$ and $\mu_i$ replaced by $(\bU\transpose y)_i$ and $(\bU\transpose \bmu)_i$, respectively. 

	Now we focus on the selection of $J$. Notice that
	\begin{eqnarray}
		\sup_{\|\by\|\leq M,\|\bmu\|\leq a}Q_{\bSigma}(\by-\bmu)\lesssim
		\sup_{\|\by\|\leq M,\|\bmu\|\leq a}\|\by-\bmu\|^2\lesssim M^2\lesssim\left(\log\frac{1}{\epsilon}\right).\nonumber
	\end{eqnarray}
	Hence the second term on the RHS of \eqref{ineq:sup_norm_lessthan_M} is upper bounded by a constant multiple of $\left(\left(c_2\log\frac{1}{\epsilon}\right)/J\right)^J$
	for some constant $c_2>0$. Set $J=\lceil(1+c_2)\left(\log\frac{1}{\epsilon}\right)\rceil$. Then
	\begin{equation}
		\sup_{\|\by\|\leq M}|f_F(\by)-f_{F^\star}(\by)|\lesssim
		\left(\frac{\left(c_2\log\frac{1}{\epsilon}\right)}{J}\right)^J\lesssim \left(\frac{c_2}{1+c_2}\right)^{(1+c_2)\log(1/\epsilon)}=\epsilon^{(1+c)\log(1+1/c)}\leq\epsilon\nonumber
	\end{equation}
	for sufficiently small $\epsilon>0$, where the last inequality is due to the fact $(1+c)\log(1+1/c)$ decrease with $c$ and converges to $1$ as $c\to\infty$.  Hence the number $N$ of support points for discrete $F^\star$ such that $\|f_F-f_{F^\star}\|_\infty\lesssim\epsilon$ is of order $J^{2p}\propto\left(\log\frac{1}{\epsilon}\right)^{2p}$. 

	For the inequality regarding $L_1$ distance, notice that for $\|\by\|>T\geq 2a$, $f_F(\by)\lesssim \exp\left(-\|\by\|^2/8\usigma^2\right)$, so that
	\begin{eqnarray}\label{ineq:L1_versus_Linfinity}
		\|f_F-f_{F^\star}\|_1&\lesssim&\int_{\|\by\|>T}\exp\left(-\frac{\|\by\|^2}{8\usigma^2}\right)\mathrm{d}\by+\int_{\|\by\|<T}\|f_F-f_{F^\star}\|_\infty\mathrm{d}\by\nonumber\\
		&\lesssim&\exp\left(-\frac{T^2}{8\usigma^2}\right)+T^p\|f_F-f_{F^\star}\|_\infty.
	\end{eqnarray} 
	Now take
	$$
	T=\max\left\{2a,\usigma\sqrt{8\log\left(\frac{1}{\|f_F-f_{F^\star}\|_{\infty}}\right)}\right\}.
	$$
	It follows that the first term on the RHS of \eqref{ineq:L1_versus_Linfinity} is bounded by $\|f_F-f_{F^\star}\|_\infty\lesssim\epsilon$, while the second term is bounded by a multiple of
	$$
	\|f_F-f_{F^\star}\|_\infty\max\left\{a^p,\log\left(\frac{1}{\|f_F-f_{F^\star}\|_\infty}\right)^{\frac{p}{2}}\right\}\lesssim\epsilon\left(\log\frac{1}{\epsilon}\right)^{\frac{p}{2}}.
	$$
	Therefore, for sufficiently small $\epsilon>0$, $\|f_F-f_{F^\star}\|_1\lesssim\epsilon\left(\log\frac{1}{\epsilon}\right)^{\frac{p}{2}}$.
\end{proof}
\subsection*{Proof of Lemma \ref{lemma:discretization_grid}}
\begin{proof}
	First for a given $\epsilon$, obtain $F'$ by \textbf{Lemma \ref{lemma:discretization}} with at most $n\lesssim\left(\log\frac{1}{\epsilon}\right)^{\frac{1}{2}}$ support points. Write $F'=\sum_kw_k\delta_{(\bmu_k,\bSigma_k)}$. For each $k$, find $\bmu_k^\star\in\{\bmu:\bmu/(2\epsilon)\in\mathbb{Z}^p\},\bSigma^\star_k\in\{\bSigma:\lambda_j(\bSigma)/(2\epsilon)\in\mathbb{N}_+,j=1,\cdots,p\}$ such that $\|\bmu_k-\bmu_k^\star\|\lesssim\epsilon$ and $\|\bSigma_k-\bSigma_k^\star\|\lesssim\epsilon$. Furthermore the function class $\{(\bmu,\bSigma)\mapsto\phi(\by\mid\bmu,\bSigma)\}_{\by\in\mathbb{R}^p}$ indexed by $\by\in\mathbb{R}^p$ is uniformly Lipschitz continuous, since $\nabla_{\bmu}\phi_{\bSigma}(\by-\bmu)$ is uniformly bounded and $\bSigma\in\calS$ is compact. Therefore, by taking $F^\star=\sum_kw_k\delta_{(\bmu_k^\star,\bSigma_k^\star)}$, we have by the triangle inequality
	\begin{align}
		\|f_F-f_{F^\star}\|_\infty\leq&\|f_F-f_{F'}\|_\infty+\sum_{k=1}^Kw_k\|\phi_{\bSigma_k}(\by-\bmu_k)-\phi_{\bSigma^\star_k}(\by-\bmu_k^\star)\|_\infty\nonumber\\
		\lesssim &\epsilon+\sum_{k=1}^Kw_kL\left(\|\bmu_k-\bmu_k^\star\|+\|\bSigma_k-\bSigma_k^\star\|\right)		\lesssim\epsilon\nonumber
	\end{align}
	where $L$ is the (uniform) Lipschitz constant for the function class $\{(\bmu,\bSigma)\mapsto\phi_{\bSigma}(\by-\bmu)\}_{\by\in\mathbb{R}^p}$. Now applying the exactly same argument used in deriving \eqref{ineq:L1_versus_Linfinity} yields $\|f_F-f_{F^\star}\|_1\lesssim \epsilon\left(\log\frac{1}{\epsilon}\right)^{\frac{p}{2}}$. 
\end{proof}
\subsection*{Proof of Lemma \ref{lemma:KL_ball_versus_Hellinger_ball}}
\begin{proof}
	Since $f_0(\by)\leq \lsigma^{p}\phi_{\eye_p}(\mathbf{0})$, and
	$$
	f_F(\by)\geq \frac{1}{\usigma^p}\int_{\{\|\bmu\|\leq B\}}\phi_{\eye_p}\left(\frac{\by-\bmu}{\lsigma}\right)\mathrm{d}F\geq\frac{1}{2\usigma^p}\phi_{\eye_p}\left(\frac{\by(\|\by\|+B)}{\|\by\|\lsigma}\right),
	$$
	then we see that $f_0/f_F\lesssim \exp(b_1\|\by\|^2)$ for some constant $b_1>0$. Hence for sufficientl small $\delta>0$, 
	\begin{eqnarray}
		\int \left(\frac{f_0(\by)}{f_F(\by)}\right)^\delta f_0(\by)\mathrm{d}\by\lesssim 
		\int \int\exp(\delta b_1\|\by\|^2)\exp\left(-\frac{1}{2\lsigma^2}\|\by-\bmu\|^2\right)\mathrm{d}F_0\mathrm{d}\by<\infty.
	\end{eqnarray}
	The proof is completed by applying {Theorem 5} in \cite{wong1995probability}. 
\end{proof}
\subsection*{Proof of Lemma \ref{lemma:discrete_approximation}}
\begin{proof}
	Let ${E}_0=\left(\bigcup_k E_k\right)^c$. We estimate
	\begin{eqnarray}
		|f_F(\by)-f_{F^\star}(\by)|
		&\leq &\int_{E_0}\phi_{\bSigma}(\by-\bmu)\mathrm{d}F+\sum_{k=1}^N\int_{E_k}|\phi_{\bSigma}(\by-\bmu)-\phi_{\bSigma_k^\star}(\by-\bmu_k^\star)|\mathrm{d}F
		\nonumber\\
		&&+\sum_{k=1}^N\phi_{\bSigma^\star_k}(\by-\bmu_k^\star)|{P}_F(E_k)-w_k^\star|.
	\end{eqnarray}
	For $(\bmu,\bSigma)\in E_k$, we see that $\|\bmu-\bmu_k^\star\|\lesssim\epsilon$ and $|\lambda_j(\bSigma)-\lambda_j(\bSigma_k^\star)|\lesssim\epsilon$. Since eigenvalues of covariance matrices are bounded away from $0$ and $\infty$, we see that $
	\left|\sqrt{\lambda_j(\bSigma)}-\sqrt{\lambda_j(\bSigma_k^\star)}\right|\lesssim{\epsilon}/{\left|\sqrt{\lambda_j(\bSigma)}+\sqrt{\lambda_j(\bSigma_k^\star)}\right|}\lesssim\epsilon.
	$
	Hence by \eqref{ineq:Hellinger_distance_upper_bound} and the relation between Hellinger distance and $\|\cdot\|_1$, we have $\|\phi_{\bSigma}(\by-\bmu)-\phi_{\bSigma_k^\star}(\by-\bmu_k^\star)\|_1\lesssim \epsilon$ whenever $(\bmu,\bSigma)\in E_k$ for all $k$ and all sufficiently small $\epsilon$. Thus we obtain from Fubini's theorem that
	\begin{eqnarray}
		\|f_F-f_{F^\star}\|_1
		&\leq&
		\int_{E_0}\int \phi_{\bSigma}(\by-\bmu)\mathrm{d}\by\mathrm{d}F
		+\sum_{k=1}^N\int_{E_k}\|\phi_{\bSigma}(\by-\bmu)-\phi_{\bSigma_k^\star}(\by-\bmu_k^\star)\|_1\mathrm{d}F\nonumber\\
		&&+\sum_{k=1}^N|F(E_k)-w_k^\star|\int\phi_{\bSigma_k^\star}(\by-\bmu_k^\star)\mathrm{d}\by
		\nonumber\\
		&\lesssim & \left[\sum_{k=1}^Nw_k^\star-\sum_{k=1}^N F(E_k)\right]+\epsilon+\sum_{k=1}^N|F(E_k)-w_k^\star|
		\lesssim \epsilon+\sum_{k=1}^N|F(E_k)-w_k^\star|.\nonumber
	\end{eqnarray}
\end{proof}

\clearpage

\section{Derivation of the Generalized Urn Model} 
\label{sec:proof_of_the_generalized_urn_model}
As shown in \cite{miller2016mixture}, the marginal distribution of $\mathcal{C}_n$ with $K$ and $z=(z_1,\cdots,z_n)$ marginalized out is given by
\begin{eqnarray}
p(\mathcal{C}_n)=V_n(|\mathcal{C}_n|)\prod_{c\in\mathcal{C}_n}\frac{\Gamma(\beta+|c|)}{\Gamma(\beta)}
\end{eqnarray}
where
\begin{eqnarray}
V_n(t):=\sum_{K=t}^{\infty}\frac{\Gamma(K+1)\Gamma(\beta K+1)}{\Gamma(K-t+1)\Gamma(\beta K+n+1)}p(K). 
\end{eqnarray}
The following generalized Bayes rule is useful: If $p(\by\mid\btheta)=\phi(\by\mid\btheta)$ and $\btheta\sim\Pi$, then
\begin{eqnarray}\label{eqn:generalized_bayes_rule}
\Pi(\btheta\in A\mid\by)=\left.\int_A \phi(\by\mid\btheta)\Pi(\mathrm{d}\btheta)\right/\int \phi(\by\mid\btheta)\Pi(\mathrm{d}\btheta)\propto \int_A \phi(\by\mid\btheta)\Pi(\mathrm{d}\btheta).
\end{eqnarray}
\begin{proof}[Proof of \textbf{Theorem \ref{thm:urn_model}}]
	The restaurant process for the exchangeable partition model proposed by \cite{miller2016mixture} is given by
	\begin{eqnarray}
		\Pi(\mathcal{C}_n=\mathcal{C}_{n-1}\cup\{\{n\}\}\mid \mathcal{C}_{n-1})&\propto&\frac{V_n(\ell+1)}{V_n(\ell)}\beta\nonumber\\
		\Pi(\mathcal{C}_n=(\mathcal{C}_{n-1}\backslash \{c\})\cup\{c\cup\{n\}\}\mid \mathcal{C}_{n-1})&\propto&|c|+\beta\nonumber
	\end{eqnarray}
	where $|\mathcal{C}_{n-1}|=\ell$. Then for any measurable $A$, the following derivation using chain rule of conditional distributions is available
	\begin{align}
		&\Pi(\btheta_n\in A\mid\btheta_1,\cdots,\btheta_{n-1})\nonumber\\
		&\qquad=\sum_{\mathcal{C}_{n}}\Pi(\btheta_n\in A\mid\btheta_1,\cdots,\btheta_{n-1},\mathcal{C}_n)p(\mathcal{C}_{n}\mid\btheta_1,\cdots,\btheta_{n-1})\nonumber\\
		&\qquad=\sum_{\mathcal{C}_{n}}\Pi(\btheta_n\in A\mid\btheta_1,\cdots,\btheta_{n-1},\mathcal{C}_n)p(\mathcal{C}_{n}\mid\mathcal{C}_{n-1})\nonumber\\
		&\qquad\quad\propto \left[\frac{V_n(\ell+1)\beta}{V_n(\ell)}\right]\Pi(\btheta_n\in A\mid\btheta_1,\cdots,\btheta_{n-1},\mathcal{C}_n=\mathcal{C}_{n-1}\cup\{\{n\}\})\nonumber\\
		&\qquad\quad+\sum_{c\in\mathcal{C}_{n-1}}(|c|+\beta)\Pi(\btheta_n\in A\mid\btheta_1,\cdots,\btheta_{n-1},\mathcal{C}_n=(\mathcal{C}_{n-1}\backslash \{c\})\cup\{c\cup\{n\}\})\nonumber
	\end{align}
	Since $\Pi(\btheta_n\in A\mid\btheta_1,\cdots,\btheta_{n-1},\mathcal{C}_n=(\mathcal{C}_{n-1}\backslash \{c\})\cup\{c\cup\{n\}\})=\delta_{(\bgamma_c^\star,\bGamma_c^\star)}(A)$, we focus on deriving $\Pi(\btheta_n\in A\mid\btheta_1,\cdots,\btheta_{n-1},\mathcal{C}_n=\mathcal{C}_{n-1}\cup\{\{n\}\}).$	Since
	\begin{align}
		&\Pi(\btheta_n\in A\mid \btheta_1,\cdots,\btheta_{n-1},\mathcal{C}_n=\mathcal{C}_{n-1}\cup\{\{n\}\})\nonumber\\
		&\qquad=\sum_{K=\ell+1}^{\infty}\Pi(\btheta_n\in A\mid\btheta_1,\cdots,\btheta_{n-1},K)p(K\mid\mathcal{C}_n=\mathcal{C}_{n-1}\cup\{\{n\}\})\nonumber\\
		&\qquad=\sum_{K=\ell+1}^{\infty}\Pi(\btheta_n\in A\mid\bgamma_c^\star,\bGamma_c^\star,c\in\mathcal{C}_{n-1},K)p(K\mid \mathcal{C}_n=\mathcal{C}_{n-1}\cup\{\{n\}\})\nonumber.
	\end{align}
	Hence 
	\begin{align}
		&\Pi(\btheta_n\in A\mid\btheta_1,\cdots,\btheta_{n-1})\nonumber\\
		&\qquad\propto \left[\frac{V_n(t+1)\beta}{V_n(t)}\right]\sum_{K=\ell+1}^{\infty}
		p(K\mid \mathcal{C}_{n}=\mathcal{C}_{n-1}\cup\{\{n\}\})
		\Pi(\btheta_n\in A\mid\bgamma_c^\star,\bGamma_c^\star,c\in\mathcal{C}_{n-1},K)
		\nonumber\\
		&\qquad\quad+\sum_{c\in\mathcal{C}_{n-1}}(|c|+\beta)\delta_{(\bgamma_c^\star,\bGamma_c^\star)}(A),
	\end{align}
	and hence, by generalized Bayes rule \eqref{eqn:generalized_bayes_rule},
	\begin{eqnarray}
		\Pi(\btheta_n\in A\mid\by_n,\btheta_1,\cdots,\btheta_{n-1})
		&\propto &\left[\frac{V_n(t+1)\beta}{V_n(t)}\right]\sum_{K=\ell+1}^{\infty}p(K\mid \mathcal{C}_n=\mathcal{C}_{n-1}\cup\{\{n\}\})\times\nonumber\\
		&&\iint_A\phi(\by_n\mid\bgamma_n,\bGamma_n)\Pi(\mathrm{d}\bgamma_n\mathrm{d}\bGamma_n\mid\bgamma_c^\star,\bGamma_c^\star,c\in\mathcal{C}_{n-1},K)\nonumber\\
		&&+\sum_{c\in\mathcal{C}_{n-1}}(|c|+\beta)\delta_{(\bgamma_c^\star,\bGamma_c^\star)}(A)\phi(\by_n\mid\bgamma_c^\star,\bGamma_c^\star)\nonumber.
	\end{eqnarray}
	By definition, for any measurable $A\subset\mathbb{R}^p\times\calS$, when $K\geq \ell+1$, we have
	\begin{align}
	&\Pi(\btheta_n\in A\mid\bgamma_c^\star,\bGamma_c^\star,c\in\mathcal{C}_{n-1},K)\nonumber\\
	&\qquad\propto\iint_A
		\left[\int\cdots\int h_K(\bgamma_{c}^\star:c\in\mathcal{C}_{n-1}\cup\calC_\varnothing)\prod_{c\in\calC_\varnothing\backslash\{\underline{c}\}}p_\bmu(\bgamma_c^\star)\mathrm{d}\bgamma_{c}^\star\right]p_\bSigma(\bGamma_{\underline{c}}^\star)p_\bmu(\bgamma_{\underline{c}}^\star)\mathrm{d}\bgamma_{\underline{c}}^\star\mathrm{d}\bGamma_{\underline{c}}^\star.\nonumber\\
		&\qquad=\iint_A L_K(\bgamma_{\underline{c}}^\star)p_\bmu(\bgamma_{\underline{c}}^\star)p_{\bSigma}(\bGamma_{\underline{c}}^\star)\mathrm{d}\bgamma_{\underline{c}}^\star\mathrm{d}\bGamma_{\underline{c}}^\star\nonumber.
	\end{align}
	Normalizing the above conditional probability distribution yields
	\begin{eqnarray}
	\Pi(\btheta_n\in A\mid\bgamma_c^\star,\bGamma_c^\star,c\in\mathcal{C}_{n-1},K)=
	\frac{\displaystyle
	\iint_A L_K(\bgamma_{\underline{c}}^\star)p_\bmu(\bgamma_{\underline{c}}^\star)p_{\bSigma}(\bGamma_{\underline{c}}^\star)\mathrm{d}\bgamma_{\underline{c}}^\star\mathrm{d}\bGamma_{\underline{c}}^\star
	}{\displaystyle
	\iint L_K(\bgamma_{\underline{c}}^\star)p_\bmu(\bgamma_{\underline{c}}^\star)p_{\bSigma}(\bGamma_{\underline{c}}^\star)\mathrm{d}\bgamma_{\underline{c}}^\star\mathrm{d}\bGamma_{\underline{c}}^\star
	}.
	\end{eqnarray}
	Hence the generalized Bayes rule \eqref{eqn:generalized_bayes_rule} yields
	\begin{eqnarray}
	\Pi(\btheta_n\in A\mid\by_n,\bgamma_c^\star,\bGamma_c^\star,c\in\mathcal{C}_{n-1},K)=
	\frac{\displaystyle
	\iint_A\phi(\by_n\mid\bgamma_{\underline{c}}^\star,\bGamma_{\underline{c}}^\star) L_K(\bgamma_{\underline{c}}^\star)p_\bmu(\bgamma_{\underline{c}}^\star)p_{\bSigma}(\bGamma_{\underline{c}}^\star)\mathrm{d}\bgamma_{\underline{c}}^\star\mathrm{d}\bGamma_{\underline{c}}^\star
	}{\displaystyle
	\iint\phi(\by_n\mid\bgamma_{\underline{c}}^\star,\bGamma_{\underline{c}}^\star) L_K(\bgamma_{\underline{c}}^\star)p_\bmu(\bgamma_{\underline{c}}^\star)p_{\bSigma}(\bGamma_{\underline{c}}^\star)\mathrm{d}\bgamma_{\underline{c}}^\star\mathrm{d}\bGamma_{\underline{c}}^\star
	}.\nonumber
	\end{eqnarray}
	Notice that, again, by the generalized Bayes rule \eqref{eqn:generalized_bayes_rule}, we have
	\begin{align}
		&p(K\mid\mathcal{C}_n=\mathcal{C}_{n-1}\cup\{\{n\}\})\iint_A\phi(\by_n\mid\bgamma_n,\bGamma_n)\Pi(\mathrm{d}\bgamma_n\mathrm{d}\bGamma_n\mid\bgamma_c^\star,\bGamma_c^\star,c\in\mathcal{C}_{n-1},K)\nonumber\\
		&\qquad=p(K\mid\mathcal{C}_n=\mathcal{C}_{n-1}\cup\{\{n\}\})\iint\phi(\by_n\mid\bgamma_n,\bGamma_n)\Pi(\mathrm{d}\bgamma_n\mathrm{d}\bGamma_n\mid\bgamma_c^\star,\bGamma_c^\star,c\in\mathcal{C}_{n-1},K)\nonumber\\
		&\qquad\quad\times \Pi(\btheta_n\in A\mid\by_n,\bgamma_c^\star,\bGamma_c^\star,c\in\mathcal{C}_{n-1},K)\nonumber\\
		&\qquad=p(K\mid\mathcal{C}_n=\mathcal{C}_{n-1}\cup\{\{n\}\})
		\left[\frac{\displaystyle\iint\phi(\by_n\mid\bgamma_{\underline{c}}^\star,\bGamma_{\underline{c}}^\star)
		L_K(\bgamma_{\underline{c}}^\star)p_\bmu(\bgamma_{\underline{c}}^\star)p_{\bSigma}(\bGamma_{\underline{c}}^\star)\mathrm{d}\bgamma_{\underline{c}}^\star\mathrm{d}\bGamma_{\underline{c}}^\star}
		{\displaystyle
		\iint
		L_K(\bgamma_{\underline{c}}^\star)p_\bmu(\bgamma_{\underline{c}}^\star)p_{\bSigma}(\bGamma_{\underline{c}}^\star)\mathrm{d}\bgamma_{\underline{c}}^\star\mathrm{d}\bGamma_{\underline{c}}^\star
		}\right]
		\nonumber\\
		&\qquad\quad\times \Pi(\btheta_n\in A\mid\by_n,\bgamma_c^\star,\bGamma_c^\star,c\in\mathcal{C}_{n-1},K)\nonumber\\
		&\qquad=p(K\mid\mathcal{C}_n=\mathcal{C}_{n-1}\cup\{\{n\}\})\nonumber\\
		&\qquad\quad\times\left[\frac
		{\displaystyle\int\cdots\iint\phi(\by_n\mid\bgamma_{\underline{c}}^\star,\bGamma_{\underline{c}}^\star)
		h_K(\bgamma_c^\star:c\in\calC_{n-1}\cup\calC_{\varnothing})
		p_{\bSigma}(\bGamma_{\underline{c}}^\star)\mathrm{d}\bGamma_{\underline{c}}^\star\prod_{c\in\calC_\varnothing}p_\bmu(\bgamma_c^\star)\mathrm{d}\bgamma_{c}^\star}
		{\displaystyle
		\int\cdots\int
		h_K(\bgamma_c^\star:c\in\mathcal{C}_{n-1}\cup \mathcal{C}_\varnothing)\prod_{c\in\mathcal{C}_\varnothing}p_\bmu(\bgamma_c^\star)\mathrm{d}\bgamma_{c}^\star
		}\right]
		\nonumber\\
		&\qquad\quad\times  \left[
		\frac{\displaystyle
		\iint_A \phi(\by_n\mid\bgamma_{\underline{c}}^\star,\bGamma^\star_{\underline{c}})L_K(\bgamma_{\underline{c}}^\star)p_\bmu(\bgamma_{\underline{c}}^\star)\mathrm{d}\bgamma_{\underline{c}}^\star\mathrm{d}\bGamma_{\underline{c}}^\star
		}{\displaystyle
		\iint \phi(\by_n\mid\bgamma_{\underline{c}}^\star,\bGamma^\star_{\underline{c}})L_K(\bgamma_{\underline{c}}^\star)p_\bmu(\bgamma_{\underline{c}}^\star)\mathrm{d}\bgamma_{\underline{c}}^\star\mathrm{d}\bGamma_{\underline{c}}^\star
		}\right]\nonumber\\
		&\qquad=\alpha_K G_K(A).\nonumber
	\end{align}
	The proof is thus completed. 
\end{proof}
\begin{proof}[Proof of \textbf{Theorem \ref{thm:Gibbs_sampler_auxiliary_variable}}]
We first check the first assertion. By definition
\begin{align}
&\iint_A
\phi(\by_i|\bgamma_{\underline{c}}^\star,\bGamma_{\underline{c}}^\star)\widetilde{g}\left(\bgamma_{\underline{c}}^\star,\bGamma_{\underline{c}}^\star\middle|\calC_{-i},\btheta_{-i}\right)\mathrm{d}\bgamma_{\underline{c}}^\star\mathrm{d}\bGamma_{\underline{c}}^\star\nonumber\\
&\qquad =\sum_{K=|\calC_{-i}|+1}^\infty p(K|\calC=\calC_{-i}\cup\{\{i\}\})
\frac
{\displaystyle \iint_A \phi(\by_i|\bgamma_{\underline{c}}^\star,\bGamma_{\underline{c}}^\star) L_K(\bgamma_{\underline{c}}^\star)p_\bmu(\bgamma_{\underline{c}}^\star)p_\bSigma(\bGamma_{\underline{c}}^\star)\mathrm{d}\bgamma_{\underline{c}}^\star\mathrm{d}\bGamma_{\underline{c}}^\star}
{\displaystyle \int L_K(\bgamma_{\underline{c}}^\star)p_\bmu(\bgamma_{\underline{c}}^\star)\mathrm{d}\bgamma_{\underline{c}}^\star}\nonumber\\
&\qquad = \sum_{K=|\calC_{-i}|+1} m_Kp(K|\calC=\calC_{-i}\cup\{\{i\}\})
\frac
{\displaystyle \iint_A \phi(\by_n|\bgamma_{\underline{c}}^\star,\bGamma_{\underline{c}}^\star) L_K(\bgamma_{\underline{c}}^\star)p_\bmu(\bgamma_{\underline{c}}^\star)p_\bSigma(\bGamma_{\underline{c}}^\star)\mathrm{d}\bgamma_{\underline{c}}^\star\mathrm{d}\bGamma_{\underline{c}}^\star}
{\displaystyle \iint \phi(\by_n|\bgamma_{\underline{c}}^\star,\bGamma_{\underline{c}}^\star) L_K(\bgamma_{\underline{c}}^\star)p_\bmu(\bgamma_{\underline{c}}^\star)p_\bSigma(\bGamma_{\underline{c}}^\star)\mathrm{d}\bgamma_{\underline{c}}^\star\mathrm{d}\bGamma_{\underline{c}}^\star}\nonumber\\
&\qquad = \sum_{K=|\calC_{-i}|+1} \alpha_KG_K(A)\nonumber,
\end{align}
and hence we see that
\begin{align}
\iint \phi(\by_i|\bgamma_{\underline{c}}^\star,\bGamma_{\underline{c}}^\star)\widetilde{g}(\bgamma_{\underline{c}}^\star,\bGamma_{\underline{c}}^\star|\calC_{-i},\btheta_{-i})\mathrm{d}\bgamma_{\underline{c}}^\star\mathrm{d}\bGamma_{\underline{c}}^\star=\sum_{K=|\calC_{-i}|+1} \alpha_K\nonumber.
\end{align}
Given observation $\by_i$, denote
\begin{align}
\widetilde{G}\left(A\middle|\by_i,\calC_{-i},\btheta_{-i}\right)=
\frac
{\displaystyle\iint_A\phi(\by_i|\bgamma_{\underline{c}}^\star,\bGamma_{\underline{c}}^\star)\widetilde{g}(\bgamma_{\underline{c}}^\star,\bGamma_{\underline{c}}^\star|\by_i,\calC_{-i},\btheta_{-i})\mathrm{d}\bgamma_{\underline{c}}^\star\mathrm{d}\bGamma_{\underline{c}}^\star}
{\displaystyle\iint\phi(\by_i|\bgamma_{\underline{c}}^\star,\bGamma_{\underline{c}}^\star)\widetilde{g}(\bgamma_{\underline{c}}^\star,\bGamma_{\underline{c}}^\star|\by_i,\calC_{-i},\btheta_{-i})\mathrm{d}\bgamma_{\underline{c}}^\star\mathrm{d}\bGamma_{\underline{c}}^\star},\nonumber
\end{align}
and let $\widetilde{g}(\bgamma_{\underline{c}}^\star,\bGamma_{\underline{c}}^\star|\by_i,\calC_{-i},\btheta_{-i})$ be the corresponding density of $\widetilde{G}(\cdot|\by_i,\calC_{-i},\btheta_{-i})$. 
By construction, given the auxiliary variable $\bgamma_{\underline{c}}^\star,\bGamma_{\underline{c}}^\star$, we have
\begin{align}
\mathbb{P}\left(\btheta_i\in A|\bgamma_{\underline{c}}^\star,\bGamma_{\underline{c}}^\star,\by_i,\calC_{-i},\btheta_{-i}\right)
&=\frac
{\displaystyle\left[\frac{V_n(|\calC_{-i}|+1)\beta}{V_n(|\calC_{-i}|)}\right]\phi(\by_i|\bgamma_{\underline{c}}^\star,\bGamma_{\underline{c}}^\star)\delta_{(\bgamma_{\underline{c}}^\star,\bGamma_{\underline{c}}^\star)}(A)}
{\displaystyle\left[\frac{V_n(|\calC_{-i}|+1)\beta}{V_n(|\calC_{-i}|)}\right]\phi(\by_i|\bgamma_{\underline{c}}^\star,\bGamma_{\underline{c}}^\star)+\sum_{c\in\calC_{-i}}(|c|+\beta)\phi(\by_i|\bgamma_c^\star,\bGamma_c^\star)}\nonumber\\
&\quad+ \frac
{\displaystyle \sum_{c\in\calC_{-i}}(|c|+\beta)\phi(\by_i|\bgamma_{{c}}^\star,\bGamma_{{c}}^\star)\delta_{(\bgamma_{{c}}^\star,\bGamma_{{c}}^\star)}(A)}
{\displaystyle\left[\frac{V_n(|\calC_{-i}|+1)\beta}{V_n(|\calC_{-i}|)}\right]\phi(\by_i|\bgamma_{\underline{c}}^\star,\bGamma_{\underline{c}}^\star)+\sum_{c\in\calC_{-i}}(|c|+\beta)\phi(\by_i|\bgamma_c^\star,\bGamma_c^\star)}\nonumber.
\end{align}
Integrate the RHS of the last display against $p(\bgamma_{\underline{c}}^\star,\bGamma_{\underline{c}}^\star|\by_i,\calC_{-i},\btheta_{-i})$ yields
\begin{align}
&\mathbb{P}\left(\btheta_i\in A|\by_i,\calC_{-i},\btheta_{-i}\right)\nonumber\\
&\qquad=\iint \mathbb{P}\left(\btheta_i\in A|\bgamma_{\underline{c}}^\star,\bGamma_{\underline{c}}^\star,\by_i,\calC_{-i},\btheta_{-i}\right)p(\bgamma_{\underline{c}}^\star,\bGamma_{\underline{c}}^\star|\by_i,\calC_{-i},\btheta_{-i})\mathrm{d}\bgamma_{\underline{c}}^\star\mathrm{d}\bGamma_{\underline{c}}^\star\nonumber\\
&\qquad=\frac
{\displaystyle\left[\frac{V_n(|\calC_{-i}|+1)\beta}{V_n(|\calC_{-i}|)}\right]\iint_A\phi(\by_i|\bgamma_i,\bGamma_i)\widetilde{g}(\bgamma_i,\bGamma_i|\calC_{-i},\btheta_{-i})\mathrm{d}\bgamma_i\mathrm{d}\bGamma_i}
{\displaystyle\left[\frac{V_n(|\calC_{-i}|+1)\beta}{V_n(|\calC_{-i}|)}\right]\sum_{K=|\calC_{-i}|+1}^\infty\alpha_K+\sum_{c\in\calC_{-i}}(|c|+\beta)\phi(\by_i|\bgamma{}_c^\star,\bGamma{}_c^\star)}\nonumber\\
&\qquad\quad + \frac
{\displaystyle\sum_{c\in\calC_{-i}}(|c|+\beta)\phi(\by_i|\bgamma_c^\star,\bGamma{}_c^\star)\delta_{(\bgamma_c^\star,\bGamma_c^\star)}(A)}
{\displaystyle\left[\frac{V_n(|\calC_{-i}|+1)\beta}{V_n(|\calC_{-i}|)}\right]\sum_{K=|\calC_{-i}|+1}^\infty\alpha_K+\sum_{c\in\calC_{-i}}(|c|+\beta)\phi(\by_i|\bgamma{}_c^\star,\bGamma{}_c^\star)}\nonumber\\
&\qquad\propto \left[\frac{V_n(|\calC_{-i}|+1)\beta}{V_n(|\calC_{-i}|)}\right]\iint_A\phi(\by_i|\bgamma_i,\bGamma_i)\widetilde{g}(\bgamma_i,\bGamma_i|\calC_{-i},\btheta_{-i})\mathrm{d}\bgamma_i\mathrm{d}\bGamma_i
\nonumber\\&\qquad\quad 
+ \sum_{c\in\calC_{-i}}(|c|+\beta)\phi(\by_i|\bgamma_c^\star,\bGamma_c^\star)\delta_{(\bgamma_c^\star,\bGamma_c^\star)}(A)\nonumber\\
&\qquad = \left[\frac{V_n(|\calC_{-i}|+1)\beta}{V_n(|\calC_{-i}|)}\right]\sum_{K=|\calC_{-i}|+1}^\infty\alpha_KG_K(A)
+ \sum_{c\in\calC_{-i}}(|c|+\beta)\phi(\by_i|\bgamma_c^\star,\bGamma_c^\star)\delta_{(\bgamma_c^\star,\bGamma_c^\star)}(A)\nonumber,
\end{align}
which coincides with \eqref{eqn:Gibbs_sampler_urn_model}.
This completes the proof the first assertion. For the second assertion, 
by construction, we have
\begin{align}
&\mathbb{P}\left(\calC=(\calC_{-i}\backslash\{c\})\cup(\{c\cup\{i\}\}),(\bgamma_{\underline{c}}^\star,\bGamma_{\underline{c}}^\star)\in A\middle|\calC_{-i},\btheta_{-i},\by_i\right)\nonumber\\
&\quad = 
\iint_A \mathbb{P}\left(\calC=(\calC_{-i}\backslash\{c\})\cup(\{c\cup\{i\}\})\middle|\calC_{-i},\btheta_{-i},\bgamma_{\underline{c}}^\star,\bGamma_{\underline{c}}^\star,\by_i\right)p(\bgamma_{\underline{c}}^\star,\bGamma_{\underline{c}}^\star|\by_i,\bgamma_c^\star,\bGamma_c^\star,c\in\calC_{-i})\mathrm{d}\bgamma_{\underline{c}}^\star\mathrm{d}\bGamma_{\underline{c}}^\star\nonumber\\
&\quad = \iint_A 
\frac{(|c|+\beta)\phi(\by_i|\bgamma_c^\star,\bGamma_c^\star)p(\bgamma_{\underline{c}}^\star,\bGamma_{\underline{c}}^\star|\by_i,\bgamma_c^\star,\bGamma_c^\star,c\in\calC_{-i})}
{\displaystyle\left[\frac{V_n(|\calC_{-i}|+1)\beta}{V_n(|\calC_{-i}|)}\right]\phi(\by_i|\bgamma_{\underline{c}}^\star,\bGamma_{\underline{c}}^\star)+\sum_{c\in\calC_{-i}}(|c|+\beta)\phi(\by_i|\bgamma_c^\star,\bGamma_c^\star)}\mathrm{d}\bgamma_{\underline{c}}^\star\mathrm{d}\bGamma_{\underline{c}}^\star\nonumber\\
&\quad = \iint_A 
\frac{(|c|+\beta)\phi(\by_i|\bgamma_c^\star,\bGamma_c^\star)\widetilde{g}(\bgamma_{\underline{c}}^\star,\bGamma_{\underline{c}}^\star|\bgamma_c^\star,\bGamma_c^\star,c\in\calC_{-i})}
{\displaystyle\left[\frac{V_n(|\calC_{-i}|+1)\beta}{V_n(|\calC_{-i}|)}\right]\sum_{K=|\calC_{-i}|+1}^\infty\alpha_K+\sum_{c\in\calC_{-i}}(|c|+\beta)\phi(\by_i|\bgamma_c^\star,\bGamma_c^\star)}\mathrm{d}\bgamma_{\underline{c}}^\star\mathrm{d}\bGamma_{\underline{c}}^\star\nonumber.
\end{align}
Since given $\calC=(\calC_{-i}\backslash\{c\})\cup(\{c\cup\{i\}\})$, $\btheta_{-i}=(\btheta_1,\cdots,\btheta_n)$, it follows that the conditional distribution can be directly computed:
\begin{align}
&\mathbb{P}\left((\bgamma_{\underline{c}}^\star,\bGamma_{\underline{c}}^\star)\in A\middle|\calC=(\calC_{-i}\backslash\{c\})\cup(\{c\cup\{i\}\}), \calC_{-i},\btheta_{-i},\btheta_i,\by_i\right)\nonumber\\
&\quad = \mathbb{P}\left((\bgamma_{\underline{c}}^\star,\bGamma_{\underline{c}}^\star)\in A\middle|\calC=(\calC_{-i}\backslash\{c\})\cup(\{c\cup\{i\}\}), \calC_{-i},\btheta_{-i},\by_i\right)\nonumber\\
&\quad = \frac{\mathbb{P}\left(\calC=(\calC_{-i}\backslash\{c\})\cup(\{c\cup\{i\}\}),(\bgamma_{\underline{c}}^\star,\bGamma_{\underline{c}}^\star)\in A\middle|\calC_{-i},\btheta_{-i},\by_i\right)}{\mathbb{P}\left(\calC=(\calC_{-i}\backslash\{c\})\cup(\{c\cup\{i\}\})\middle|\calC_{-i},\btheta_{-i},\by_i\right)}\nonumber\\
&\quad = \widetilde{G}(A|\bgamma_c^\star,\bGamma_c^\star,c\in\calC_{-i})\nonumber.
\end{align}
On the other hand, we know from definition that
\begin{align}
\mathbb{P}\left(\btheta_i\in A|\calC=\calC_{-i}\cup\{\{i\}\},\by_i,\bgamma_{\underline{c}}^\star,\bGamma_{\underline{c}}^\star,\btheta_{-i},\calC_{-i}\right)=\delta_{(\bgamma_{\underline{c}}^\star,\bGamma{}_{\underline{c}}^\star)}(A).\nonumber
\end{align}
It follows directly that $\mathbb{P}\left((\bgamma_{\underline{c}}^\star,\bGamma{}_{\underline{c}}^\star)\in A|\calC_{-i}\cup\{\{i\}\},\by_i,\btheta_i,\btheta_{-i},\calC_{-i}\right)=\delta_{\btheta_i}(A)$, and hence the second assertion is proved. 
\end{proof}


\clearpage
\section{Details of Posterior Inference} 
\label{sec:details_of_posterior_inference}
In this section we provide the detailed blocked-collapsed Gibbs sampler in \textbf{Algorithm \ref{alg:blocked_collapsed_gibbs_sampler}} when a conjugate prior on the covariance matrices for all components is used: $\bSigma_k=\mathrm{diag}(\lambda_{1k},\cdots,\lambda_{pk})$ and $\lambda_{jk}\iidsim p(\lambda)\propto \mathbb{I}(\lambda\in[\usigma^{-2},\lsigma^{-2}])\lambda^{-a_0-1}\exp(-b_0/\lambda)$, $j=1,\cdots,p,k=1,\cdots,K$. Easy extension of the sampler is available when one use Inverse-Wishart distribution on the non-diagonal covariance matrices $\bSigma_k$'s. A practical issue for the implementation of the Gibbs sampler is sampling from the conditional prior $p(K|\mathcal{C})$ as well as the conditional posterior $p(K|\mathcal{C},\by_1,\cdots,\by_n,\bGamma_c^\star:c\in\calC)$. Using formula (3.7) in \cite{miller2016mixture}, we see that 
\[
p(K\mid\mathcal{C})\propto \frac{K!}{(K+n)!(K-|\mathcal{C}|)!}.
\] 
Notice that for $K>>|\mathcal{C}|$, $p(K)\approx 0$, and therefore in practice one may use the following approximate sampling scheme
\begin{eqnarray}\label{eqn:approximate_p_K_given_partition}
p(K\mid\mathcal{C})\propto \frac{K!}{(K+n)!(K-|\mathcal{C}|)!},\quad  K=|\mathcal{C}|,|\mathcal{C}|+1,\cdots,|\mathcal{C}|+m\nonumber	
\end{eqnarray}
for a moderate choice of the perturbation range $m$, especially when $n$ is large, in which case the probability of having large number of empty components(\emph{i.e.} $K>>|\mathcal{C}|$) is negligible. 
Sampling from the conditional posterior $p(K\mid\mathcal{C},\by_1,\cdots,\by_n,\bGamma_c^\star:c\in\calC)$, however, is a slightly harder issue. 
Denote 
\begin{align}
p_\bmu(\bgamma_c^\star\mid\by_i:i\in c,\bGamma_c^\star)=\frac{\displaystyle p_\bmu(\bgamma_c^\star)\prod_{i\in c}\phi(\by_i|\bgamma_c^\star,\bGamma_c^\star)}{\displaystyle \int p_\bmu(\bgamma_c^\star)\prod_{i\in c}\phi(\by_i\mid\bgamma_c^\star,\bGamma_c^\star)\mathrm{d}\bgamma_c^\star}\nonumber
\end{align}
to be the conditional posterior of $\bgamma_c^\star$ given observations when the repulsive prior is not introduced, namely, when $\bgamma_c^\star\sim p_\bmu$ independently. Given the partition $\calC$, the cluster-spefic covariance matrices $(\bGamma_c^\star:c\in\calC)$, and the observations $(\by_i)_{i=1}^n$, the posterior of $(\bgamma_c^\star:c\in\calC)$ 
\begin{align}
p(\bgamma_c^\star:c\in\calC\mid\bGamma_c^\star,c\in\calC,(\by_i)_{i=1}^n)
\propto\sum_{K=|\calC|}^\infty p(\bgamma_c^\star:c\in\calC\mid K,\bGamma_c^\star,c\in\calC,(\by_i)_{i=1}^n)p(K\mid \bGamma_c^\star,c\in\calC,(\by_i)_{i=1}^n)\nonumber,
\end{align}
where
\begin{align}
&p(\bgamma_c^\star:c\in\calC \mid K,\bGamma_c^\star,c\in\calC,(\by_i)_{i=1}^n)\nonumber\\
&\qquad\propto
{\displaystyle\int\cdots\int h_K(\bgamma_c^\star:c\in\calC\cup\calC_\varnothing)\left[\prod_{c\in\calC}p(\bgamma_c^\star\mid\by_i:i\in c,\bGamma_c^\star)\right]\left[\prod_{c\in\calC_\varnothing}p_\bmu(\bgamma_c^\star)\mathrm{d}\bgamma_c^\star\right]}
\nonumber,
\end{align}
and
\begin{align}
p(K\mid\bGamma_c^\star,c\in\calC,(\by_i)_{i=1}^n)\propto\frac{p(K\mid\calC)}{Z_K}\int\cdots\int h_K(\bgamma_c^\star:c\in\calC\cup\calC_\varnothing)\left[\prod_{c\in\calC\cup\calC_\varnothing}p(\bgamma_c^\star\mid\by_i:i\in c,\bGamma_c^\star)\mathrm{d}\bgamma_c^\star\right]\nonumber.
\end{align}
\textbf{Step 4} of the blocked-collapsed Gibbs sampler in Section \ref{sec:posterior_inference} of the manuscript samples from $p(\bgamma_c^\star:c\in\calC|K,\bGamma_c^\star,c\in\calC,(\by_i)_{i=1}^n)$. To sample from $p(K|\bGamma_c^\star,c\in\calC,(\by_i)_{i=1}^n)$, we use numerically compute $Z_K$ and the intractable integral when sampling from $p(\bgamma_c^\star|\by_i:i\in c,\bGamma_c^\star)$ is tractable, which is usually the case when $p_\bmu$ is the conjugate normal prior. In what follows we provide the detailed blocked-collapsed Gibbs sampler. Alternatively, to gain computational efficiency, one can use $p(K\mid\calC)$ to approximate $p(K\mid\bGamma_c^\star,c\in\calC,(\by_i)_{i=1}^n)$ in the resampling steps. 

\begin{algorithm}[h] 
	\renewcommand{\algorithmicrequire}{\textbf{Input:}}
	\renewcommand{\algorithmicensure}{\textbf{Output:} }
	\caption{Blocked-Collapsed Gibbs Sampler} 
	\label{alg:blocked_collapsed_gibbs_sampler} 
	\begin{algorithmic}[1] 
		\State{\textbf{Input: }
			\State{\INDSTATE Observations $(\by_i)_{i=1}^n$; }
			\State{\INDSTATE Hyperparameters $(a_0,b_0)$, $\tau,g_0$, $0<\lsigma<\usigma<\infty$; }
			\State{\INDSTATE Burn-in time $B$; }
			\State{\INDSTATE Number of posterior samples $T$; }
			\State{\INDSTATE Guess upper bound $K_{\max}$ on $K$; }
			\State{\INDSTATE Perturbation range $m$ for approximate sampling $p(K\mid\mathcal{C})$. }
			}
		\State{\textbf{Initialize: }
			\State{\INDSTATE Set $\ell=1$, select $K\leq n$, and sample $\bmu_k\sim\mathrm{N}(\mathbf{0},\tau^2\eye_p)$; }
			\State{\INDSTATE For $k=1,\cdots,K$, set $\bSigma_k=\eye_p$; }
			\State{\INDSTATE For $i=1,\cdots,n$, set $\btheta_i=(\bmu_k,\bSigma_k)$ if $k=\mathrm{argmax}_k\phi(\by_i\mid\bmu_k,\bSigma_k)$; }
			\State{\INDSTATE Compute $\mathcal{C}$ from $(\btheta_1,\cdots,\btheta_n)$; }
			\State{\INDSTATE Set $(\btheta_1^{(1)},\cdots,\btheta_n^{(1)})=(\btheta_1,\cdots,\btheta_n)$. }
			\State{\INDSTATE For $k=1,\cdots,K_{\max}$, numerically compute the normalization constants $Z_K$. }
		}
		\algstore{blocked_collapsed_gibbs_sampler}
	\end{algorithmic}
\end{algorithm}

\begin{algorithm}[h] 
	\renewcommand{\algorithmicrequire}{\textbf{Input:}}
	\renewcommand{\algorithmicensure} {\textbf{Output:} }
	\begin{algorithmic}[1] 
	\algrestore{blocked_collapsed_gibbs_sampler}
	\For{$t_{\mathrm{it}}=2,\cdots,(B+T)$}
		\State{\textbf{Set} $(\btheta_1,\cdots,\btheta_n)=(\btheta_1^{(t_{\mathrm{it}}-1)},\cdots,\btheta_n^{(t_{\mathrm{it}}-1)})$}
		\For{$i=1,\cdots,n$}
		\State{\textbf{Set} $\btheta_{-i}=(\btheta_1,\cdots,\btheta_n)\backslash\{\btheta_i\}$; \textbf{Compute} $\mathcal{C}_{-i}$ from $\btheta_{-i}$; }
		\State{\textbf{Set} $\ell=|\mathcal{C}_{-i}|$; \textbf{Label} $(c:c\in\mathcal{C}_{-i})$ as $(c_1,\cdots,c_\ell)$; }
		\For{$k=1,\cdots,\ell$}
			\State{\textbf{Set} $(\bgamma_k^\star,\bGamma_k^\star)=(\bgamma_i,\bGamma_i)$ if $i\in c_k$; }
		\EndFor. 
		\If{$\calC=\calC_{-i}\cup\{\{i\}\}$}
			{\textbf{Set} $(\bgamma_{\underline{c}}^\star,\bGamma_{\underline{c}}^\star)=\btheta_i$; }
		\Else
			\State{\textbf{Sample} $K$ from
			$$
				p(K\mid \mathcal{C})\propto \frac{K!}{(K-\ell)!(K+n)!},\quad K=\ell,\ell+1,\cdots,\ell+m.
			$$
			}
			\For{$j=1,\cdots,p$}
				\State{\textbf{Sample }$\lambda_j^\star$ from $\text{Inv-Gamma}(a_0,b_0)\mathbb{I}([\lsigma^2,\usigma^2])$; }
			\EndFor
			\State{\textbf{Set }$\bSigma_K^\star=\mathrm{diag}(\lambda_1^\star,\cdots,\lambda_p^\star)$; }
			\For{$k = \ell + 1,\cdots, K$}
				\State{\textbf{Sample} $\bgamma_k^\star\sim\mathrm{N}\left(\mathbf{0},\tau^2\eye_p\right)$; }
			\EndFor
			\State{\textbf{Sample} $U\sim\mathrm{Unif}(0,1)$, and compute
				$$
				g(\bgamma_1^\star,\cdots,\bgamma_K^\star)=\min_{1\leq k<k'\leq K}\left(\frac{\|\bgamma_k^\star-\bgamma_{k'}^\star\|}{g_0+\|\bgamma_k^\star-\bgamma_{k'}^\star\|}\right);
				$$}
			\If{$U<g(\bgamma_1^\star,\cdots,\bgamma_K^\star)$}
				{\textbf{Accept} the new proposed samples; }
			\Else
				{ \textbf{Go} to line NO. 32 and resample. }
			\EndIf
			\State{\textbf{Set} $\bgamma_{\underline{c}}^\star=\bgamma_K^\star$; }
		\EndIf
		\State{\textbf{Sample} $\mathcal{C}$ according to the categorical distribution
		\begin{align}
			&\mathbb{P}(\mathcal{C}=\mathcal{C}_{-i}\cup\{\{i\}\}\mid-)\propto\left[\frac{V_{n}(\ell+1)}{V_{n}(\ell)}\beta\right]\phi(\by_i\mid\bgamma_{\underline{c}}^\star,\bGamma_{\underline{c}}^\star),\nonumber\\
			&\mathbb{P}(\mathcal{C}=(\mathcal{C}_{-i}\backslash\{c\})\cup\{c\cup\{i\}\}\mid-)\propto\phi(\by_i\mid\bgamma_c^\star,\bGamma_c^\star)\left(|c|+\beta\right),\quad c\in\mathcal{C}_{-i}.\nonumber
		\end{align}}
		\If{$\calC = \calC_{-i}\cup\{\{i\}\}$}
			\State{\textbf{Set} $\btheta_i=(\bgamma_{\underline{c}}^\star,\bGamma_{\underline{c}}^\star)$}
		\ElsIf{$\calC = (\calC_{-i}\backslash\{c\})\cup(\{c\cup \{i\}\})$ for some $c\in\calC_{-i}$}
			\State{\textbf{Set} $\btheta_i=(\bgamma_j^\star,\bGamma_j^\star)$ for some $j\in c$; }
		\EndIf
		\EndFor
	\algstore{blocked_collapsed_gibbs_sampler2}
	\end{algorithmic}
\end{algorithm}

\begin{algorithm}[h] 
	\renewcommand{\algorithmicrequire}{\textbf{Input:}}
	\renewcommand{\algorithmicensure} {\textbf{Output:} }
	\begin{algorithmic}[1] 
		\algrestore{blocked_collapsed_gibbs_sampler2}	
		\State{\textbf{Set} $\ell=|\mathcal{C}|$, label $(c:c\in\mathcal{C})$ as $(c_1,\cdots,c_\ell)$; }
		\For{$k=1,\cdots,\ell$}
			\State{\textbf{Set} $(\bgamma_k^\star,\bGamma_k^\star)=(\bgamma_i,\bGamma_i)$ if $i\in c_k$; }
		\EndFor
		\For{$K=\ell,\cdots,\ell+m-1$}
			\State{Numerically compute the integral}
			\[
			\widetilde{Z}_K=\int\cdots\int h_K(\bgamma_c^\star:c\in\calC\cup\calC_\varnothing)\left[\prod_{c\in\calC\cup\calC_\varnothing}p(\bgamma_c^\star|\by_i:i\in\calC,\bGamma_c^\star)\mathrm{d}\bgamma_c^\star\right].
			\]
		\EndFor
		\State{\textbf{Sample} $K$ from
		$$
		p(K\mid\bGamma_c^\star,c\in\mathcal{C},(\by_i)_{i=1}^n)\propto \frac{\tilde{Z}_K K!}{Z_K(K-\ell)!(K+n)!},\quad K=\ell,\ell+1,\cdots,\ell+m.
		$$}
		\For{$k=1,\cdots,\ell$}
			\State{\textbf{Set} $\bgamma_k^\star=(\gamma_{1k}^\star,\cdots,\gamma_{pk}^\star)\transpose$; }
			\For{$j=1,\cdots,p$}
				\State{\textbf{Sample} $\lambda_{jk}$ from 
				$$
				\text{Inv-Gamma}\left(a_0+\frac{|c_k|}{2},b_0+\frac{1}{2}\sum_{i\in c_k}(y_{ij}-\gamma_{kj}^\star)^2\right)\mathbb{I}([\lsigma^2,\usigma^2]);
				$$
				}
			\EndFor
			\State{\textbf{Set} $\bGamma_k^\star=\mathrm{diag}(\lambda_{1k}^\star,\cdots,\lambda_{pk}^\star)$; }
		\EndFor
		\For{$k=1,\cdots,K$}
				\State{\textbf{Sample} $\bgamma_k^\star$ from $\mathrm{N}(\mb_k,\bV_k)$ where
				\begin{align}
					\bV_k&=\left(\bGamma_k^{\star-1}\sum_{i=1}^n\mathbb{I}(\bgamma_i=\bgamma_k^\star)+\frac{1}{\tau^2}\eye_p\right)^{-1},\nonumber\\
					\mb_k&=\bV_k\left(\bGamma_k^{\star-1}\sum_{i=1}^n\mathbb{I}(\bgamma_i=\bgamma_k^\star)\by_i\right)\nonumber;
					\end{align}
				}
		\EndFor
		\State{\textbf{Sample} $U$ from $\mathrm{Unif}(0,1)$;
		}
		\If{$U<g(\bgamma_1^\star,\cdots,\bgamma_K^\star)$}
			{\textbf{Accept} the new proposed samples $(\bgamma_1^\star,\cdots,\bgamma_K^\star)$; }
		\Else
			{\textbf{Go} to Line NO. 61 and resample; }
		\EndIf; 
		\For{$i=1,\cdots,n$}
			{\textbf{Set} $\btheta_i=(\bgamma_i,\bGamma_i)=(\bgamma_k^\star,\bGamma_k^\star)$ if $i\in c_k$. }
		\EndFor
		\State{\textbf{Change} the current state to $(\btheta_1^{(t_{\mathrm{it}})},\cdots,\btheta_n^{(t_{\mathrm{it}})})=(\btheta_1,\cdots,\btheta_n)$;}
	\EndFor
	\State{\textbf{Output: }
			\State{\INDSTATE Posterior Samples $(\btheta_1^{(t_{\mathrm{it}})},\cdots,\btheta_n^{(t_{\mathrm{it}})})_{t_{\mathrm{it}}=B+1}^{B+T}$, where $\btheta_i^{(t_{\mathrm{it}})}=(\bgamma_i^{(t_{\mathrm{it}})},\bGamma_i^{(t_{\mathrm{it}})})$, $i=1,\cdots,p$}
		}
	\end{algorithmic}
\end{algorithm}

\clearpage
\section{Convergence Diagnostics}
\label{sec:convergence_diagnostics}
\subsection*{Convergence Check for Subsection \ref{sub:multi_modal_estimation_finite_gaussian_mixtures}}
We check convergence via the trace plots and autocorrelations of some randomly selected $\bgamma_i$'s (which are identifiable compared to the exact means for different components) in Figure \ref{trace_plot_acf_mu}, showing no signs of non-convergence. 
\begin{figure}[h!]
	\centerline{\includegraphics[width=1\textwidth]{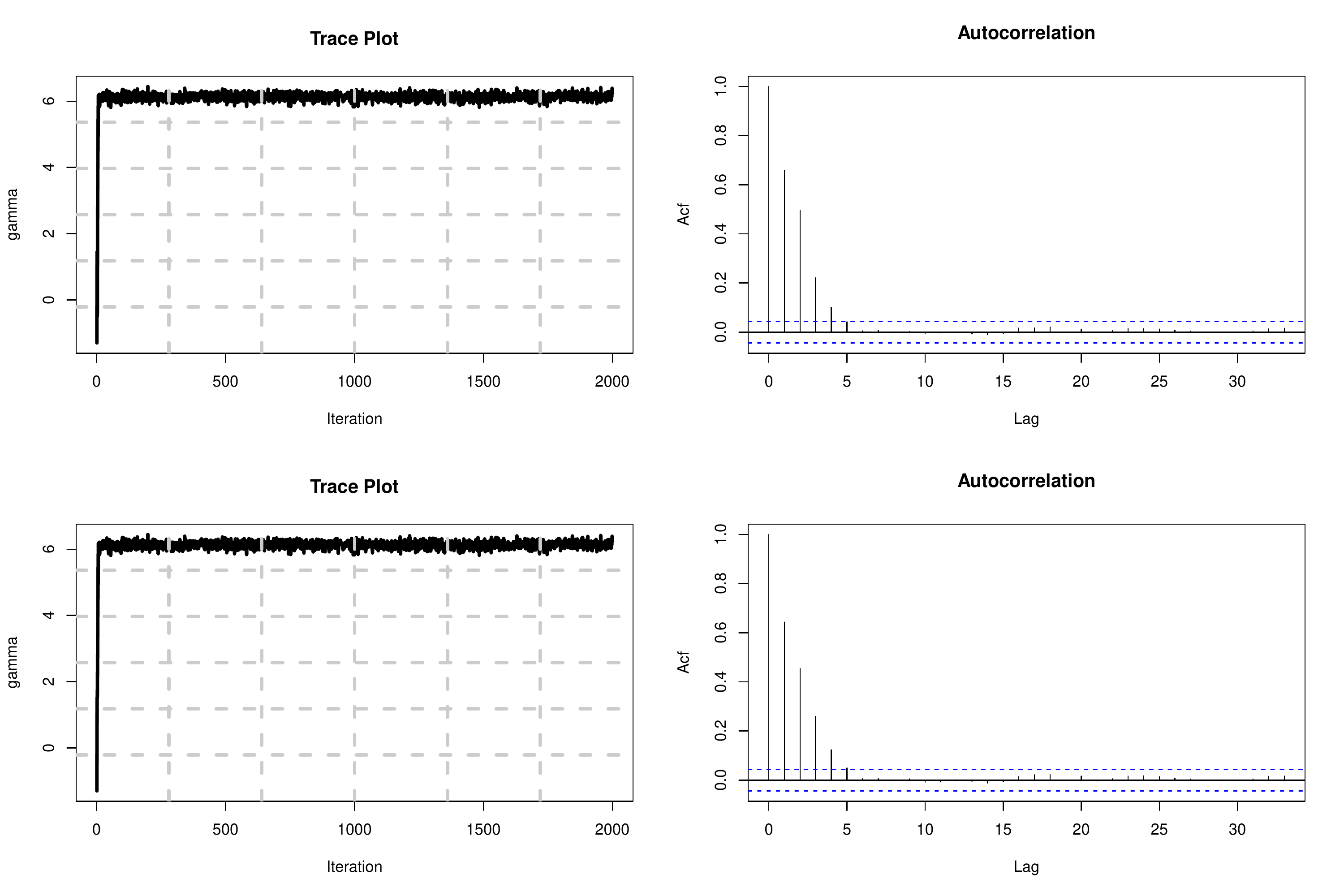}}
	\caption{Fitting Multi-Modal Density: The trace plots and the autocorrelation plots of the post-burn-in posterior samples of some randomly selected $\bgamma_i$'s. }
	\label{trace_plot_acf_mu}
\end{figure}

\newpage
\subsection*{Convergence Check for Subsection \ref{sub:uni_modal_density_continuous_gaussian_mixtures}}
We check convergence via the trace plots and the autocorrelations of some randomly selected $\bgamma_i$'s in Figure \ref{trace_plot_acf_mu}, showing no signs of non-convergence. 
\begin{figure}[h!]
	\centerline{\includegraphics[width=1\textwidth]{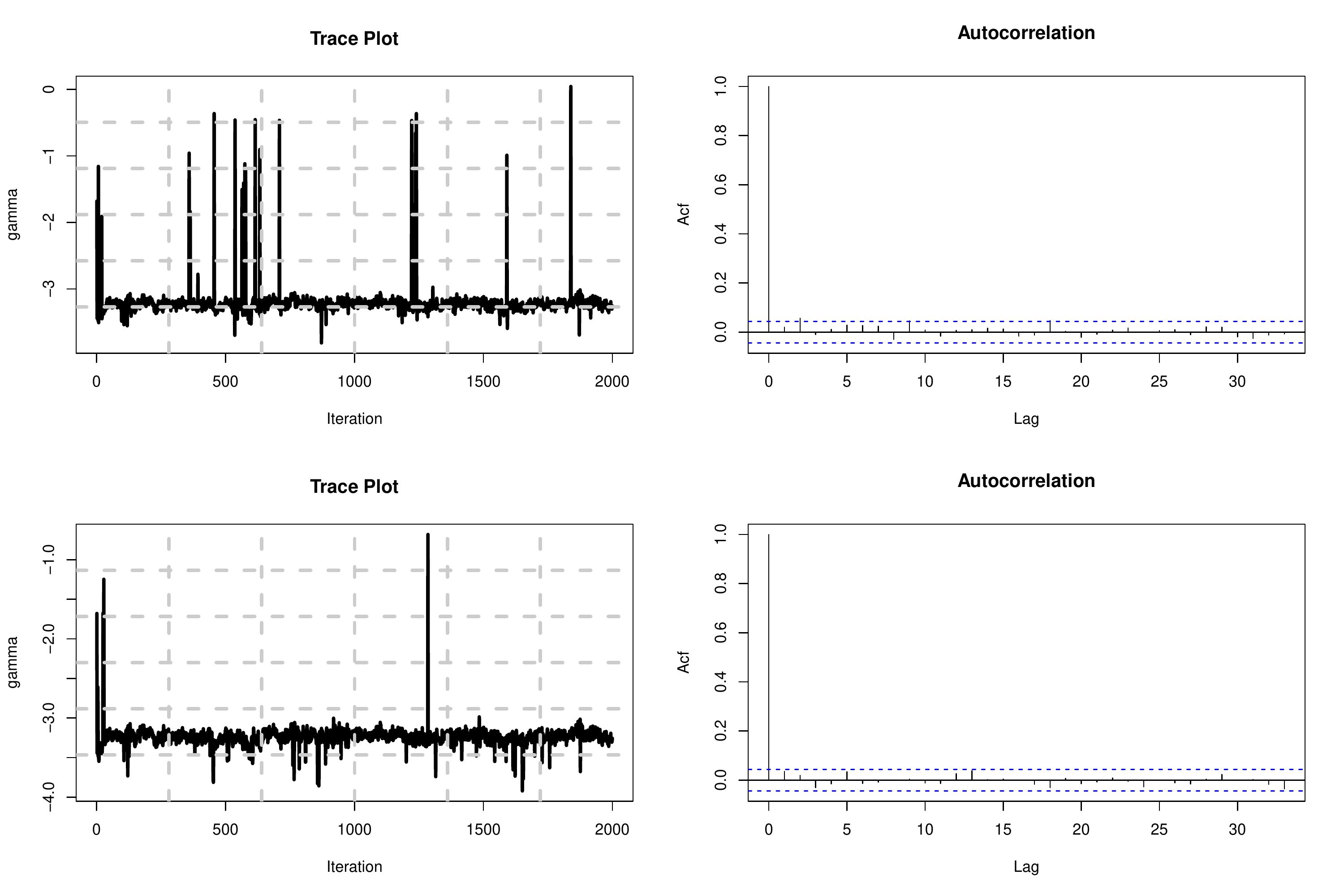}}
	\caption{Fitting Uni-Modal Density: The trace plots and the autocorrelation plots of the post-burn-in posterior samples of some randomly selected $\bgamma_i$'s. }
	\label{trace_plot_acf_mu_unimodal}
\end{figure}

\newpage
\subsection*{Convergence Check for Subsection \ref{sub:multivariate_model_based_clustering}}
The trace plots and the autocorrelations of some randomly selected $\bgamma_i$'s in Figure \ref{trace_plot_acf_mu_10D}, indicate no signs of non-convergence. 

\begin{figure}[h!]
	\centerline{\includegraphics[width=1\textwidth]{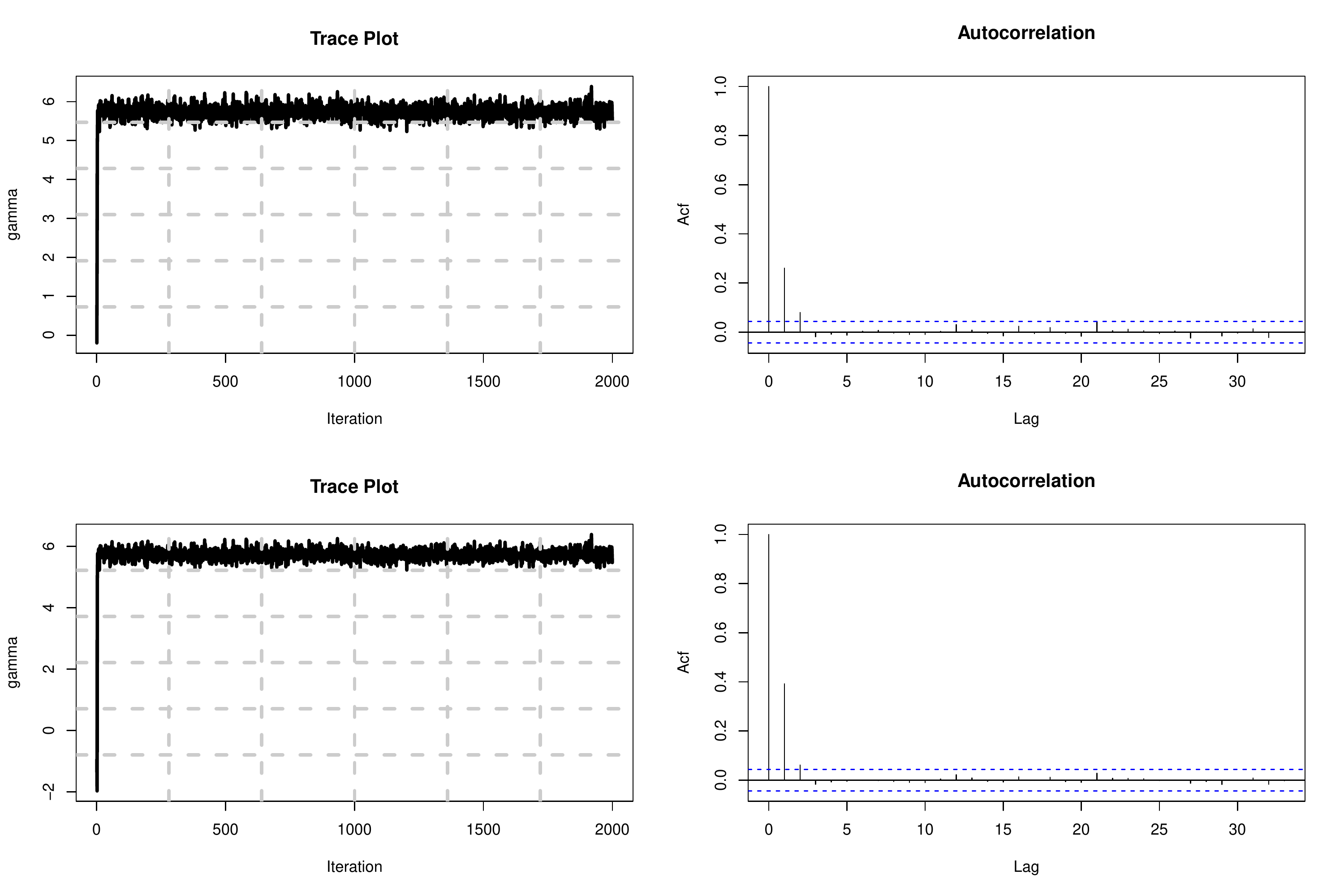}}
	\caption{Multivariate Model-Based Clustering: The trace plots and the autocorrelation plots of the post-burn-in posterior samples of some randomly selected $\bgamma_i$'s. }
	\label{trace_plot_acf_mu_10D}
\end{figure}

\newpage
\subsection*{Convergence Check for Subsection \ref{sub:real_data_analysis}}
The trace plots and the autocorrelations of some randomly selected $\bgamma_i$'s in Figure \ref{trace_plot_acf_mu_10D}, indicate no signs of non-convergence. 

\begin{figure}[h!]
	\centerline{\includegraphics[width=1\textwidth]{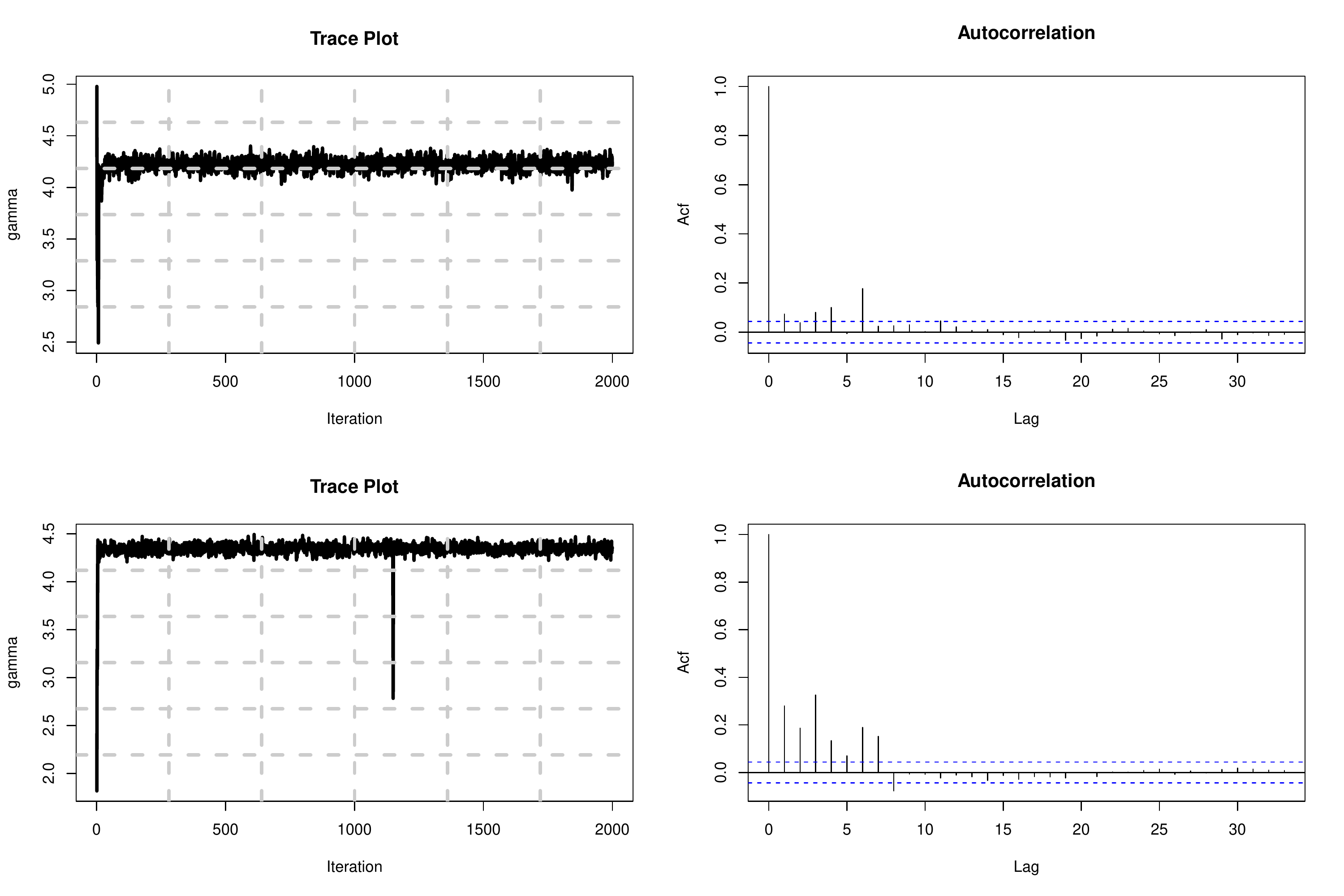}}
	\caption{Old Faithful Geyser Eruption Data: Trace plots and autocorrelation plots of the post-burn-in posterior samples of some randomly selected $\bgamma_i$'s. }
	\label{trace_plot_acf_mu_faithful}
\end{figure}

\clearpage

\section{Additional Simulation Study} 
\label{sec:additional_simulation_study}
In this section we consider a synthetic example where the number of observations and number of components are moderately large. The ground true density is given by a mixture of $K=13$ Gaussians. The first $12$ Gaussian components are equally weighted with mixing weight being $1/24$, and the weight of the last component is $12/24$. The first $12$ components are centered at
\begin{align}
&\left[\begin{array}{c}6\\6\end{array}\right],\quad
\left[\begin{array}{c}6\\12\end{array}\right],\quad
\left[\begin{array}{c}12\\6\end{array}\right],\quad
\left[\begin{array}{c}-6\\6\end{array}\right],\quad    
\left[\begin{array}{c}-6\\12\end{array}\right],\quad
\left[\begin{array}{c}-12\\6\end{array}\right],\nonumber\\
&\left[\begin{array}{c}6\\-6\end{array}\right],\quad
\left[\begin{array}{c}6\\12\end{array}\right],\quad
\left[\begin{array}{c}12\\-6\end{array}\right],\quad
\left[\begin{array}{c}-6\\-6\end{array}\right],\quad
\left[\begin{array}{c}-6\\-12\end{array}\right],\quad
\left[\begin{array}{c}-12\\-6\end{array}\right],\nonumber
\end{align}
respectively, with identical covariance matrix $\eye_2$. The last component is centered at the origin with covariance matrix $30\eye_2$. 

We collect $2000$ i.i.d. observations from this Gaussian mixture distribution, and implement the proposed blocked-collapsed Gibbs sampler with $g_0 = 10$, $\tau = 10$, $m = 2,\lsigma=0.1, \usigma = 10$, and a total number of 2000 iterations with the first $1000$ iterations discarded as burn-in. For comparison, we consider the following DPM model,
\begin{eqnarray}
(\by_i\mid\bmu_{z_i},\bSigma_{z_i})\sim N(\bmu_{z_i},\bSigma_{z_i}),\quad (\bmu_{z_i},\bSigma_{z_i}\mid G)\iidsim G,\quad\text{and }(G\mid\alpha, G_0)\sim\mathrm{DP}(\alpha, G_0), \nonumber
\end{eqnarray}
where $G_0=\mathrm{N}(\bmu, \bSigma)$ with $\bmu \sim \mathrm{N}\left(\mb_1,\bSigma/k_0\right)$ and $\bSigma\sim\text{Inv-Wishart}(4,\boldsymbol{\Psi}_1)$, $\alpha\sim\mathrm{Gamma}(1, 1)$, $\mb_1\sim\mathrm{N}(\mathbf{0},2\eye_2)$, $k_0\sim\mathrm{Gamma}(0.5,0.5)$, and $\boldsymbol{\Psi}_1\sim\text{Inv-Wishart}(4, 0.5\eye_2)$. 
\begin{figure}[hb!]
	\centerline{\includegraphics[width=1\textwidth]{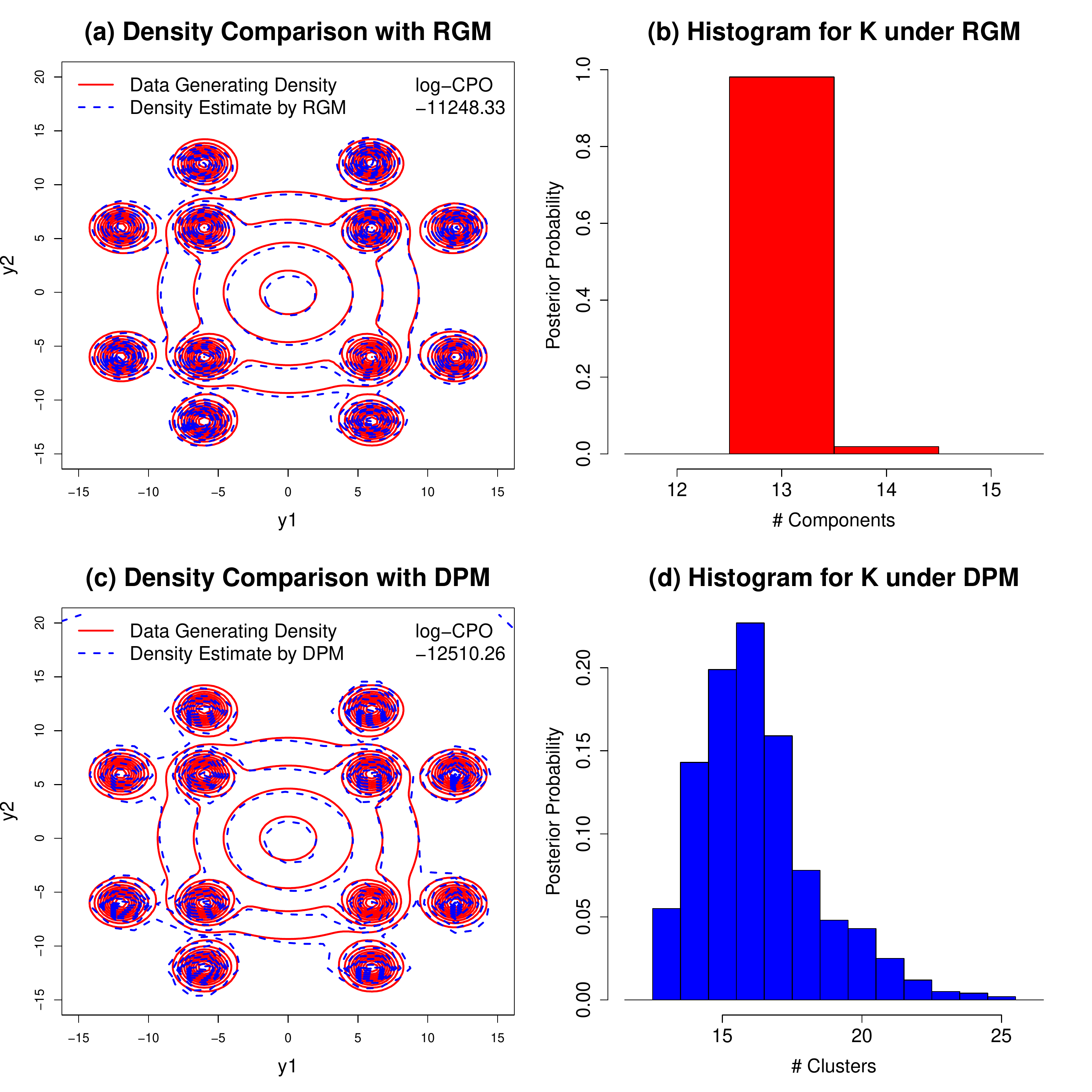}}
	\caption{Synthetic Example for Section \ref{sec:additional_simulation_study}. Panels (a) and (c) are the contour plots for the posterior density estimation under the RGM model and the DPM model, respectively. Panels (b) and (d) are the histograms of the posterior number of components under the RGM model and the posterior number of clusters under the DPM model, respectively, where the underlying true number of components is $K = 13$. 
	}
	\label{fig:GMM_K13_figures}
\end{figure} 
Figures \ref{fig:GMM_K13_figures}a and \ref{fig:GMM_K13_figures}c visualize the comparison between the posterior mean of the density under the RGM model and the DP mixture model with the data generating density, respectively, together with the corresponding log-CPO values. The log-CPO values indicate that the RGM model is a better choice compared to the DP mixture model. Furthermore, Figure \ref{fig:GMM_K13_figures}b indicates that the posterior distribution of $K$ is highly concentrated around the underlying true $K=13$ under the RGM model, whereas the DPM model assigns relatively higher posterior probability to redundant clusters (see Figure \ref{fig:GMM_K13_figures}d). 


\end{document}